\documentclass[acmsmall,nonacm]{acmart}

\setcopyright{none}
\renewcommand\footnotetextcopyrightpermission[1]{} %
\input{datalabmacros} %
\setcopyright{none}
\begin{document}

\title[A Unified Approach for Resilience with Integer Linear Programming (ILP) and LP Relaxations]
{A Unified Approach for Resilience and Causal Responsibility with Integer Linear Programming (ILP) and LP Relaxations}
\author{Neha Makhija}
\orcid{0000-0003-0221-6836}
\affiliation{%
    \orcidicon{0000-0003-0221-6836}
	Northeastern University\country{USA}}
\email{makhija.n@northeastern.edu}

\author{Wolfgang	Gatterbauer}
\orcid{0000-0002-9614-0504}
\affiliation{%
    \orcidicon{0000-0002-9614-0504}
	Northeastern University \country{USA}}
\email{w.gatterbauer@northeastern.edu}



\begin{abstract}
	
\emph{What is a minimal set of tuples to delete from a database in order to eliminate all query answers}?
This problem is called ``the resilience of a query'' and is one of the key algorithmic problems underlying various forms of reverse data management, such as view maintenance, deletion propagation and causal responsibility.
A long-open question is determining the conjunctive queries (CQs) for which resilience can be solved in $\PTIME$.
	
We shed new light on this problem by proposing a unified Integer Linear Programming (ILP) formulation.
It is unified in that it can solve both previously studied restrictions (e.g., self-join-free CQs under set semantics that allow a $\PTIME$ solution) and new cases (all CQs under set or bag semantics).
It is also unified in that all queries and all database instances are treated \emph{with the same approach},yet the algorithm is \emph{guaranteed to terminate in $\PTIME$} for all known $\PTIME$ cases.
In particular, we prove that for all known easy cases, the optimal solution to our ILP is identical to a simpler Linear Programming (LP) relaxation, which implies that standard ILP solvers return the optimal solution to the original ILP in $\PTIME$.

Our approach allows us to explore new variants and obtain new complexity results.
1) It works under bag semantics, for which we give the first dichotomy results in the problem space. 
2) We extend our approach to the related problem of causal responsibility and give a more fine-grained analysis of its complexity.
3) We recover easy instances for generally hard queries, including instances with read-once provenance and instances that become easy because of Functional Dependencies \emph{in the data}.
4) We solve an open conjecture about a unified hardness criterion from PODS 2020 and prove the hardness of several queries of previously unknown complexity.
5) Experiments confirm that our findings accurately predict the asymptotic running times, and that our universal ILP is at times even quicker than a previously proposed dedicated flow algorithm.
\end{abstract}

\maketitle
\setcounter{page}{1}

\section{Introduction}
\label{sec:introduction}

\emph{What is a minimum set of changes to a database in order to produce a certain change in the output of a query}?
This question underlies many problems of practical relevance, including explanations \cite{SudeepaSuciu14,glavic2021trends}, algorithmic fairness \cite{salimi2019interventional,galhotra2017fairness}, and diagnostics \cite{wang2017qfix,wang2015error}.
Arguably, the simplest formulation of such reverse data management~\cite{DBLP:journals/pvldb/MeliouGS11} questions is ``resilience'': 
\emph{What is the minimal number of tuples to delete from a database in order to eliminate all query answers}?\footnote{While the formal definition (which we give later) applies only to Boolean queries, the above more intuitive formulation can be easily transformed into the Boolean variant.}
An early variant of the problem was formulated 40 years ago in the context of view-maintenance~\cite{Dayal82} and has been studied over the years in various forms.
The problem has received considerable attention in the context of provenance and deletion propagation~\cite{Buneman:2002, DBLP:conf/icdt/BunemanKT01, Buneman07}.
\emph{Deletion propagation} seeks a set of tuples that can be
deleted from the database to delete a particular tuple from the view.
A variation we study in this paper is \emph{causal responsibility}, which involves finding a minimum subset of tuples to remove to make a given input tuple ``counterfactual.''~\cite{MeliouGHKMS10,DBLP:journals/pvldb/MeliouGMS11}.

\begin{table}
	\small
	\crefname{theorem}{Thm.}{Thms.}
	\centering
	\setlength{\tabcolsep}{1pt}
	\begin{tabular}{p{1.5cm}|p{3.5cm}|p{3cm}|p{3cm}|p{2.5cm}|}
		\multicolumn{5}{r}{\includegraphics[scale=0.45]{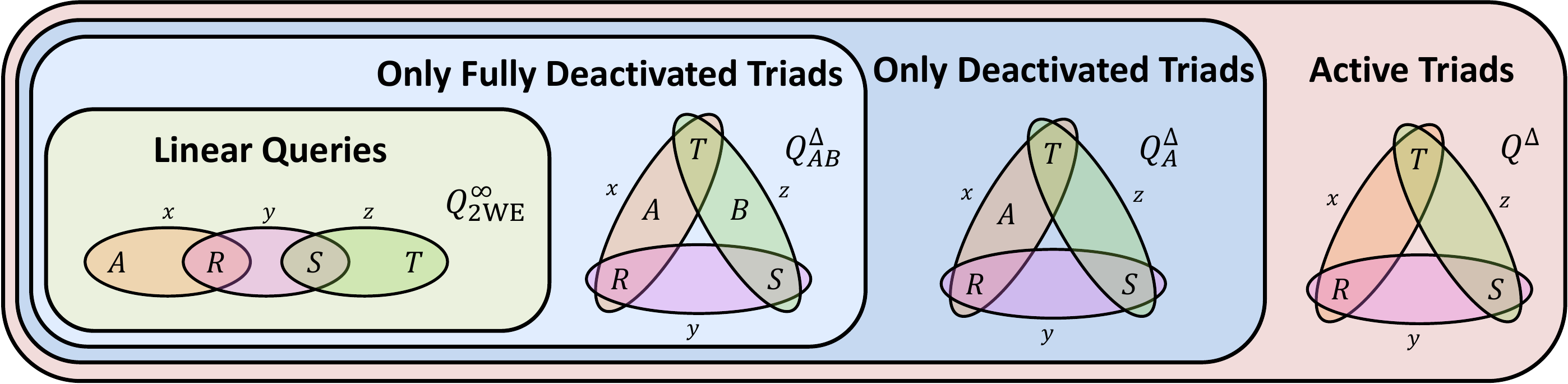}} \\
		\hhline{>{\arrayrulecolor{white}}->{\arrayrulecolor{black}}->{\arrayrulecolor{black}}->{\arrayrulecolor{black}}->{\arrayrulecolor{black}}->{\arrayrulecolor{black}}}
		$\res$ (Sets) 
		& \multicolumn{1}{c|}{$\PTIME$ (\cref{thm:reseasy})}
		& \multicolumn{2}{c|}{$\PTIME$ (\cref{thm:reseasydomination})}
		& \multicolumn{1}{c|}{$\NPC$ (\cref{th:IJPs_for_SJFCQs})} \\
		\hhline{>{\arrayrulecolor{white}}->{\arrayrulecolor{black}}->{\arrayrulecolor{black}}->{\arrayrulecolor{black}}->{\arrayrulecolor{black}}->{\arrayrulecolor{black}}}
		$\rsp$ (Sets) 
		& \multicolumn{1}{c|}{$\PTIME$ (\cref{thm:rspeasy})} 
		& \multicolumn{1}{c|}{$\PTIME$ (\cref{thm:rspeasyfullydominated})}  
		& \multicolumn{1}{p{2.6cm}|}{\cellcolor{yellow!40} \footnotesize $\PTIME$: dominating (A)  
		\normalsize
		(\cref{thm:rspeasydominatingtable})}  
		& \multicolumn{1}{c|}{ $\NPC$ (\cref{thm:rspharderthanres})} \\
		\cline{4-4}
		&&& \multicolumn{1}{p{2.6cm}|}{\footnotesize $\NPC$: non-dominating (R,S,T)  
		\normalsize
		(\cref{thm:rspharddominatedtable})} & \\
		\hhline{>{\arrayrulecolor{white}}->{\arrayrulecolor{black}}->{\arrayrulecolor{black}}->{\arrayrulecolor{black}}->{\arrayrulecolor{black}}->{\arrayrulecolor{black}}}
		$\res$ (Bags) 
		& \multicolumn{1}{c|}{\cellcolor{yellow!40} $\PTIME$ (\cref{thm:reseasy})}
		& \multicolumn{3}{c|}{\cellcolor{yellow!40} $\NPC$ (\cref{thm:reshardbag}) } \\
		\hhline{>{\arrayrulecolor{white}}->{\arrayrulecolor{black}}->{\arrayrulecolor{black}}->{\arrayrulecolor{black}}->{\arrayrulecolor{black}}->{\arrayrulecolor{black}}}
		$\rsp$ (Bags) 
		& \multicolumn{1}{c|}{\cellcolor{yellow!40} $\PTIME$ (\cref{thm:rspeasy})}
		& \multicolumn{3}{c|}{\cellcolor{yellow!40} $\NPC$ (\cref{thm:rspharderthanres})} \\
		\hhline{>{\arrayrulecolor{white}}->{\arrayrulecolor{black}}->{\arrayrulecolor{black}}->{\arrayrulecolor{black}}->{\arrayrulecolor{black}}->{\arrayrulecolor{black}}}
	\end{tabular}
	\caption{Overview of complexity results for self-join-free conjunctive queries (SJ-free CQs)
	that follow from our unified framework in this paper. 
	Results highlighted with yellow background are new.
	$\res$ stands for resilience and $\rsp$ for causal responsibility.
	Not shown 
	are additional results we give for queries with self-joins.
	}
	\label{Fig_query_venn}
\end{table}

	The problems of resilience and causal responsibility have practical applications in helping users better understand transformations of their data and to explain surprising query results. 
	They are both based on the idea of \emph{minimal interventions}, which aims to find the simplest possible satisfying explanations.
	Intuitively, the resilience of a query provides a minimal set of tuples (i.e.\ a minimal explanation) 
	without which a Boolean query would not return true.
	In addition, it is known that a solution to resilience immediately also provides an answer to the deletion propagation with source-side effects problem ~\cite{DBLP:journals/pvldb/FreireGIM15}, which seeks a minimal intervention, or a minimal set of input tuples to be deleted to perform deletion propagation (delete a tuple from the view).

	The problem of causal responsibility uses the same idea of minimal interventions to provide explanations at a more fine-grained tuple level.
	For any desired input tuple, users can calculate the ``responsibility'' of that tuple based on formal, mathematical notions of causality adapted to databases \cite{MeliouGHKMS10}.
    Then one can derive explanations by ranking input tuples using their responsibilities: tuples with a high degree of responsibility are better explanations for a particular query result.
	This makes causal responsibility an invaluable tool for query explanations and debugging ~\cite{glavic2021trends}.

Our goal is to understand the complexity of solving resilience and causal responsibility. The first result 
by Buneman et al.~\cite{DBLP:conf/icdt/BunemanKT01} showed that the problem is \np-complete ($\NPC$) for conjunctive queries (CQs) with projections.
Later work under the topic of causal responsibility~\cite{DBLP:journals/pvldb/MeliouGS11}
and the simpler notion of resilience~\cite{DBLP:journals/pvldb/FreireGIM15} 
showed that a large fraction of self-join-free CQs (triad-free queries)
can be solved in $\PTIME$, solving the complexity of self-join-free (SJ-free) queries. 
However, few results are known for the cases of CQs \emph{with self-joins}~\cite{DBLP:conf/pods/FreireGIM20}.
This state is similar to other database problems where establishing 
complexity results for self-joins is often considerably more involved than for self-join-free queries 
(e.g., compare the
results on probabilistic databases for either self-join-free queries~\cite{DBLP:journals/vldb/DalviS07}
with those for self-joins~\cite{DBLP:journals/jacm/DalviS12}).
Moreover, all these problems have been studied only for set semantics, whereas relational databases actually use bag semantics i.e., they allow duplicate tuples~\cite{chaudhuri1993optimization}.
Like self-joins, bags usually make problems harder to analyze
\cite{10.1145/3472391, atserias2022structure, yannakakis2022technical}, and few complexity results for bag semantics exist.

This paper gives the first dichotomy results under bag semantics for problems in reverse data management (\cref{Fig_query_venn}).
We also give a simple-to-verify sufficient hardness criterion for all conjunctive queries (including queries with self-joins and under set or bag semantics).
Based on this criterion, we build an automatic hardness certificate finder, that, 
a given query $Q$ and a fixed domain size $d$, finds a hardness certificate for $Q$ of domain size $\leq d$, whenever such a certificate exists.
We use this construction to find hardness certificates for 5 previously open queries with self-joins.

Our attack on the problem is unconventional:
Rather than deriving a dedicated $\PTIME$ algorithm for certain queries
(and proving hardness for the rest), 
we instead propose a unified Integer Linear Program (ILP) formulation for all problem variants (self-joins or not, sets or bags, Functional Dependencies or not).
We then show that, for all $\PTIME$ queries, \emph{the Linear Program (LP) relaxation of our ILP has the same optimal value}, thereby
proving that existing ILP solvers are \emph{guaranteed} to solve problems for those queries in $\PTIME$.

\introparagraph{Contributions and Outline}
We propose a unified framework for solving resilience and causal responsibility, 
give new theoretical results, approximation guarantees, and experimental results:

1) \emph{Unified ILP framework}:
We propose an ILP formulation for the problems of resilience and causal responsibility 
that can not only encode all previously studied variants of the problem, but can also encode \emph{all} formulations of the problem, 
including self-joins and bag semantics 
(\cref{sec:res-ilp,sec:resp-ilp}).
This unified encoding allows us to model and solve problems for which currently no algorithm (whether easy or hard) has been proposed.
It also allows us to study LP relaxations (\cref{sec:relaxations}) of our formulation, which form the basis of several of our theoretical results.

2) \emph{Unified hardness criterion}: 
We prove a variant of an open conjecture from PODS 2020~\cite{DBLP:conf/pods/FreireGIM20} 
by defining a structural certificate called Independent Join Path (IJP) and proving that it implies hardness (\cref{sec:IJP}).
Most interestingly, we give a Disjunctive Logic Program (DLP) formulation 
that can computationally derive such certificates.
We use this certificate to both ($i$) prove hardness for all hard queries in our dichotomies, and 
($ii$) obtain computationally derived hardness certificates for 5 previously open queries with self-joins.
While solving such programs is general in $\Sigma^2_p$ (i.e.\ on the 2nd level of the polynomial hierarchy) 
a modern ASP solver \emph{clingo}~\cite{gebser2011potassco} allowed us to obtain all the new, easy-to-verify proofs in under two hours, including some obtained in seconds.

3) \emph{First results for resilience and responsibility under bag semantics}:
We give full dichotomy results for both resilience and causal responsibility under bag semantics
for the special case of SJ-free CQs (\cref{sec:theoreticalresults}).
We show that under bag semantics, the $\PTIME$ cases for resilience and responsibility are exactly the same (\cref{Fig_query_venn}).

4) \emph{Recovering $\PTIME$ cases}:
We prove that for all prior known and newly found $\PTIME$ cases of SJ-free queries 
(under both set and bag semantics), 
our ILP is solved in guaranteed $\PTIME$ by standard solvers (\cref{sec:theoreticalresults}).
This means that our formulation is unified not only in being able to model all cases but also in that it is guaranteed to \emph{recover all known $\PTIME$ cases by terminating in $\PTIME$}.
In addition, we uncover more tractable cases for causal responsibility, due to obtaining more fine-grained complexity results  (\cref{sec:theory:responsibility}).
Our new way of modeling the problem opens up a new route for solving various open problems in reverse data management: 
by proposing a universal algorithm for solving all variants, 
future development does not depend on finding new dedicated $\PTIME$ algorithms, 
but rather on proving that the universal method terminates in $\PTIME$ 
(in similar spirit to proofs in this paper).

5) \emph{Novel approximations}:
We show 3 different approximation algorithms for both resilience and causal responsibility. 
The first approach based on LP-rounding
provides a guaranteed $m$-factor approximation (where $m$ is the number of atoms in the query) for all queries (\emph{including self-joins and bag semantics}).
The other two are new flow-based approximation techniques designed for hard queries without self-joins (\cref{sec:approximations}).
	    
6) \emph{Experimental Study}:
We compare all approaches proposed in this paper on different problem instances:
easy or hard, for set or bag semantics, queries with self-joins, and Functional Dependencies. 
Our results establish the accuracy of our asymptotic predictions, uncover novel practical trade-offs, 
and show that our approach and approximations create an end-to-end solution
(\cref{sec:experiments}).

We make all code and experiments available online ~\cite{resiliencecode}.
We provide a \emph{proof intuition} for each theorem in the main text, and full proofs are available in the appendix.
The appendix also contains additional examples and details, and discusses some additional results as well.
Our approach can solve resilience and causal responsibility for otherwise hard queries in $\PTIME$ for database instances such as \emph{read-once} instances, or instances that obey certain \emph{Functional Dependencies} (not necessarily known at the query level).
We show these \emph{instance-based tractability results} in~\cref{SEC:THEORY:READONCEANDFDS}.

\section{Related Work}

\introparagraph{Resilience and Causal Responsibility}
Foundational work by Halpern, Pearl, et al.~\cite{HalpernPearl:Cause2005, HalpernPearl:Explanations2005,ChocklerH04}
defined the concept of \emph{causal responsibility} based \emph{minimal interventions} in the input.
Meliou et al.~\cite{DBLP:journals/pvldb/MeliouGMS11} adapted this concept to define causal responsibility for database queries
and proposed a flow algorithm to solve the tractable cases.
Freire et al.~\cite{DBLP:journals/pvldb/FreireGIM15} defined a simpler notion of \emph{resilience} and
gave a dichotomy of the complexity for both resilience and responsibility for SJ-free queries under set semantics.
While the tractability frontier for self-join case remains open to this day, 
Freire et al.~\cite{DBLP:conf/pods/FreireGIM20} gave partial complexity results for resilience for queries with self-joins 
and conjectured that the notion of Independent Join Paths (IJPs) could imply hardness for resilience.
We prove one direction of this conjecture (with a slight fix of the original statement).
After acceptance of this paper, an interesting preprint was published on arXiv \cite{bodirsky2023complexity} 
that formulates resilience as a Valued Constraint Satisfaction problem (VCSP) and
applies results from an earlier VCSP dichotomy \cite{kolmogorov2017complexity}.
Interestingly, it also ends with a dichotomy conjecture (not proof) for resilience, notably for bag semantics but not set semantics.
We discuss these connections in more detail in~\cref{SEC:VCSPRESILIENCECONNECTION}
.

\introparagraph{Other Problems in View Maintenance}
There are several variants to resilience such as destroying a pre-specified fraction of witnesses from the database instead of all witnesses~\cite{ADP}. 
They all are instances of reverse data management~\cite{DBLP:journals/pvldb/MeliouGS11} and deletion propagation~\cite{Buneman:2002,Dayal82}.
\emph{Deletion propagation} seeks to delete a set of input tuples in order to delete a particular tuple from the view.
Intuitively, this deletion should be achieved with minimal side effects, where side effects are defined with one of two objectives: 
(a) deletion propagation with \emph{source side effects} seeks a minimum set of input tuples in order to delete a given 
output tuple; whereas (b) deletion propagation with \emph{view side effects} seeks a set of input tuples that results in a minimum number of output tuple deletions in the view, other than the tuple of interest~\cite{Buneman:2002}.
The dichotomies for self-join queries remain open for the problems in this space.
We believe that our core ideas can be applied to many such problems.

\introparagraph{Explanations and fairness}
Data management research has recognized the need to derive \emph{explanations} for query results and surprising observations ~\cite{glavic2021trends}.
Existing work on explanations use many approaches~\cite{DBLP:conf/chi/LimDA09}, including modifying the input (i.e.\ performing \emph{interventions})~\cite{DBLP:journals/pvldb/MeliouGMS11,DBLP:journals/pvldb/HuangCDN08,DBLP:journals/pvldb/HerschelHT09, DBLP:journals/corr/abs-0912-5340,SudeepaSuciu14,DBLP:journals/pvldb/0002M13}, which is our focus as well.
Recent approaches show that explanations benefit a variety of applications, such as ensuring or testing fairness ~\cite{pradhan2021interpretable,salimi2019interventional,galhotra2017fairness} or finding bias ~\cite{youngmann2022explaining}.
We believe our unified framework of solving both easy and hard cases with one algorithm can also be useful for these applications.

\introparagraph{Bag semantics}
Real-world databases consist of bags instead of sets.
This gap between database theory and database practice  
has been pointed out years ago ~\cite{chaudhuri1993optimization}. 
However, studying properties of CQs under bag semantics is often considerably harder.
For example, the connection between local and global consistency has only been recently solved for bags ~\cite{atserias2022structure, yannakakis2022technical}, 
and the fundamental problems of query containment of CQs under bag semantics remain open 
despite recent progress ~\cite{10.1145/3472391, DBLP:conf/pods/KonstantinidisM19}. 
Our paper gives the first dichotomy result for reverse data management problems under bag semantics.

\introparagraph{Linear Optimization and Data Management}
Ideas from the two fields have been connected in the past, both to solve data management problems efficiently~\cite{meliou2012tiresias,brucato2019scalable}, 
and to use the factorized nature of data to solve linear optimization problems more efficiently~\cite{capelli2022linear}.
The Tiresias system ~\cite{meliou2012tiresias} implements how-to queries by translating them to MILPs in order to solve them efficiently. 
Package queries~\cite{brucato2019scalable} allow users to define constraints over multiple tuples with extensions of SQL, and also leverage ILP solvers in the background.
Recent work by Capelli at al
~\cite{capelli2022linear} provides an approach to solve a specific class of linear programs (LP(CQ)), whose variables correspond to answers of a CQ. 
They show that such LPs have $\PTIME$ query complexity for CQs with bounded fractional hypertreewidth, by leveraging the factorized structure of the data. 
Our work similarly leverages the structure of data, 
but focuses on \emph{data} complexity of \emph{Integer} 
Linear Programs to investigate the tractability of reverse data management problems and solve them efficiently when possible.

\section{Preliminaries}
\label{sec:preliminaries}

\subsection{Formal Problem Setup}

\introparagraph{Standard database notations}
A \emph{conjunctive query} (CQ) is a first-order formula $Q(\vec y)$ $= \exists
\vec x\,(g_1 \wedge \ldots \wedge g_m)$ 
where the variables $\vec x = (x_1, \ldots, x_\ell)$ are called existential variables,
$\vec y$ are called the head or free variables,
and each atom 
$g_i$ represents a relation 
$g_i= R_{j_i}(\vec x_i)$ where $\vec x_i \subseteq \vec x \cup \vec y$.\footnote{WLOG, we assume that
$\vec x_i$ is a tuple of only variables and don't write the constants.
Selections can always be directly pushed into the database before executing the query.
In other words, for any constant in
the query, we can first apply a selection on each relation and then consider the modified query with
a column removed.}
$\var(X)$ denotes the variables in a given relation/atom.
Notice that a query has at least one output tuple iff the Boolean variant of the query 
(obtained by making all the free variables existential) is true.
Unless otherwise stated, a query in this paper denotes a Boolean CQ, i.e.\ $\vec y = \emptyset$.
We write $Q $ to denote that that query $D \models Q$ to denote that query $Q$ evaluates to $\true$ over database instance $D$, and $D \not\models Q$ to denote it evaluates to $\false$.

Queries are interpreted as hypergraphs with edges formed by atoms and nodes by variables.
Two hyperedges are connected if they share at least one node.
We use concepts like paths and reachable nodes on the hypergraph of a query in the usual sense \cite{bollobas1998modern}.
A query $Q$ is minimal
if for every other equivalent conjunctive query $Q'$ has at least as many atoms as $Q$ ~\cite{DBLP:conf/pods/FreireGIM20}. 
WLOG we discuss only connected queries in the rest of the paper.\footnote{Results for disconnected queries follow by treating each of the components \emph{independently}.}
A \emph{self-join-free CQ} (SJ-free CQ) is one where no relation symbol occurs more than once and thus every atom represents a different relation. 

We write $D$ for the database, i.e.\ the set of tuples in the relations.
When we refer to bag semantics, we allow $D$ to be a multiset of tuples in the relations.
We write $[\vec w / \vec x]$ as a valuation (or substitution) of query variables $\vec x$ by $\vec w$.
A \emph{witness} $\vec w$ is a valuation of $\vec x$ that is permitted by $D$ and that makes $Q$ \true
(i.e.\ $D \models Q[\vec w/\vec x]$).\footnote{Note that our notion of witness slightly differs from the one used in  provenance literature where a ``witness'' refers to a subset of the input database records that is sufficient to ensure that a given output tuple appears in the result of a query \cite{DBLP:journals/ftdb/CheneyCT09}.} 
The set of witnesses is then
\[
	\witnesses(Q,D) = \bigset{\vec w}{D \models Q[\vec w /\vec x]}\; .
\]

Since every witness implies exactly one set of up to $m$ tuples from $D$ 
that make the query true, 
we will slightly abuse the notation and also refer to this set of tuples as ``witnesses.'' 
For example, consider the 2-chain query 
$Q^\infty_2 \datarule R(x, y), S(y, z)$ 
over the database 
$D = \{r_{12}{:\,}R(1,2), s_{23}{:\,}S(2,3), s_{24}{:\,}S(2,4)\}$.
Then 
the $\witnesses(Q^\infty_2, D) =$ $\{(1, 2, 3), (1, 2, 4)\}$
and their respective tuples (also henceforth referred to as witnesses) are 
$\set{r_{12},s_{23}}$, and $\set{r_{12},s_{24}}$.
A set of witnesses may be represented as a connected hypergraph, where tuples are the nodes of the graph and each witness as a hyperedge around a set of tuples. 
\introparagraph{Resilience, Responsibility, and related terminology}
\begin{definition}[Resilience \cite{DBLP:journals/pvldb/FreireGIM15}]\label{def: resilience}
Given a query $Q$ and database $D$, we say that $k \in \resdecision(Q, D)$ if and only if $D \models Q$ and there exists some contingency set $\Gamma \subseteq D$ with $|\Gamma| \leq k$ such that $D - \Gamma \not\models Q$.
\end{definition}

In other words, $k\in\resdecision(Q, D)$ means that there is a set of $k$ or fewer tuples in
$D$, the removal of which makes the query false. 
We are interested in the optimization version $\res^*(Q,D)$ of this decision problem:
given $Q$ and $D$, find the \emph{minimum} $k$ so that $k\in\res(Q,D)$. 
A larger $k$ implies that the query is more ``\emph{resilient}'' and requires the deletion of more tuples to change the query output. 
A contingency size of minimum size is called a \emph{resilience set}.

\begin{definition}[Responsibility \cite{DBLP:journals/pvldb/MeliouGMS11}]\label{def: responsibility}
    Given query $Q$ and an input tuple $\resptuple$,
    we say that $k \in \rspdecision(Q, D, \resptuple)$ if and only if $D \models
    Q$ and there is a contingency set $\Gamma \subseteq D$ 
	with $|\Gamma| \leq k$
	such that $D - \Gamma \models
    Q$ but $D - (\Gamma \cup \{ \resptuple \}) \not\models Q$.
\end{definition}

In other words, 
causal responsibility aims to determine whether \emph{a particular input tuple} $\resptuple$ (the \emph{responsibility tuple})
can be made ``counterfactual'' by deleting a set of other input tuples $\Gamma$ of size $k$ or less. 
Counterfactual here means that the query 
is
true with that input tuple present, but false if it 
is
also deleted.
In contrast to resilience, the problem of responsibility is defined for \emph{a
particular tuple} $\resptuple$ in $D$, and instead of finding a $\Gamma$ that will leave no witnesses 
for $D - \Gamma \models q$, we want to preserve only witnesses that involve $\resptuple$, so that
there is no witness left for $D - (\Gamma \cup \{ \resptuple \}) \models Q$.
Responsibility measures the \emph{degree} of causal contribution of a particular tuple $\resptuple$ to the output of a query as a function of the size of a minimum contingency set (the \emph{responsibility set}).
We are again interested in the optimization version of this problem: $\rsp^*(Q,D,\resptuple)$.\footnote{Note that it is possible that a given tuple cannot be made counterfactual. For example, given witnesses $\{\{r_{11}\}, \{r_{11},r_{12}\} \}$,
tuple $r_{12}$ cannot be made counterfactual without deleting $r_{11}$, which in turn would delete both witnesses.}

\begin{definition}[Exogenous / Endogenous tuples]\label{def:exogenoustuple}
	A tuple is exogenous if it must not or need not participate in 
	a contingency set, 
	and endogenous otherwise.
\end{definition}

Prior work~\cite{DBLP:journals/pvldb/MeliouGMS11} has defined relations (or atoms) to be exogenous or endogenous, i.e.\ when all tuples in any relation (or relation of the atom) are either exogenous or endogenous.
We use but also generalize this notation to allow \emph{individual tuples} to be declared exogenous (but keep them endogenous by default).
We will see later in \cref{sec:IJP} that this generalization allows us to formulate resilience and responsibility 
with a simple universal hardness criterion.\footnote{
In more detail, we will formulate hardness of responsibility via an Independent Join Path which is only possible because one specified tuple is exogenous, e.g.\ \cref{thm:rspharddominatedtable}.}
The set of exogenous tuples $\exoset \subset D$ 
can be provided as an additional input parameter as in $\res(Q, D, \exoset)$ and $\rsp(Q, D, \resptuple, \exoset)$.
We assume a database instance has no exogenous tuples unless explicitly specified, and we omit the parameter for simplicity.

\introparagraph{Our focus}
We are interested in the \emph{data complexity}~\cite{DBLP:conf/stoc/Vardi82} 
of $\res(Q,D)$ and $\rsp(Q,D, \resptuple)$, i.e.\ 
the complexity of the problem as $D$ increases but $Q$ remains fixed.
We refer to $\res(Q)$ and $\rsp(Q)$ to discuss the complexity of the problems of query $Q$ over an arbitrary data instance (and arbitrary responsibility tuple). 

\subsection{Tools and Techniques}

We use \emph{Integer Linear Programs} and their \emph{relaxations} to model and solve resilience and causal responsibility.
\emph{Disjunctive Logic Programs}, 
which can solve problems higher in the polynomial hierarchy, are used to find certificates for hard cases.

\introparagraph{Linear Programs (LP)}
Linear Programs are standard optimization problems \cite{aardal2005handbooks, schrijver1998theory}
in which the objective function and the constraints are linear.
A standard form of an LP is $\min \vec c^\transpose \vec x$ s.t. $\vec W \vec x \geq \vec b$, where $\vec x$ 
denotes the variables, the vector $\vec c^\transpose$ denotes weights of the variables in the objective, the matrix $\vec W$ denotes the weights of $\vec x$ for each constraint, and $\vec b$ denotes the right-hand side of each constraint.
If the variables are constrained to be integers, the resulting program is called an Integer Linear Program (ILP), while a program with some integral variables is referred to as a Mixed Integer Linear Program (MILP).
The \emph{LP relaxation} of an ILP program is obtained by removing the integrality constraint for all variables.

\introparagraph{Complexity of solving ILPs}
ILPs are $\NPC$ and 
part of Karp's $21$ problems \cite{karp1972reducibility}, while LPs can be solved in $\PTIME$ with Interior Point methods \cite{grotschel1993ellipsoid,cohen2021solving}.
The complexity of MILPs is exponential in the number of integer variables. 
However, there are conditions under which ILPs become tractable.
In particular, if 
there is an optimal integral assignment to the LP relaxation, then the original ILP can be solved in $\PTIME$ as well.
A lot of work studies conditions when this property holds~\cite{ford1956maximal, schrijver1998theory, cornuejols2002ideal, lau2011iterative}.
A famous example is the max-flow min-cut problem
which can be solved with LPs despite integrality constraints.
The \emph{max-flow Integrality Theorem} ~\cite{ford1956maximal}
states that for every flow graph with all capacities as integer values, there is an optimal maximum flow such that all flow values are integral.
Therefore, in order to find an integral max-flow for such a graph, one need not solve an ILP but rather an LP relaxation suffices to get the same optimal value.
There are many other structural characteristics that define when the LP is guaranteed to have an integral minimum, and thus where ILPs are in $\PTIME$. 
For example, if the constraint matrix of an ILP is \emph{Totally Unimodular} \cite{schrijver1998theory} then the LP always has the same optima. 
Similarly, if the constraint matrix is \emph{Balanced}~\cite{conforti2006balanced}, several classes of ILPs are $\PTIME$.

We use the results of Balanced Matrices to show that the resilience and responsibility of any read-once data instances can be found in $\PTIME$ (as an additional result in~\cref{SEC:THEORY:READONCEANDFDS}).
For other $\PTIME$ cases, we have ILP constraint matrices that do not fit into any previous tractability characterization. 
Despite this, we are able to use these results indirectly (via an intermediate flow representation) to show that the \emph{LP relaxation} has the same objective as the original ILP and thus the ILP can be solved in $\PTIME$.

\introparagraph{Linear Optimization Solvers}
A key advantage of modeling problems as ILPs is practical.
There are many highly-optimized ILP solvers, both commercial~\cite{gurobi} and free \cite{mitchell2011pulp}
which can obtain exact results fast in practice.
ILP formulations are standardized, and thus programs can easily be swapped between solvers.
Any advances made over time by these solvers (improvements in the presolve phase, heuristics, and even novel techniques) can automatically make implementations of these problems better over time.

For our experimental evaluation we use Gurobi.\footnote{Gurobi offers a free academic license \url{https://www.gurobi.com/academia/academic-program-and-licenses/}.}
Gurobi uses an LP based branch-and-bound method to solve ILPs and MILPs~\cite{gurobi_working}. 
This means that it first computes an LP relaxation bound and then explores the search space to find integral solutions that move closer to this bound.
If an integral solution is encountered that is equal to the LP relaxation optimum, then the solver 
has found a guaranteed optimal solution and is done.
In other words, if we can prove that the LP relaxation of our given ILP formulation has an integral optimal solution, then we are guaranteed that our original ILP formulation will terminate in $\PTIME$ even without changing the formulation or letting the solver know anything about the theoretical complexity.

\introparagraph{Disjunctive Logic Programs (DLPs)}
Disjunctive Logic Programs 
are Logic Programs that allow disjunction in the head of a rule~\cite{10.1007/BF03037171, 10.1145/502807.502810}.
DLPs have been shown to be $\Sigma^2_p$-complete \cite{eiter1995computational,10.1145/261124.261126}, and are more expressive than Logic Programs without disjunctions that are $\NPC$.
The key to higher expressivity is the non-obvious \emph{saturation} technique that can check if \emph{all} possible assignments satisfy a given property~\cite{EITER1993231}.
Logic Programs have been used for database repairs~\cite{gelfond2014knowledge} 
and to determine the responsibility of tuples in a database~\cite{bertossi2021specifying}. 
We go beyond this to build a DLP that searches for a certificate that proves that solving the resilience/responsibility problem is $\NPC$ for a given query.
We represent our DLP as an Answer Set Program (ASP)~\cite{eiter2009answer} and use \emph{clingo} \cite{clingo}
to solve it.

\section{ILP for Resilience}
\label{sec:res-ilp}
\label{SEC:RES-ILP}

We construct an Integer Linear Program $\resilpparam{Q,D}$ 
from a CQ $Q$ and a database 
$D$
which returns the solution to the optimization problem 
$\resdecision^*(Q,D)$ for any Boolean CQ (even with self-joins) under either set or bag semantics.\footnote{Notice that we also write $\textsf{ILP}[\textsf{problem}]$ for the optimal value of the program}
This section focuses on the correctness of the ILP.
\Cref{sec:relaxations} later investigates how easy cases 
can be solved in $\PTIME$, despite the problem being $\NPC$ in general.

To construct the ILP, we need to specify the decision variables, constraints and objective.
As input to the ILP, we first run the query on the database instance to compute all the witnesses. 
This can be achieved with a modified witness query, a query that returns keys for each table, and thus each returned row is a \emph{set} of tuples from each of the tables.\footnote{Duplicate tuples have the same key.}

\introparagraph{1. Decision Variables} 
We create an indicator variable $X[t] \in \{0, 1\}$ for each tuple $t$ in the database instance $D$.
A value of $1$ for $X[t]$ means that $t$ is included in a contingency set, and $0$ otherwise.
For bag semantics, \cref{thm:res-bag} shows that it suffices to define a single variable for a set of duplicate tuples 
(intuitively, an optimal solution chooses either all or none).

\introparagraph{2. Constraints} Each witness must be destroyed in order to make the output false
for a Boolean query (or equivalently, to eliminate all output tuples from a non-Boolean query).
A witness is destroyed, when at least one of its tuples is removed from the input. 
Thus, for each witness, we add one constraint enforcing that at least one of its tuples must be removed.
For example, for a witness $\w = \{ r_i, r_j, r_k \}$ 
we add the constraint that $X[r_i] + X[r_j] + X[r_k] \geq 1$.\footnote{Notice that for SJ-free queries, the number of tuples in each constraint is exactly equal to the number of atoms in the query. 
But for  queries with self-joins, the number of tuples in each constraint is not fixed (is lower when a tuple joins with itself).}

\introparagraph{3. Objective} Under set semantics, we simply want to minimize the number of tuples deleted. 
Since for bag semantics we have made a simplification that we use only one variable per ``\emph{unique tuple},'' marking that tuple as deleted has cost equal to deleting all copies of the tuple.
Thus, we weigh each tuple by the number of times it occurs to create the minimization objective.

\begin{example}[$\res$ ILP]
    Consider the Boolean two-chain query with self-join $\qsjtwochain \datarule R(x,y), R(y,z)$
	and a database $D$ with a single table $R$
	$\{(1,1), (2,3) (3,4) \}$
	The query over $D$ has 2 witnesses:
	\begin{center}
        \begin{tabular}{|c|c|c|l}
            \cline{1-3}
            x & y & z  \\
            \cline{1-3}
            1 & 1 & 1 & $\vec w_1= \{r_{11}\}$ \\
            2 & 3 & 4 & $\vec w_2= \{r_{23}, r_{34}\}$ \\
            \cline{1-3}
        \end{tabular}
    \end{center}

    Each tuple has a decision variable.
    Thus, our ILP has 3 variables $X[r_{11}]$, $X[r_{23}]$, and $X[r_{34}]$.
    We create a constraint for each unique witness in the output, resulting in two constraints:
	\begin{align*}
    X[r_{11}] \geq 1 \qquad
    X[r_{23}] + X[r_{34}] \geq 1
	\end{align*}

    Finally, the objective is to minimize the tuples deleted, thus, to minimize: $X[r_{11}] + X[r_{23}] + X[r_{34}]$.
    Solving this results in an objective of 2 at $X[r_{11}] = 1$, $X[r_{23}] = 1$, $X[r_{34}] = 0$.
    Intuitively, one can see that $\res(Q,D)=2$ as removing $r_{11}$ and $r_{23}$ from $R$ is the smallest change required to make the query false.
    \label{example:res}
\end{example}

\begin{example}[$\res$ ILP: Bag Semantics]
	Assume the same problem as \cref{example:res}, but we allow duplicates in the input. 
    Concretely assume $r_{23}$ appears twice:
	$R' = \{(1,1):1, (2,3):\textcolor{red}{2}, (3,4):1 \}$. 
	The variables and constraints stay the same, only the objective function changes now to 
	$$
		\min \big\{ X[r_{11}] + \textcolor{red}{2} X[r_{23}] + X[r_{34}] \big\}
	$$
    Removing $r_{11}$ and $r_{23}$ is no longer optimal since it incurs a cost of $3$.
    The optimal solution is now at $X[r_{11}] = 1$, $X[r_{23}] = 0$, $X[r_{34}] = 1$, with the objective value $2$.

\end{example}

Before we prove the correctness of $\resilpparam{Q, D}$ in \cref{thm:res-correctness}, we will justify our decision to use a single decision variable per unique tuple with the help of \cref{thm:res-bag}.

\begin{restatable}{lemma}{thmresilpbag}
    \label{thm:res-bag}
    There exists a resilience set where for each unique tuple in D, either all occurrences of the tuple are in the resilience set, or none are. 
\end{restatable}

\begin{proofintuition*}[\cref{thm:res-bag}]
We show that if a tuple $t$ is in a contingency set $\Gamma$ but a duplicate tuple $t'$ is not, 
then removing $t$ leads to a now smaller contingency set $\Gamma'$.
This is due to the fact that since $t$ and $t'$ they are identical, they form witnesses with the same set of tuples.
If $t'$ is not in the contingency set, there must be another tuple in the contingency set for every witness of $t'$.
This implies that all the witnesses $t$ participates in are already covered, and $t$ need not be in the contingency set.
\qed
\end{proofintuition*}

\begin{restatable}{theorem}{thmresilpcorrect}[$\res$ \textsf{ILP} correctness]
    \label{thm:res-correctness}
	$\resilpparam{Q, D} = \resdecision^*(Q,D)$ for any CQ $Q$ and database $D$ under set or bag semantics.
\end{restatable}

\begin{proofintuition*}
    We prove validity by showing that any satisfying solution would necessarily destroy all witnesses i.e.\ make the query false. 
    Thus if we consider any invalid solution i.e.\ one in which not all witnesses have been destroyed, we can see that there is an unsatisfied constraint in $\resilp$. 
    Hence all $\resilp$ are valid.
    Next we prove optimality by showing that any valid resilience set would be a valid solution for the ILP. 
    This is equivalent to showing that any valid contingency set is a solution to $\resilp$, since they must satisfy all constraints.
    Since $\resilp$ always gives a valid, optimal solution, it is correct.
    \qed
\end{proofintuition*}

We would like to stress to the reader that changing from sets to bags affects only the objective function, not the constraint matrix.
Later in \cref{sec:theoreticalresults}, we will prove that for queries such as $\qtriangleunary$, the problem of finding resilience becomes $\NPC$ under bag semantics, while it is solvable in $\PTIME$ under set semantics.
This observation is significant because most literature on tractable cases in ILP focuses exclusively on analyzing the constraint matrix. For example, if an ILP has a constraint matrix that is Totally Unimodular it is $\PTIME$ no matter the objective function \cite[Section 19]{schrijver2003combinatorial}.

\section{ ILP for Responsibility}
\label{sec:resp-ilp}
\label{SEC:RESP-ILP}

The ILP for $\rsp$ builds upon $\resilp$ with an important additional consideration.
While the goal of $\resilp$ was to destroy \emph{all} output witnesses, in $\rspilpparam{Q, D, \resptuple}$ we must 
also ensure that not all the output is destroyed.
To enforce this, we need additional constraints and additional decision variables to track the witnesses that are destroyed.

\introparagraph{1. Decision Variables} 
$\rspilpparam{Q, D, \resptuple}$ has two types of decision variables:
\begin{enumerate}[label=(\alph*)]
    \item $X[t]$: Tuple indicator variables are defined for all tuples in the set of witnesses we wish to destroy.
    \item $X[\w]$: Witness indicator variables help preserve at least $1$ witness that contains $\resptuple$.
    We track all witnesses that contain $\resptuple$ and set $X[\vec w]=1$ if the witness is destroyed and  $X[\vec w]=0$ otherwise.
\end{enumerate}

\introparagraph{2. Constraints} 
We deal with three types of constraints.
\begin{enumerate}[label=(\alph*)]
\item Resilience Constraints: Every witness that does not contain $\resptuple$ must be destroyed.
    As before, for such witnesses $\vec w_i= (r_i, r_j \hdots r_k)$ we enforce $X[r_i] + X[r_j] +\hdots + X[r_k] \geq 1$

\item Witness Tracking Constraints: For those witnesses that contain $\resptuple$, we need to track if the witness is destroyed.
    If any tuple that participates in a witness is deleted, then the witness is deleted as well. 
    Thus, we can enforce that $X[\vec w] \geq X[t]$ where $t \in \vec w$. 
    Notice that we just care about tuples that need to be potentially deleted, 
	i.e.\ only tuples that occur in witnesses without $\resptuple$. 

\item Counterfactual Constraint: A single constraint ensures that at least one of the witnesses that contains the responsibility tuple is preserved.
As example, if only the witnesses $\vec w_1, \vec w_2, \vec w_3$ contain $\resptuple$, 
then this constraint is $X[\vec w_1]+X[\vec w_2]+X[\vec w_3] \leq 2$.
\end{enumerate}

\introparagraph{3. Objective} 
The objective is the same as for $\resilpparam{Q, D}$: we minimize the number of tuples deleted (weighted by the number of occurrences).

\begin{restatable}{theorem}{thmrespilpcorrect}
    \label{thm:resp-correctness}
    $\rspilpparam{Q,D,\resptuple} =$ 
    $\rsp^*(Q, D, \resptuple)$ of a tuple $\resptuple$ in database instance $D$ under CQ $Q$ under set or bag semantics.
\end{restatable}

\begin{proofintuition*}[\cref{thm:resp-correctness}]
    Like \cref{thm:res-correctness}, we prove validity and then optimality. 
    We show that for any responsibility set we can assign values to the ILP variables such that they can form a satisfying solution (this follows from that fact that the responsibility set must preserve at least one witness containing $\resptuple$).
    Thus the correct solution is captured
	by $\rspilp$, while any invalid contingency set violates at least one constraint.
    \qed
\end{proofintuition*}

\begin{example}
    Consider $\qtwochain \datarule R(x,y), S(y, z)$ and database instance $D$ with $R = {(1,1)}$,
    $S = \{(1,1), (1,2), (1,3)\}$. 

    \begin{center}
    \begin{tabular}{|c|c|c|l}
        \cline{1-3}
        x & y & z & \\
        \cline{1-3}
        1 & 1 & 1 & $\vec w_1 =$ $\{r_{11},$ $s_{11}\}$ \\
        1 & 1 & 2 & $\vec w_2 =$ $\{r_{11},$ $s_{12}\}$ \\
        1 & 1 & 3 & $\vec w_3 =$ $\{r_{11},$ $s_{13}\}$ \\
        \cline{1-3}
    \end{tabular}
    \end{center}
	
    How do we calculate the responsibility of $s_{11}$? 
    First, we must destroy the two witnesses that do not contain $s_{11}$ i.e.\ $\vec w_2$ and $\vec w_3$.
    The tuple indicator variables we need are - $X[r_{11}]$, $X[s_{12}]$, $X[s_{13}]$.
    (Notice that $s_{11}$ is not tracked itself.)
    Since we need to track $\vec w_1$ to ensure it isn't destroyed, we need the witness indicator variable $X[\vec w_1]$.
    The resilience constraints are:
    \begin{align*}
    X[r_{11}] + X[s_{12}] \geq 1\\
    X[r_{11}] + X[s_{13}] \geq 1
    \end{align*}

    The witness tracking constraints apply only to $X[\vec w_1]$:
    \begin{align*}
    X[\vec w_{1}] \geq X[r_{11}]
    \end{align*}

    Finally, we use the counterfactual constraint to enforce that at least one witness is preserved.
    In this example, this implies directly that $\vec w_1$ may not be destroyed.
    \begin{align*}
    X[\vec w_{1}] \leq 0
    \end{align*}

    Solving this ILP gives us an objective of $2$ when $X[s_{12}] = 1$ and $X[s_{13}] = 1$ and all other variables are set to 0.
    Notice that setting $X[r_{11}]$  to $1$ will force $X[\vec w_{1}]$ to take value $1$ and hence violate the counterfactual constraint. 
    Intuitively, $r_{11}$ cannot be in the responsibility set because deleting it will delete all output witnesses, and not allow $s_{11}$ to be counterfactual.
    \label{example:rsp}
\end{example}

\section{LP Relaxations of $\resilp$ \& $\rspilp$}
\label{sec:relaxations}
\label{SEC:RELAXATIONS}

The previous sections introduced unified ILPs to solve for $\res$ and $\rsp$.
However, ILPs are $\NPC$ in general, and we would like stronger runtime guarantees for cases where $\res$ and $\rsp$ 
can be solved in $\PTIME$. 
We do this with the introduction of LP relaxations,
which generally act as lower bounds for minimization problems. 
However, in \cref{sec:theoreticalresults} we prove that these relaxations $\reslp$ and $\rspmilp$ are actually always equal to the corresponding ILPs 
for all easy SJ-free queries.
Thus, whether easy or hard, exact or approximate, problems can be solved within the same framework, with the same solver, with minimal modification, and with the best-achievable time guarantees.

\subsection{LP Relaxation for $\res$}

LP Relaxations are constructed by relaxing (removing) integrality constraints on variables. 
In $\resilp$, a tuple indicator variable $X[t]$ only takes values $0$ or $1$.
$\reslp$ removes that constraint and allows the variables any (``fractional'') value in 
$[0,1]$.

\subsection{MILP Relaxation for $\rsp$}

For responsibility, the relaxation is more intricate. 
It turns out that an LP relaxation is not optimal for $\PTIME$ cases (\cref{example:resprelaxation}).
We introduce a Mixed Integer Linear Program $\rspmilp$, where tuple indicator variables are relaxed and take values in $[0,1]$ 
whereas witness indicator variables are restricted to values $\{0,1\}$. 
Typically, MILPs are exponential in the number of integer variables i.e.\ if there are $n$ integer binary variables, a solver explores $2^n$ possible branches of assignments.
However, despite having an integer variable for every witness that contains $\resptuple$ (thus up to linear in the size of the database), we show that $\rspmilp$ is in $\PTIME$.

\begin{restatable}{lemma}{thmrspmilpptime}
    \label{thm:rspmilpptime}
    For any CQ $Q$ and tuple $\resptuple$, $\rspmilpparam{Q,D,\resptuple}$ can be solved in $\PTIME$ in the size of database $D$.
\end{restatable}

\begin{proofintuition*}
    We show that is possible to solve $\rspmilp$ in $\PTIME$ by solving a linear number of linear programs.
    Instead of looking at all possible 0-1 assigments to witness indicator variables - we simply need to select $1$ witness indicator variable that is to be set to $0$.   
    All witness indicator variables are combined into one counterfactual constraint. 
    This constraint is always satisfied when any one of the variable takes value $0$, irrespective of other variable values.
    Thus, we only need to explore the assignments where exactly $1$ variable takes on value $0$, thus a linear number of assignments in the size of the database.
    \qed
\end{proofintuition*}
In addition to the above theoretical proof of the $\PTIME$ solvability of $\rspmilp$, we see experimentally in \cref{sec:experiments} that a typical ILP solver indeed scales in polynomial time to solve $\rspmilp$.

\begin{example}
    \label{example:resprelaxation}
    Consider again the problem in \cref{example:rsp}. 
    The solution of $\rspilp$ was $2$ at  $X[s_{12}]\!=\!1$,  $X[s_{13}]\!=\!1$  $X[r_{11}]\!=\!0$ and $X[\vec w_{1}]\!=\!0$.
    What happens if we relax the integrality constraints and allow $0\!\leq\!X[v]\!\leq\!1$ for all variables?
    We can get a smaller satisfying solution $1.5$ at the point $X[s_{12}]\!=\!0.5$,  $X[s_{13}]\!=\!0.5$  $X[r_{11}]\!=\!0.5$ and $X[\vec w_{1}]\!=\!0.5$. 
    This value is $\rsplp$ and is \emph{not} guaranteed to be equal to $\rspilp$.
	If we instead create $\rspmilp$ and apply \emph{integrality constraints only for the witness indicator variables}, 
    then $X[\vec w_1]$ is forced to be in $\{0,1\}$ while all other variables can be fractional.
    We see that the $\rsplp$ solution is no longer permitted, and solving $\rspmilp$ results in the true $\rsp$ value of 2.
    We show in \cref{sec:theory:responsibility} that $\rspmilp$ = $\rspilp$ for all easy cases like chain queries such as $\qtwochain$ (\cref{Fig_query_venn}).
\end{example}

We conjecture that these relaxations are all we need to solve the problems of resilience and causal responsibility efficiently, whenever an efficient solution is possible. 
In \cref{sec:theoreticalresults}, we prove that \cref{conj:reseasy,conj:rspeasy} are true for all self-join free queries.

\begin{conjecture}[$\res$ is easy $\Rightarrow$ LP=ILP]
	\label{conj:reseasy}
	If $\res(Q)$ can be solved in $\PTIME$ under set/bag semantics, then $\reslpparam{Q,D} = \resilpparam{Q,D}$ for any database $D$ under the same semantics.
\end{conjecture}

\begin{conjecture}[$\rsp$ is easy $\Rightarrow$ MILP=ILP]
	\label{conj:rspeasy}
	If $\rsp(Q)$ can be solved in $\PTIME$ under set/bag semantics, then $\rspmilpparam{Q,D} = \rspilpparam{Q,D}$ for any database $D$ under the same semantics.
\end{conjecture}

\section{Finding hardness certificates}
\label{sec:IJP}
\label{SEC:IJP}

Freire et al.~\cite{DBLP:conf/pods/FreireGIM20} conjectured that the ability to construct a particular certificate 
called ``Independent Join Path'' is a sufficient criterion to prove hardness of resilience for a query.
We prove here that not the original, but a slight variation of that idea is indeed correct. 

We also prove that this construction is a \emph{necessary criterion for hardness of self-join free queries}
and conjecture it to be also necessary for any query.
In addition, we also give a Disjunctive Logic Program ($\ijpdlp$) that can create hardness certificates and use it to prove hardness for 5 previously open queries with self-joins.

\subsection{Independent Join Paths (IJPs)}

We slowly build up intuition to define IJPs (\cref{def:JP,def:IJP}).
Recall the concept of a \emph{canonical database} for a minimized CQ
resulting from replacing each variable with a \emph{different constant}~\cite{DBLP:conf/stoc/ChandraM77,UllmanBookPDK}.
For example $A(1),R(1,2), S(2,3), T(3,1)$ is a canonical database for the triangle query $\qtriangleunary \datarule A(x),R(x,y), S(y,z),$ $T(z,x)$.
Intuitively, one can think of a \emph{witness} as more general than a canonical database in that several variables may map to the \emph{same constant}.
A join path is then a set of witnesses that share enough constants to be connected
(this sharing of constants can be best formalized as a \emph{partition of the constants} among a fixed number of witnesses).
In addition, join paths are defined with two ``\emph{isomorphic}'' sets of tuples,
the start $\mathcal{S}$ and terminal $\mathcal{T}$ (both together called the ``\emph{endpoints}'').
We call two sets of tuples 
\emph{isomorphic} iff 
there a bijective mapping between the constants of the sets that preserves the sets of shared constants across table attributes.
For example, $\mathcal{S}_1 = \{R(1,2), A(2), R(2,2)\}$
is isomorphic to
$\mathcal{S}_2 = \{R(3,4), A(4), R(4,4)\}$
but not to
$\mathcal{S}_2 = \{R(3,4), A(4), R(4,5)\}$.

\begin{definition}[Join Path (JP)]\label{def:JP}
A database $D$ (under set or bag semantics) forms a Join Path from a set of tuples
$\mathcal{S}$ (start) to a set of tuples $\mathcal{T}$ (terminal),
for query $Q$ if 
\begin{enumerate}
	\item Each tuple in $D$ participates in some witness (i.e.\ $D$ is reduced).
	\item The witness hypergraph is connected.
	\item $\mathcal{S}$ and $\mathcal{T}$ form a \emph{valid endpoint pair}, i.e.:
	\begin{enumerate}[label=(\roman*)]
		\item $\mathcal{S}$ and $\mathcal{T}$ are isomorphic and non-identical.
		\item There is no endogenous tuple $t \in D$, $t \notin \mathcal{S} \cup \mathcal{T}$ 
		whose constants are a subset of the constants of tuples in $\mathcal{S} \cup \mathcal{T}$.
	\end{enumerate}
\end{enumerate}
\end{definition}

\begin{figure}[t]
\centering
\begin{subfigure}[b]{.37\linewidth}
	\centering
	\includegraphics[scale=0.22]{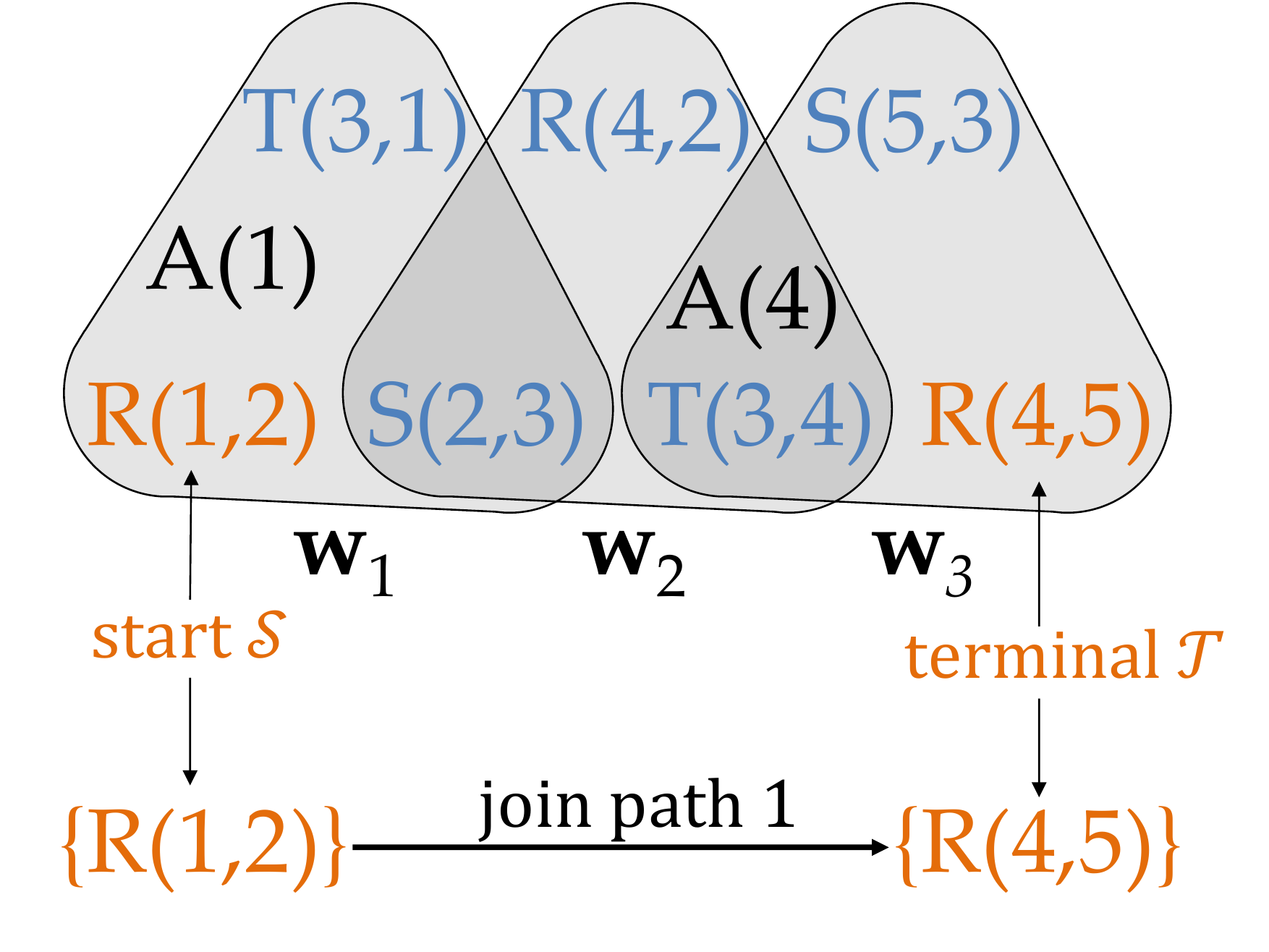}
	\caption{}\label{Fig_example_IJP}
\end{subfigure}
\begin{subfigure}[b]{.62\linewidth}
	\centering
	\includegraphics[scale=0.22]{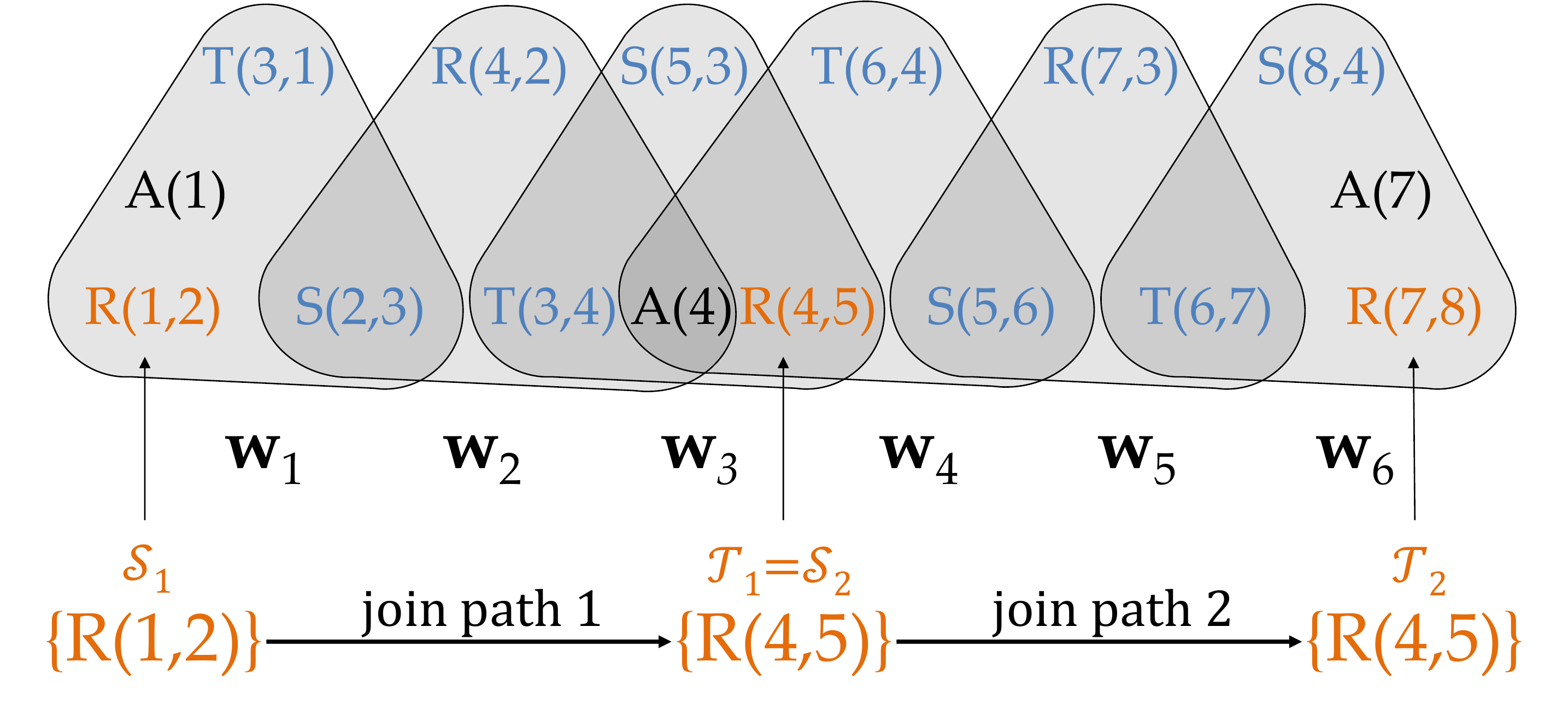}
	\caption{}\label{Fig_example_IJP_composition}
\end{subfigure}
\caption{(a) IJP for triangle query $\qtriangleunary$. 
(b) IJPs are composed by sharing their endpoints
(start or terminal tuples).}
\end{figure}

\begin{example}[Join paths]
	\label{ex:IJP:triangle}
	Consider again the query 
	$\qtriangleunary$.
	The following database of 9 tuples (\cref{Fig_example_IJP})
	$
	D=\{A(1), A(4), R(1,2), R(4,2), \\
	R(4,5), S(2,3), S(5,3), T(3,1), T(3,4) \}
	$
	where $A(1)$ and $A(4)$ are exogenous,
	forms a join path from $\{S = \{ R(1,2) \}$ to $\T = \{ R(4,5) \}$. 
	It has 3 witnesses
	$\vec w_1 = \{A(1), R(1,2), S(2,3), T(3,1)\}$,
	$\vec w_2 = \{A(4), \\ R(4,2), S(2,3), T(3,4)\}$, and
	$\vec w_3 = \{A(4), R(4,5), S(5,3), T(3,4)\}$.
	This join path can also be interpreted as a partition 
	$\{\{x^1\},\{x^2, x^3\},$ $\{y^1, y^2\}, \{y^3\}, \{z^1, z^2, z^3\} \}$
	on the canonical databases 
	for three witnesses 
	$\vec w_i = \{A(x^i), R(x^i,y^i), S(y^i,z^i),$ $T(z^i,x^i)\}$,
	$i=1,2,3$,
	expressing the shared constants in each subset.
	Then above database instance results from the following valuation $\nu$ of the quotient set
	$\{[x^1],[x^2], [y^1], [y^3], [z^1] \}$ to constants:
	$\nu: (x^1,y^1,z^1,x^2, y^3) \rightarrow (1, 2, 3, 4, 5)$.
	Notice that $\S$ and $\T$
	form a valid endpoint pair because
	($i$) $\mathcal{S}$ and $\mathcal{T}$ are isomorphic with the mapping $f=\{1:3, 2:4\}$ and 
	($ii$) there is no endogenous tuple with constants only from $\{1,2,3,4\}$. 
	$A(1)$ and $A(4)$ violate the subset requirement, however they are exogenous, so the definition is fulfilled.
\end{example}

We also call two join paths \emph{isomorphic}
if there is a bijective mapping between the shared constants across the witnesses.
Given a fixed query, we usually leave away the implied qualifier ``isomorphic'' when discussing join paths.
We talk about the ``\emph{composition}'' of two join paths
if one endpoint of the first is identical to an endpoint of the second, and all other constants are different.
We call a composition of join paths ``non-leaking'' if 
the composition adds no additional witnesses that were not already present in any of the non-composed join paths.

\begin{example}[Join path composition]
\label{ex:IJP:triangle:composition}
Consider the composition of two JPs shown in \cref{Fig_example_IJP_composition}.
They are isomorphic because there is a reversible mapping 
$(1, 2, 3, 4, 5) \rightarrow (4, 5, 6, 7, 8)$ from one to the other.
They are composed because they share no constants except for their endpoints:
The terminal
$\T_1=\{R(4,5)\}$ of the first is identical to the start of the second ($\S_2$).
The composition is non-leaking since no additional witnesses results from their composition.
\end{example}

\begin{restatable}[Triangle composition]{proposition}{proptrianglecomposition}
\label{prop:triangle:composition}
Assume a join path (JP) with endpoints $\mathcal{S}$ and $\mathcal{T}$.
If 3 isomorphic JPs composed in a triangle with directions as shown in \cref{Fig_minimal_JP_composition} are non-leaking, then any composition of JPs is non-leaking.
\end{restatable}

\begin{proofintuition*}[\cref{prop:triangle:composition}]
Since JPs can be asymmetric, the composability due to sharing the $\mathcal{S}$ tuples in two isomorphic JPs differs from sharing $\mathcal{S}$ and $\mathcal{T}$. 
We show that the three JP interactions in \cref{Fig_minimal_JP_composition} act as sufficient base cases to model all types of interactions. 
We show via induction that sharing the same end tuples across multiple JPs cannot leak if it does not leak in the base case.
\qed
\end{proofintuition*}

\begin{figure}[t]
\centering
\includegraphics[scale=0.28]{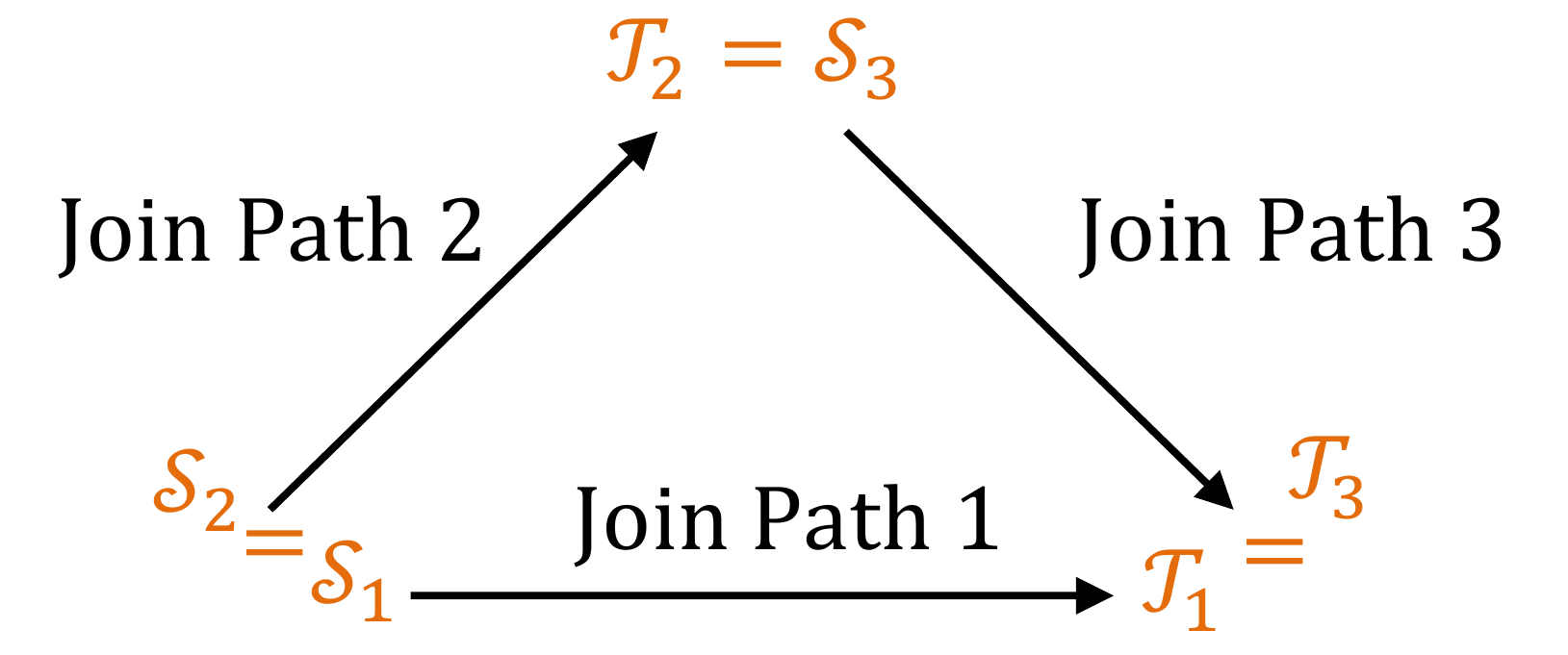}
\caption{3 JPs composed in a triangle with shown edge directions.}
\label{Fig_minimal_JP_composition}
\end{figure}

\begin{definition}[Independent Join Path]
\label{def:IJP}
A Join Path $D$ forms an \emph{Independent Join Path (IJP)} if it fulfills two additional conditions:
\begin{enumerate}
	\setcounter{enumi}{3}
	\item ``OR-property'': Let $c$ be the resilience of $Q$ on $D$.
	Then resilience is $c\!-\!1$ in all 3 cases of removing either
	$\mathcal{S}$ or $\mathcal{T}$ or both.

	\item Any composition of two or more isomorphic JPs is non-leaking.
\end{enumerate}
\end{definition}

Our definition of Independent Join Paths differs from earlier work~\cite{DBLP:conf/pods/FreireGIM20}, in that it is a completely \emph{semantic} definition that is based on all the properties that must be captured by an Independent Join Path that does not enforce any structural criteria. 
We believe such a semantic definition will help show that IJPs are a sufficient criterion for hardness. 
This definition allows us to find IJPs via an automatic search procedure (\cref{Fig_IJP_sj_new}).

\begin{example}[IJPs]
	\label{ex:IJPs}
	Consider again the JP from 
	\cref{Fig_example_IJP}.
	The resilience is $c=2$ as removing $\Gamma = \{S(2,3), T(3,4)\}$
	destroys all 3 witnesses.
	Removing $\S=\{(1,2)\}$ destroys $\w_1$, 
	and it suffices to just remove one tuple $\Gamma'=\{T(3,4)\}$ to destroy the remaining 2 witnesses. 
	Similarly, for removing either $\T$, or both $\S$ and $\T$.
	This proves the OR-property of this JP.
	Further, composing 3 JPs in a triangle as shown in \cref{Fig_minimal_JP_composition} is non-leaking 
	(the resulting database has 9 witnesses), and thus this JP is an IJP.
\end{example}

We now prove that the ability to create an IJP for a query proves its resilience to be hard.
This was left as an open conjecture in \cite[Conjecture 49]{DBLP:conf/pods/FreireGIM20}.

\begin{restatable}[IJPs $\Rightarrow$ NPC]{theorem}{thmijpsnpc}
\label{th:ICPs}	
If there is a database $D$ under set/bag semantics that forms an IJP for a query $Q$, then $\res(Q)$ is $\NPC$ for the same semantics.
\end{restatable}

\begin{proofintuition*}
	We use a reduction from minimum vertex cover to prove that $\res(Q)$ is $\NPC$ for any database that forms an IJP for $Q$.
	IJPs allow us to abstract the hardness gadgets (and can be thought of as a "template") that are used to reduce vertex cover to our problems.
	The problem of minimum vertex cover in graphs is closely related to resilience (resilience can be thought of as minimum vertex cover in the data instance hypergraph). 
	For the reduction, IJPs are used as edge gadgets to compute the Vertex Cover while the endpoint tuples form the nodes. 
	The reduction is based on the idea that a node is in the min vertex cover set iff the tuples are in the corresponding resilience/responsibility set. 
	The IJPs are designed such that they have the OR property (if one endpoint set is not chosen, then the other needs to be chosen in order to get the resilience for that edge). 
	This is just like in Vertex Cover: either one of the nodes is required and sufficient to cover an edge.
\qed
\end{proofintuition*}

We next prove that the ability to create an IJP for a self-join free CQ is not only a sufficient 
but also a \emph{necessary criterion} for hardness.
We prove \Cref{th:IJPs_for_SJFCQs}, which does not add new complexity results over \cite{DBLP:journals/pvldb/FreireGIM15}, 
but together with 
\cref{th:ICPs} 
shows that IJPs are strictly more general and thus a strictly more powerful criterion for resilience than the previous notion of triads\cite{DBLP:journals/pvldb/FreireGIM15} (a triad always implies an IJP, but not vice versa)
: they capture the same hardness for SJ-free queries, 
but can also prove hardness for queries with self-joins that do not contain a triad.

\begin{restatable}[IJPs $\Leftrightarrow$ NPC for SJ-free CQs]{theorem}{thmijpsjfcqs}
\label{th:IJPs_for_SJFCQs}	
The resilience of a SJ-free CQ under set/bag semantics is $\NPC$ iff it has an IJP under the same semantics.
\end{restatable}

\begin{proofintuition*}[\cref{th:IJPs_for_SJFCQs}]
We generalize all past hardness results \cite{DBLP:journals/pvldb/FreireGIM15} 
for SJ-free queries by showing that the same hardness criteria (\emph{triads}) 
that was necessary and sufficient for hardness, 
can always be used to construct an IJP and show this construction.
\qed
\end{proofintuition*}

We conjecture that the existence of an IJP is a necessary criterion for hardness for all queries. 
In addition, we conjecture that the size of smallest IJP formed by database under a hard query $Q$ is bounded by a small constant factor of the query size.

\begin{conjecture}[Necessary hardness condition]
	\label{conj:necessaryhardnesscondition}
	If there exists no database $D$ under set/bag semantics that forms an IJP from some tuples $\mathcal{S}$ to $\mathcal{T}$ under query $Q$, then $\res(Q)$ is in $\PTIME$ under the same semantics.
\end{conjecture}

\begin{conjecture}[IJP Size Bound]
	\label{conj:hardness}
	If there exists a database $D$ under set/bag semantics of domain size that forms an IJP under query $Q$, then there exists a database under same semantics as $D$, with domain size $d\leq 7 \cdot |\var(Q)|$, that forms an IJP from some tuples $\mathcal{S}$ to $\mathcal{T}$ under query $Q$.
\end{conjecture}

\begin{intuition*}[\cref{conj:hardness}]
	The intuition for bounding the size of the certificate to domain $d=7 \cdot |\var(Q)|$ comes from the connections between an IJP and the OR property. 
	Each known IJP exhibits a “core” of 3 witnesses that exhibit the OR property (which can be seen simply in the self-join free case as parallel to the three independent relations of the triad as in \cref{Fig_example_IJP}). 
	This core could take up to $d=3 \cdot |\var(Q)|$ size. 
	However, this “core” may (1) not have isomorphic endpoint tuple pairs and (2) not be able to exist "independently" and form additional witnesses under $Q$ due to Join dependencies (this is the intuition behind \hyperref[{def:IJP}]{\cref{def:IJP} (5)}). 	
	We hypothesize that the endpoint tuple pairs can each be connected to “legs” of 2 witnesses each, thus resulting in a new endpoint pair that is isomorphic. This would add up to 2 times $2 \cdot |\var(Q)|$ constants, bringing the total size up to $7 \cdot |\var(Q)|$. 
	To resolve (2), we must add the witnesses formed due to join dependencies to the certificate. 
	However, this does not increase the number of constants used and hence we hypothesize $d=7 \cdot |\var(Q)|$ as an upper bound. 
	We show an additional figure in the appendix (\cref{Fig_IJP_sj_new_legs}),	
    in which we highlight the cores and legs of the example IJPs in \cref{Fig_IJP_sj_new}.
    \qed
\end{intuition*}

\subsection{Automatic creation of hardness certificates}
We introduce a Disjunctive Logic Program $\ijpdlp$ 
that finds IJPs to prove hardness for $\res$.
Each DLP requires $Q$, a domain $d$ (which bounds the size of the IJP), and two endpoints $\S, \T$.\footnote{Since the number of possible endpoint configurations is polynomial in the query size, we can simply run parallel programs for different endpoints as input. Notice that endpoints $e_1 = \{A(1)\}, e_2 = \{A(2)\}$ is exactly the same as $e_1 = \{A(3)\}, e_2 = \{A(4)\}$ since the actual value does not matter.
In practice, we used any subset of endogenous tuples from a canonical database that can be shared across two witnesses without creating another witness.}
$\ijpdlp$ programs are generated automatically for a given input, are short (200-300 lines depending on the query) 
and leverage many key technical insights used to model DLPs.

The goal of $\ijpdlp$ is to find a database that fulfills the conditions of \cref{def:IJP}. 
The search space is a database with all possible tuples given domain $d$
(thus of size $\O(d^a)$ where $a$ is the maximum arity of any relation).
Each tuple in the search space must be either ``picked'' in the target database or not.
The constraints of our definition are modeled as disjunctive rules with negation.
We solve our DLP with the open-source ASP solver \emph{clingo}~\cite{clingo}
which uses an enhancement of the DPLL algorithm~\cite{DPLL}
(used in SAT solvers) and works far faster in practice than a brute force approach.
Here we talk only about the overall structure and intuition, but make examples available in the code~\cite{resiliencecode} and in \cref{sec:appendix:dlp}.

\begin{enumerate}
	\item \textbf{Search Space:}
	For all relations in $Q$, we initialize all possible tuples permitted in domain $d$ as input facts and provide them with an additional tuple id ($\TID$). 
	Thus, each relation $R$ has a corresponding relation in the program with $\arity(R)^d$ facts. 

	\item \textbf{``Guess'' an IJP:} 
	Each tuple either participates in the IJP or not.
	We follow the Guess-Check methodology~\cite{eiter2006towards}
	and use a relation $\indb(\underline{R}, \underline{\TID}, I)$ 
	to ``guess'' for each tuple whether it is in the IJP database or not.
	Here $R$ stands for a relation and together with $\TID$ uniquely identifies a tuple.
	The binary value $I$ is $1$ if the tuple is in the IJP, and $0$ otherwise.

	\item \textbf{Enforce JP endpoint conditions:} Since the endpoints are considered ``input'', 
	we do not need to check condition ($3i$) for the JP endpoints (\cref{def:JP}).
	However, we need to verify condition ($3ii$) as it depends on the other tuples in the IJP and translate the condition  directly into a logic rule.

	\item \textbf{Calculate Resilience using ``Saturation'':} 
	We solve a problem that is $\NPC$
	(i.e.\ check that there is a valid contingency set of size $c$), 
	and a problem that is co-$\np$-complete (i.e.\ there is no valid contingency set of size $c-1$). 
	For solving the $\np$ problem we use the guess-check methodology and to solve the co-NP problem, we use the saturation technique. 

	\item \textbf{Enforce OR-property:} We calculate resilience for $4$ databases using the previous step: 
	our original ``guess'', and the guess with either or both endpoints removed. 
	The removal of endpoints here simply implies defining a new relation 
	that has all tuples of $indb$ except the removed endpoint tuples.

	\item \textbf{Enforce non-leaking composition:} We define a mapping relation to create $3$ isomorphs of the tuples in $indb$. We combine them into one database and check that computing query $Q$ results in exactly $3$ times the number of original witnesses. 

	\item \textbf{(Optional) Minimize the size of the IJP:} 
	To generate smaller certificates that are more human-readable, we simply minimize the number of witnesses in the IJP.
	We use weak constraints \cite{eiter2009answer} to perform this optimization.  

\end{enumerate}

\begin{corollary}[Sufficient hardness condition]
If there is a domain $d$ and endpoints $\mathcal{S}, \mathcal{T}$ 
such that $\ijpdlpparam{Q, d, \mathcal{S}, \mathcal{T}}$ is satisfiable, then $\res(Q)$ is $\NPC$.
\end{corollary}

\begin{figure*}[t]{}
\centering
	\includegraphics[width=\textwidth]{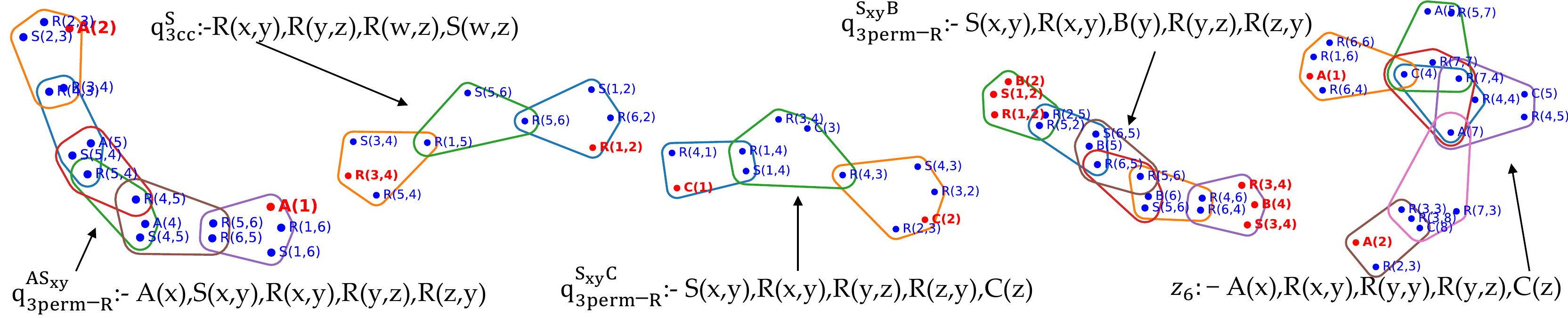}
\caption{Automatically generated and visualized IJPs for $5$ previously open queries. The nodes corresponding to tuples in $\mathcal{S} \cup \mathcal{T}$ are in red.
}
\label{Fig_IJP_sj_new}
\end{figure*}

\begin{corollary}[Complexity bound]
	It is in $\Sigma^2_p$ of $d$ to check if a query $Q$ can form an IJP of domain size $d$ or less. 
\end{corollary}

The guarantees of our DLP is one-sided: if it finds a certificate, then resilience of the query is guaranteed to be $\NPC$.
If it does not provide a certificate, then we have no guarantee.
So far we have not found any query that is known to be hard and for which our DLP could not create a certificate for 
$d=3 \cdot |\var(Q)|$.
This is in line with \cref{conj:hardness} that implies that $\ijpdlp$ is not only a sufficient but also complete algorithm for $d=7 \cdot |\var(Q)|$ (i.e.\ if the algorithm does not find a certificate for $d= 7 \cdot |\var(Q)|$, then the query is in $\PTIME$).

\introparagraph{Example IJPs}
Prior work\!~\cite{DBLP:conf/pods/FreireGIM20} left open the complexity of resilience for $7$ binary CQs with three self-join atoms.
Our DLP proved $5$ of them to be hard (\cref{Fig_IJP_sj_new} shows them and their IJPs).

\section{Complexity results for SJ-free CQs}
\label{sec:theoreticalresults}
\label{SEC:THEORECTICALRESULTS}

This section gives complexity results for both $\res$ and $\rsp$ for SJ-free queries, under set and bag semantics (see \cref{Fig_query_venn}).
Our results include both prior known results and new results.
Importantly, all our hard cases are derived with our unified hardness criterion (IJPs) from \cref{sec:IJP}, 
and all tractable cases
follow from our unified algorithms in \cref{sec:res-ilp,sec:resp-ilp,sec:relaxations}.

\subsection{Necessary notations}

Before diving into the proofs, we define a few key concepts stemming from domination (\cref{def:domination}) that lead up to the three structural criteria (\cref{def:activedeactivetriads}) which completely describe our dichotomy results.
Notice that the notion of \emph{triads} has been previously defined~\cite{DBLP:journals/pvldb/FreireGIM15}.
However, we extend this notion and make it more-fine grained. 
The previous definition of \emph{triad} now corresponds exactly to the special case of ``\emph{active triads}.'' 

\begin{definition}[Domination~\cite{DBLP:journals/pvldb/FreireGIM15}]\label{def:domination}
	In a query $Q$ with endogenous atoms $A$ and $B$, we say $A$ \emph{dominates} $B$ iff $\var(A) \subset \var(B)$.
\end{definition}

\begin{definition}[Triad (different from \cite{DBLP:journals/pvldb/FreireGIM15})]
	\label{def:triad}
	A \emph{triad} is a set of three atoms, $\mathcal{T} = \set{R_1, R_2, R_3}$  
	s.t.\ for every pair $i \neq j$, there is a path from $R_i$ to $R_j$ that uses no variable occurring in the third atom of $\mathcal{T}$.
\end{definition}

\begin{definition}[Solitary variable \cite{DBLP:journals/pvldb/FreireGIM15}]\label{def:solitaryvariable}
	In a query $Q$ 
	a variable $v$ in relation $A$ is solitary if, 
	in the query hypergraph it cannot reach any
	endogenous atom $B\neq A$  without passing through one of the nodes in $\var(A)-v$.
\end{definition}

\begin{definition}[Full domination \cite{DBLP:journals/pvldb/FreireGIM15}]\label{def:fulldomination}
	An atom $A$ of CQ $Q$ is
	\emph{fully dominated} iff for all non-solitary variables $y\in \var(A)$
	there is another atom $B$ such that $y\in\var(B)\subset\var(A)$.
\end{definition}

\begin{definition}[Active or (fully) deactivated triads]
A triad is \emph{deactivated} iff at least one of its three atoms is dominated by another atom of the query.
A triad is \emph{fully deactivated} iff at least one of its three atoms is \emph{fully} dominated by another atom of the query.
A triad is \emph{active} iff none of its atoms are dominated.
\label{def:activedeactivetriads}
\end{definition}

We call queries \emph{linear} if they do not contain triads.
Here we depart from prior work that referred to linear queries as queries without what we now call active triads~\cite{DBLP:journals/pvldb/FreireGIM15}. 
We instead say that queries without active triads are \emph{linearizable}.\footnote{The intuition of ``linearity'' is that the vertices of the dual hypergraph $H_d$ of $Q$ can be mapped onto a line s.t.\  $H_d$ has the running intersection property~\cite{Beeri+83}. }

\begin{example}
	Consider the triad $\{R, S, T\}$ in all 3 queries
	$\qtriangle$, 
	$\qtriangleunary$,
	and
	$\qtrianglebinary$ from \cref{Fig_query_venn}.
	The triad is deactivated in $\qtriangleunary$ and $\qtrianglebinary$ because $A$ dominates both $R$ and $T$.
	The triad is fully deactivated in $\qtrianglebinary$ because $T$ is fully dominated by $A$ and $B$.
	The triad is active in $\qtriangle$ since none of the three tables in the triad are 
	dominated.
	The chain with ends query  
	$\qtwochainwe$ 
	has no triad and is thus linear.
	\label{example:triads}
\end{example}

\subsection{Dichotomies for $\resdecision$ under Sets and Bags}
\label{sec:theory:resilience}

This section proves that for all SJ-free CQs, 
either $\reslp$ solves $\resdecision$ exactly (and the problem is hence easy for any instance),
or we can form an IJP (and thus the problem is hard).
Our results cover both set and bag semantics (see \cref{Fig_query_venn}).

\begin{restatable}{theorem}{thmreseasy}
  \label{thm:reseasy}
  $\reslpparam{Q,D} = \res^*(Q,D)$ for all database instances $D$
  under set or bag semantics
  if $Q$ is linear.
\end{restatable}

\begin{proof}[Proof \cref{thm:reseasy}]
	Prior approaches show that the witnesses generated by a linear query $Q$ over database instance $D$ can be encoded in a flow graph \cite{DBLP:journals/pvldb/MeliouGMS11} such that each path of the flow graph represents a witness and each edge with non-infinite weight represents a tuple.
	The flow graph is such that an edge participates in a path iff the corresponding tuple is part of the corresponding witness. 
	The min-cut of this graph (or the minimum edges to remove to disconnect the source from the target), is equal to $\res(Q,D)$.
	We use this prior result to prove that $\reslpparam{Q,D} = \res^*(Q,D)$ by showing that the Linear Program solution is a valid cut for the flow graph, and vice versa.
	Then the minimal cut must also be admitted by $\reslpparam{Q,D}$ and $\reslpparam{Q,D}$ also cuts the flow graph.
	Assume we have a fractional LP solution - then for each witness, we still fulfill the constraint that sum of all tuple variables $\geq 1$.
	This implies that the path corresponding to each witness has been cut.  
	Since the number of paths in the flow graph is equal to the number of witnesses, all paths from source to target are cut. 
	By the max-flow Integrality Theorem, there is an equivalent optimal integral solution as well. 
	This integral solution still cuts all paths, and fulfills all conditions of the LP.
	Thus, for linear queries, $\reslpparam{Q,D} = \resilpparam{Q, D} = RES(Q,D)$.
\qedhere
\end{proof}

\begin{restatable}{theorem}{thmreseasydomination}
  \label{thm:reseasydomination}
  $\reslpparam{Q,D} = \res^*(Q,D)$ for all database instances $D$ under set semantics if all triads in $Q$ are deactivated.
\end{restatable}

\begin{proofintuition*}[\cref{thm:reseasydomination}]
	Prior work~\cite{DBLP:journals/pvldb/FreireGIM15} has shown that queries that contain only deactivated triads (previously called dominated triads) can be linearized due to domination (\cref{def:domination})
	We show that this linearization does not change the optimal solution to the LP formulation under set semantics.
	This is since the dominated table in the deactivated triad can simply be made exogenous, resulting in a linear query. 
	This is equivalent to saying that there is an optimal solution of $\reslp$ where the decision variables of all tuples in dominated table are set to $0$.
	Thus, $\reslp$ models a linear query indirectly, and hence \cref{thm:reseasy} applies to complete the proof.
	Notice that domination does not work under bag semantics, which leads to a different tractability frontier.
	\qed
\end{proofintuition*}

\begin{restatable}{theorem}{thmreshardbag}
	\label{thm:reshardbag}
	$\res(Q)$ is $\NPC$ under bag semantics if $Q$ is not linear.
\end{restatable}

\begin{proofintuition*}[\cref{thm:reshardbag}]
For queries with active triads, the IJPs (\cref{th:IJPs_for_SJFCQs}) imply hardness for bag semantics as well. 
We prove that \emph{all triads} are hard by showing that including a fixed number of copies of a dominating table is equivalent to making it exogenous. 
This is equivalent to creating a new IJP where the tuples of the dominating table have $c_w$ copies, where $c_w$ is the number of witnesses in the IJP under set semantics.
Now, no minimal contingency set will use tuples of the dominating table, and hence we must consider the tuples from the dominated tables still.
Thus, domination does not work under bag semantics, and any triad (even a fully deactivated one) implies hardness.
\qed
\end{proofintuition*}

The results in this section, along with \cref{th:IJPs_for_SJFCQs} imply the following dichotomies under both set and bag semantics:

\begin{corollary}
Under set semantics,
$\res^*(Q)$ is in $\PTIME$ for queries that do not contain active triads, otherwise it is NPC.
\end{corollary}
\begin{corollary}
Under bag semantics,
$\res^*(Q)$ is in $\PTIME$ for queries that do not contain triads, otherwise it is NPC.
\end{corollary}

\subsection{Dichotomies for $\rspdecision$ under Sets and Bags}
\label{sec:theory:responsibility}

This section follows a similar pattern as the previous one
to prove that for every SJ-free CQ, either $\rspmilp$ solves $\rsp$ exactly
(and the problem is hence easy), 
or we can form an IJP for $\rsp$.

\begin{restatable}{theorem}{thmrspeasy}
  \label{thm:rspeasy}
  $\rspmilpparam{Q,D,\resptuple} =$ $\rsp^*(Q,D,\resptuple)$ for all database instances $D$ under set or bag semantics if $Q$ is linear.
\end{restatable}

\begin{proof}[Proof \cref{thm:rspeasy}]
		Let $X_m$ be an optimal variable assignment generated by solving $\mathtt{MILP}[\rspdecision^*(
			\\{Q,D,\resptuple})]$ 
		There must be at least one witness $w_p \in D$ such that $\resptuple \in w_p$ and $X_m[w_p] = 0$ i.e.\ the witness is not destroyed (this follows from the fact that the counterfactual clause enforces that all witnesses containing $t$ cannot take value $1$). 
		For such a witness, any tuple $t' \in w_p$, must have $X[t'] = 0$ since it satisfies the witness tracking constraints.
		We also know that since $Q$ is a linear query, the witnesses can be encoded in a flow graph to find the responsibility \cite{DBLP:journals/pvldb/MeliouGMS11,DBLP:journals/pvldb/FreireGIM15}.
		We can map the values of $X_m$ to the flow graph, where $X_m[t]$ now denotes if an edge in the flow graph is cut or not.
		Consider $X_m[\resptuple] = 0$, since it is not modeled in $\rspmilp$.
		We see that this disconnects all paths in the graph (since paths that do not contain $\resptuple$ are disconnected by virtue of the resilience constraints of $\rspmilp$).
		If we set the weight of all tuples in $w_p$ to $\infty$, the cut value does not change since these tuples were not part of the cut.
		Prior work \cite{DBLP:journals/pvldb/MeliouGMS11} has shown that $\rsp(Q,D)$ for linear queries can be calculated by taking the minimum of min-cuts of \emph{all} flow graphs such that have $1$ of witnesses that contains $\resptuple$, has weight of all other tuples edges set to $\infty$.
		Thus, $\rspmilpparam{Q,D,\resptuple}$ is at least as much as the responsibility computed by a flow graph.
		In addition to this, the flow graph with the smallest cut also fulfills all the solutions for $\rspmilp$ (since at least one witness containing $\resptuple$ is preserved, and all witnesses not containing $\resptuple$ are cut).
		Thus, the optimal value of $\rsp(Q,D, \resptuple)$ can be mapped back to a $\rspmilp$ assignment.
	\qedhere
\end{proof}

\begin{restatable}{theorem}{thmrspeasyfullydominated}
  \label{thm:rspeasyfullydominated}
  $\rspmilpparam{Q,D,\resptuple} = \rsp^*(Q,D,\resptuple)$ for any database
  $D$ 
  under set semantics
  if all triads in $Q$ are fully deactivated.
\end{restatable}

\begin{proofintuition*}[\cref{thm:rspeasyfullydominated}]
This follows directly from the fact that fully deactivated triads can be linearized without changing the optimal solution~\cite{DBLP:journals/pvldb/FreireGIM15} and \cref{thm:rspeasy}.
\qed
\end{proofintuition*}

\begin{restatable}{theorem}{thmrspeasydominatingtable}
\label{thm:rspeasydominatingtable}
$\rsplpparam{Q,D,\resptuple} = \rsp^*(Q,D,\resptuple)$ for all database instances $D$ under set semantics
if $Q$ does not contain any active triad
and $\resptuple$ belongs to an atom that dominates some atom in all deactivated triads in $Q$.
\end{restatable}

\begin{proofintuition*}[\cref{thm:rspeasydominatingtable}]
	We prove that in every deactivated triad dominated by $A$, it is always safe to make the dominated table $R$ exogenous since any tuple from $R$ in the responsibility set is either replaceable, or invalid. 
	This linearizes the query, and the rest follows from \cref{thm:rspeasy}.
	Notice that prior work \cite{DBLP:journals/pvldb/FreireGIM15} identified as tractable cases those without any \emph{active} triad, 
	which a special case of our more general tractable cases.
	\qed
\end{proofintuition*}

\begin{restatable}{theorem}{thmrspharddominatedtable}
  \label{thm:rspharddominatedtable}
  $\rsp(Q,D,\resptuple)$ is $\NPC$ if $\resptuple$ belongs to 
  an atom that
  is part of a triad that is not fully deactivated.
\end{restatable}

\begin{proofintuition*}[\cref{thm:rspharddominatedtable}]
The key principle behind this proof is our more fine-grained notion of exogenous tuples. A tuple $a$ such that $a$ has all the same values for the same variables as $t$ and $\var(a) \subseteq \var(\resptuple)$ 
is necessarily exogenous since it is not possible for $\resptuple$ to become counterfactual if $a$ is removed. 
We construct an IJP  possible due to such an \emph{exogenous tuple} from a dominated table.
\qed
\end{proofintuition*}
	
\begin{restatable}{theorem}{thmrspharderthanres}
  \label{thm:rspharderthanres}
  If \,$\res(Q)$ is $\NPC$ for a query $Q$ under set or bag semantics then so is $\rsp(Q)$.
\end{restatable}

\begin{proofintuition*}[\cref{thm:rspharderthanres}]
We give a reduction from $\res(Q)$ to $\rsp(Q)$ in both set and bag semantics
by adding a witness to the given database instance 
and selecting a tuple whose responsibility 
is equal the resilience of the original instance. 
Our approach extends a prior result~\cite{DBLP:journals/pvldb/FreireGIM15} that applied only to set semantics.
\qed
\end{proofintuition*}

These results imply the following dichotomies
under both set and bag semantics:

\begin{corollary}
Under set semantics,
$\rsp(Q)$ is in $\PTIME$ for queries that contain only fully deactivated triads
or {deactivated triads} that are dominated by the relation of $\resptuple$, otherwise it is NPC.
\end{corollary}

\begin{corollary}
Under bag semantics,
$\rsp(Q)$ is in $\PTIME$ for queries that do not contain any triads, otherwise it is NPC.
\end{corollary}

Notice that the tractability frontier for bag semantics notably differs from set semantics,
where the tractable cases for $\rsp(Q)$ are \emph{a strict subset} of those for $\res(Q)$. 
For bags, they coincide:

\begin{corollary}
Under bag semantics, the tractable cases for $\rsp(Q)$ are the same as for $\res(Q)$.
\end{corollary}

\section{Three Approximation Algorithms}
\label{sec:approximations}
\label{SEC:APPROXIMATIONS}

We describe one LP-based approximation algorithm
and two flow-based approximation algorithms for $\res$ and $\rsp$, all three of which apply to both set and bag semantics.

\subsection{LP-based m-factor Approximation}
For a given query with $m$ atoms, we use a standard LP rounding technique \cite{vazirani2001approximation} with the threshold of $1/m$ i.e., we round up variables whose value is $\geq\!1/m$ or set them to $0$ otherwise.

\begin{restatable}{theorem}{thmkfactorapproxres}
    \label{thm:k-factor-approx}
    The LP Rounding Algorithm is a $\PTIME$ $m$-factor approximation for $\res$ and $\rsp$.
\end{restatable}

\begin{proofintuition*}[\cref{thm:k-factor-approx}]
	Verification of $\PTIME$ solvability and the m-factor bound is trivial, and correctness follows by showing validity of each constraint for a rounded solution.
    \qed
\end{proofintuition*}

\subsection{Flow-based Approximations}
Non-linear queries cannot be encoded as a flow graph since they do not have the running-intersection property.
The idea behind flow-based approximations is to add either 
witnesses or tuples 
(while keeping the other constant) 
to linearize a non-linear query.
This works since adding more tuples or witnesses can only increase $\res$ and $\rsp$ for monotone queries.
Since there are multiple arrangements to linearize a query, we take the minimum over all non-symmetric arrangements,
explained next for
the two variants:

\introparagraph{Constant Tuple Linearization Approximation (\texttt{Flow-CT})}
We keep the same tuples as the original database in each arrangement.
However, since the query is non-linear, these flow graphs may have spurious paths that do not correspond to any original witnesses, thus inadvertently adding witnesses.
For a query with $m$ atoms, there are up to $m!/2$ linearizations due to the number of asymmetric ways to order them.

\begin{figure}[t]
	\centering
	\begin{subfigure}[b]{.49\linewidth}
		\centering
		\includegraphics[scale=0.5]{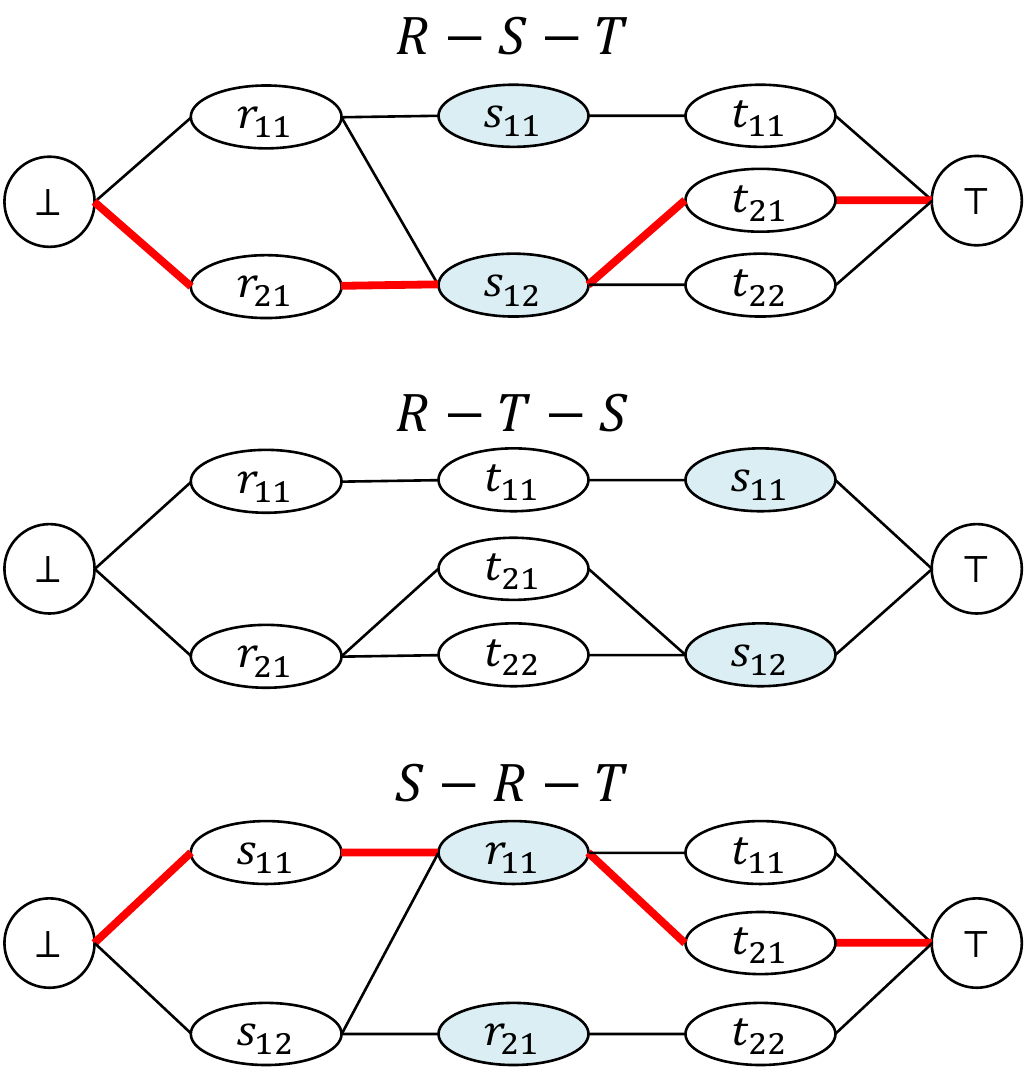}
		\caption{}\label{Fig_example_Flow_CTL}
	\end{subfigure}
	\begin{subfigure}[b]{.49\linewidth}
		\centering
		\includegraphics[scale=0.5]{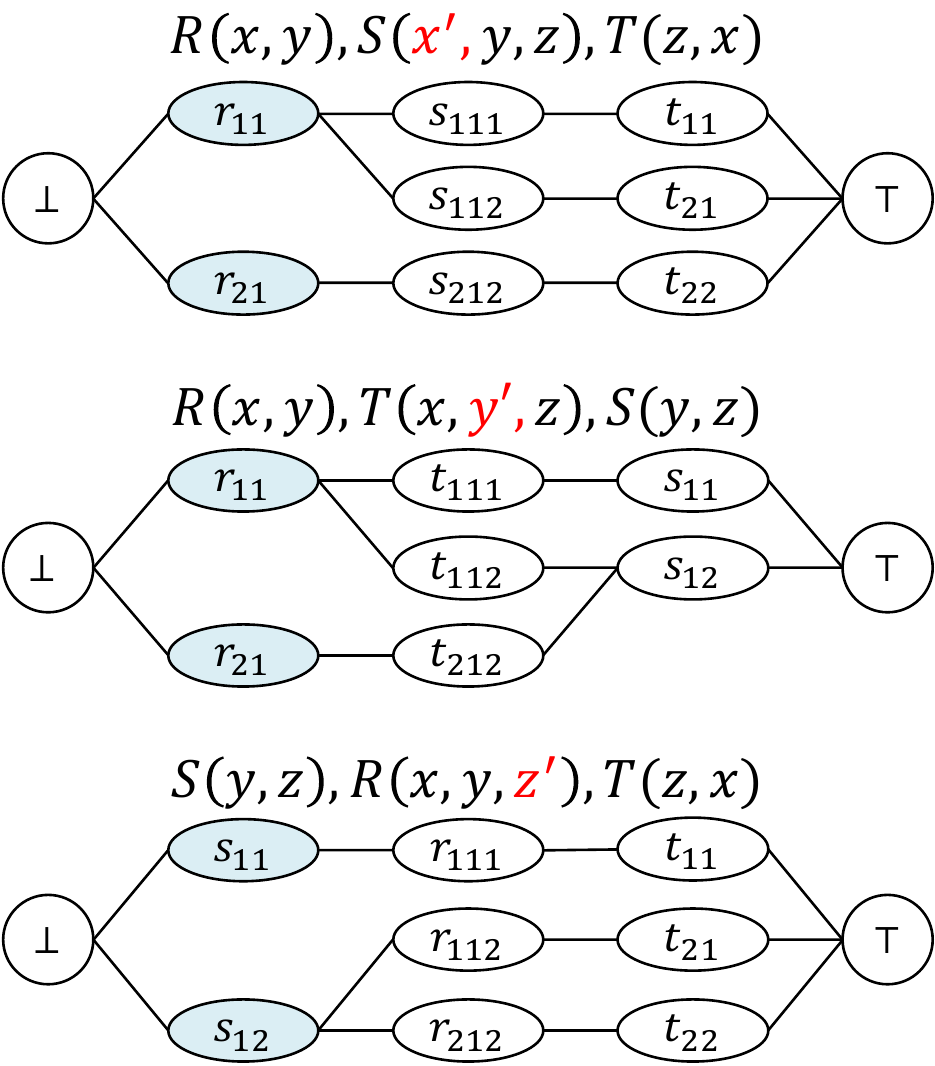}
		\caption{}\label{Fig_example_Flow_CWL}
	\end{subfigure}
\caption{Flow approximation linearizations for \cref{example:flow-approx}. 
We use $\bot$ and $\top$ to represent the source and the target, respectively, of the flow graph to make a connection to an ordering of the atoms of the query.}
\end{figure}

\introparagraph{Constant Witness Linearization Approximation (\texttt{Flow-CW})}
We keep the same witnesses as the original database instance in each linearization, however the query is changed by adding variables to tables (which is equivalent to dissociating tuples) to make it linear. 
The number of such linearizations is equal to the number of minimal dissociations~\cite{DBLP:journals/vldb/GatterbauerS17}.
\footnote{
A detail of implementation here is that for $\rsp$ it is possible that responsibility tuple $\resptuple$ is split into multiple tuples. 
Then we find responsibility over the set of those tuples, instead of a single tuple. 
This is a simple extension to make, but differs from the standard definition of responsibility, which allows for just one responsibility tuple.}

\begin{example}
    Consider the $\qtriangle$ query with the following witnesses:
    \begin{center}
        \begin{tabular}{|c|c|c|l}
            \cline{1-3}
            x & y & z & \\
            \cline{1-3}
            1 & 1 & 1 & $\vec w_1 =$ $\{r_{11},$ $s_{11}$, $t_{11}\}$ \\
            1 & 1 & 2 & $\vec w_2 =$ $\{r_{11},$ $s_{12}$, $t_{21}\}$ \\
            2 & 1 & 2 & $\vec w_3 =$ $\{r_{21},$ $s_{12}$, $t_{22}\}$\\
            \cline{1-3}
        \end{tabular}
    \end{center}
    Then there are $3$ \texttt{Flow-CT} linearizations (\cref{Fig_example_Flow_CTL}) and $3$ \texttt{Flow-CW} linearizations (\cref{Fig_example_Flow_CWL}). 
    The approximated resilience corresponds to the minimum of the min-cut over all linearized flow graphs.
    In this example, we see that both \texttt{Flow-CT} and \texttt{Flow-CW} happen to return the optimal value of $2$ as approximation. 
    \label{example:flow-approx}
\end{example}

\section{Experiments}
\label{sec:experiments}
\label{SEC:EXPERIMENTS}

Our experimental objective is to answer the following questions:
(1) How does our ILP scale for $\PTIME$ queries,
and how does it compare to previously proposed algorithms that use flow-based encodings~\cite{DBLP:journals/pvldb/MeliouGMS11}?
(2) Are our LP relaxations (proved to be correct for $\PTIME$ queries in \cref{sec:theoreticalresults}) 
indeed correct in practice?
(3) What is the scalability of ILPs and LPs for settings that are proved $\NPC$?
(4) What is the quality of our approximations from \cref{sec:approximations}?

\introparagraph{Algorithms}
\texttt{ILP} denotes our ILP formulations for $\res$ and $\rsp$.
\texttt{ILP(10)} denotes the solution obtained by stopping the solver after $10$ seconds.\footnote{Solvers often already have the optimal solution by this cutoff, despite the ILP taking longer to terminate.
This is because although the solver has stumbled upon an optimal solution, it may not yet have a proof of optimality (in cases where LP!=ILP).}
\texttt{LP} denotes LP relaxations for $\res$ and $\rsp$.
\texttt{MILP} denotes the MILP formulation for $\rsp$.
\texttt{Flow} denotes an implementation of the prior max-flow min-cut algorithm for $\res$ and $\rsp$
for queries that are in $\PTIME$~\cite{DBLP:journals/pvldb/FreireGIM15, DBLP:journals/pvldb/MeliouGMS11}.\footnote{For the min-cut algorithm, we also experimented with both LP and Augmented Path-based algorithms via the NetworkX library \cite{networkx}.
Since the time difference in the methods was not significant, we leave it out and all running times reported in the figures use the same LP library Gurobi~\cite{gurobi}.}
\texttt{LP-UB} denotes our $m$-factor upper bound obtain by the LP rounding algorithm.
\texttt{Flow-CW} and \texttt{Flow-CT} represent our approximations via
Constant Witness Linearization and Constant Tuple Linearizations, respectively.

\introparagraph{Data}
We use both synthetic and TPC-H data~\cite{tpch}. 
For any synthetic data experiment, 
we fix the maximum domain size, and sample randomly from all possible tuples.
For testing our methods under bag semantics, each tuple is replicated by a random number that is smaller than a pre-specified max bag size.
For TPC-H data, we use the TPC-H data generator at logarithmically increasing scale factors, creating $18$ databases ranging from scale factor $0.01$ to $1$.

\introparagraph{Software and Hardware} 
We implement the algorithms using Python 3.8.5
and solve the respective optimization problems with Gurobi Optimizer $8.1.0$~\cite{gurobi}. 
Experiments are run on an Intel Xeon E5-2680v4 @2.40GHz machine available via the Northeastern Discovery Cluster.

\introparagraph{Experimental Protocol}
For each plot we run $30$ runs of logarithmically and monotonically increasing database instances. 
We plot all obtained points with a low saturation, 
and draw a trend line between the median points from logarithmically increasing sized buckets.
All plots are log-log, with the x-axis representing the number of witnesses.
The y-axis for plots on the left shows the solve-time (in seconds)
taken by the solver to solve a $\res$, $\rsp$ or min-cut problem.\footnote{The build-times to create the ILP or flow graphs are not plotted since they were negligible in comparison to the solve-time.}  
We include a dashed line to show linear scalability as reference in the log-log plot.

\begin{figure}
	\begin{subfigure}{\columnwidth}
	\centering
	\includegraphics[width=0.62\columnwidth]{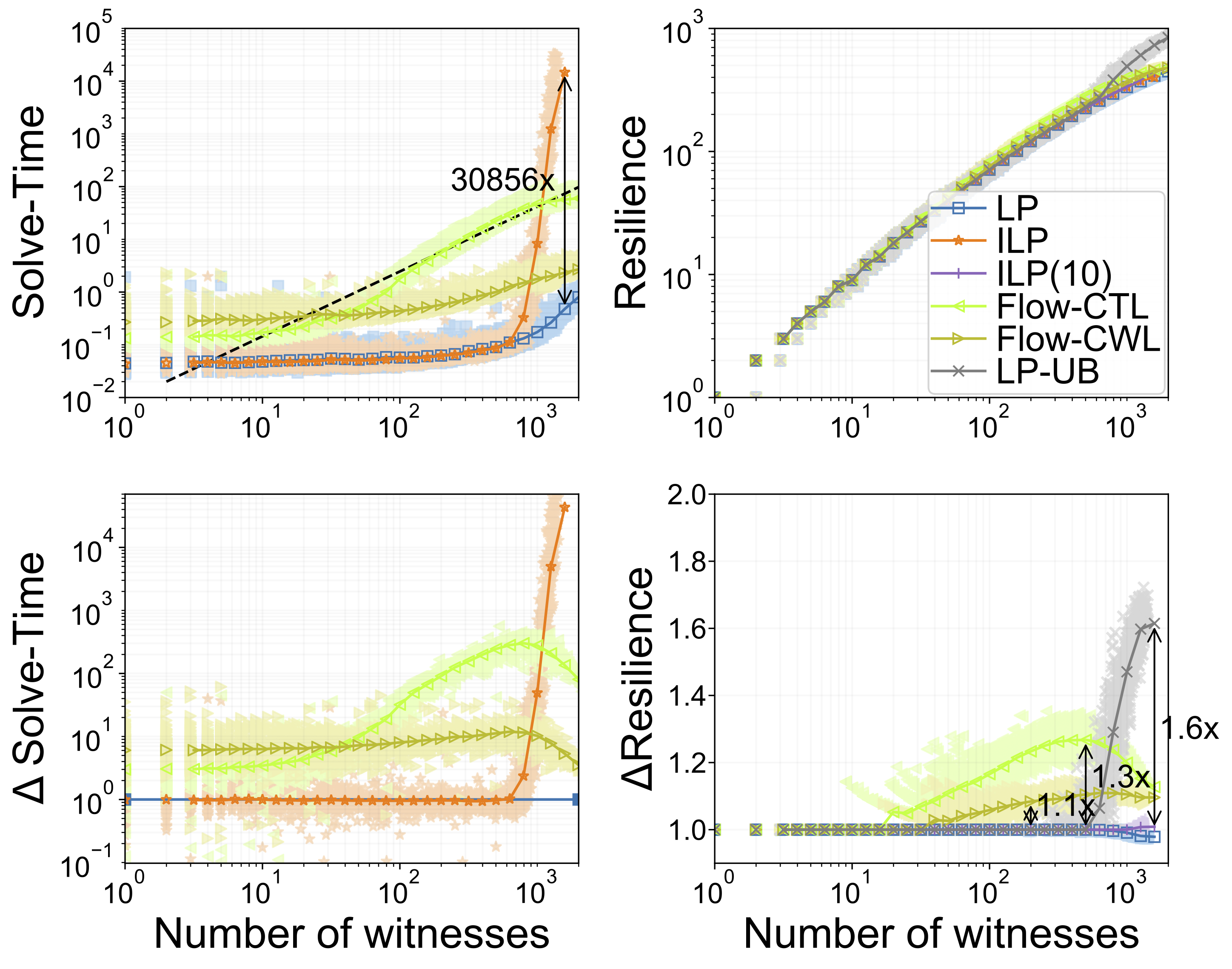}	%
	\end{subfigure}
	\caption{
		Setting 1: Hard 3-star query $\qthreestar$.
		}
		\label{fig:Fig_Expt_3Star}
\end{figure}%
\begin{figure}
	\begin{subfigure}{\columnwidth}
		\centering
		\includegraphics[width=0.62\columnwidth]{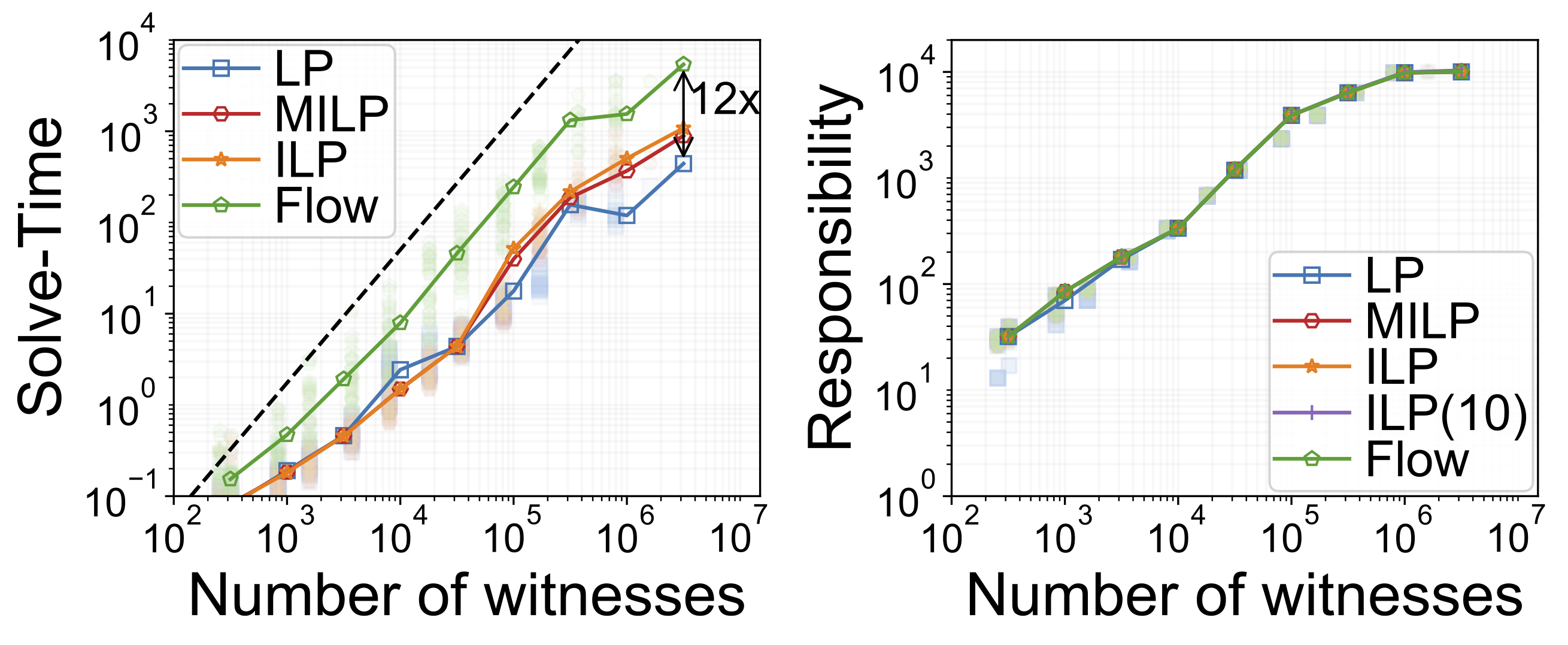}
		\caption{5 Chain Query $\qfivechain$ (an easy query)}
		\label{fig:Fig_expt_tpch_5chain}
	\end{subfigure}
	\begin{subfigure}{\columnwidth}
		\centering
			\includegraphics[width=0.62\columnwidth]{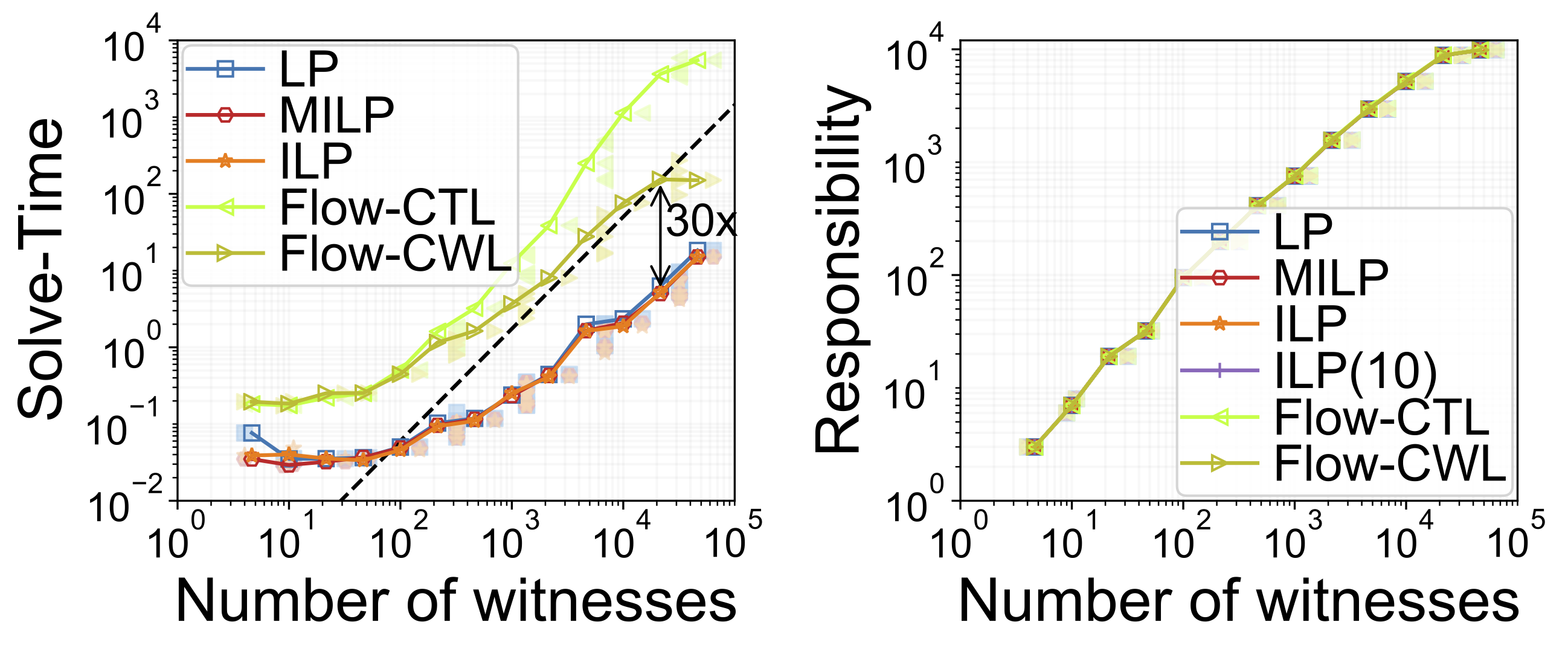}
		\caption{5 Cycle Query $\qfivecycle$ (a hard query)}
		\label{fig:Fig_expt_tpch_5cycle}
	\end{subfigure}
	\caption{
		Setting 2: TPC-H data with FDs.
		}
	\label{fig:Fig_expt_tpch}
\end{figure}%
\begin{figure}
	\begin{subfigure}{\columnwidth}
         \centering
 		\includegraphics[width=0.62\columnwidth]{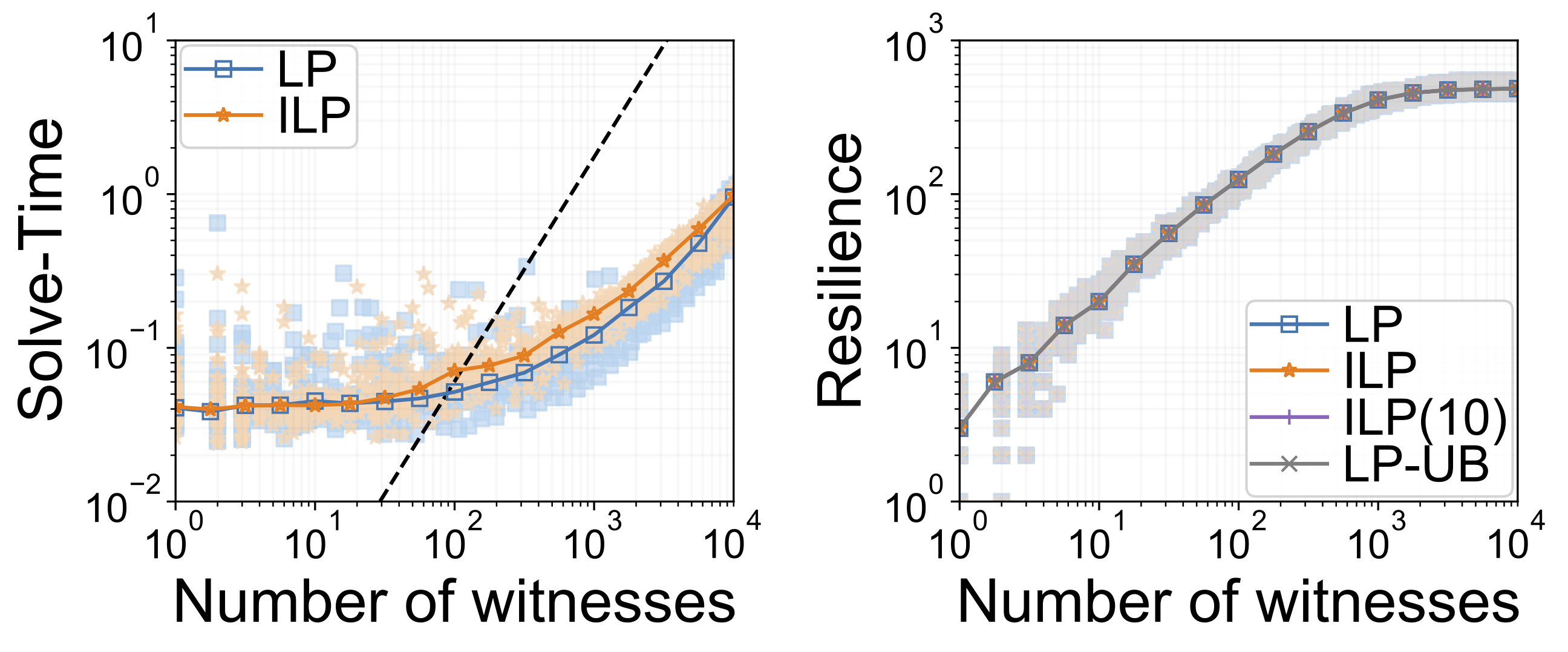}
 		\caption{SJ-Conf query (an easy query) }
 		\label{fig:Fig_expt_sj_conf}
     \end{subfigure}
 	\begin{subfigure}{\columnwidth}
 		\centering
 			\includegraphics[width=0.62\columnwidth]{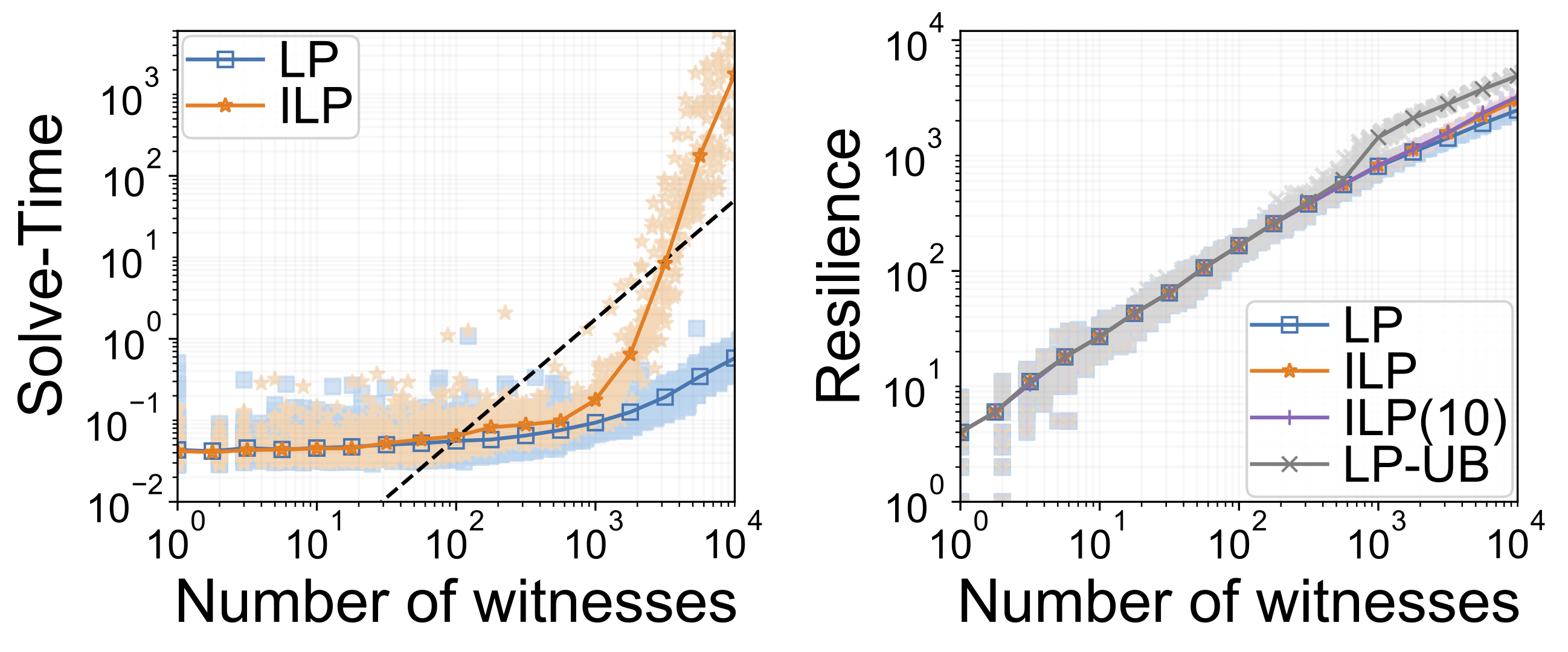}
 		\caption{SJ-Chain query (a hard query)}
 		\label{fig:Fig_expt_sj_chain}
 	\end{subfigure}
 	\caption{
	Setting 3: Queries with self-joins.}
     \label{fig:Fig_Expt_sj_contrast}
  \end{figure}%

\subsection{Experimental Settings}

\introparagraph{Setting 1: Resilience Under Set Semantics}
We consider the 3-star query $\qthreestar \datarule R(x), S(y), T(z), \\W(x,y,z)$ which contains an active triad and is hard (\cref{fig:Fig_Expt_3Star}).
The top plots show the growth of solve-time and resilience 
for increasing instances, while the bottom plots show the growths as a fraction of the optimal.\footnote{The optimal solve-time is $\reslp$ and the optimal resilience is from $\resilp$.}
We see that the solve-time of  $\resilp$ quickly shoots up, while $\reslp$ and the approximations remain $\PTIME$.
The bottom plots show a more zoomed-in look, and we see even in the worst case instances, the approximations are only between $1.1$x to $1.6$x off.

\introparagraph{Setting 2: Responsibility With TPCH Data}
\cref{fig:Fig_expt_tpch} shows results for the $5$-chain query $\qfivechain \datarule$ Customer(custname, custkey), Orders(custkey, orderkey), Lineitem(orderkey, psid), Partsupplier(id, suppkey) and $5$-cycle query $\qfivecycle \datarule $ Customer(custname, custkey), Orders(custkey, orderkey),
Lineitem(orderkey, psid), Partsupplier(id, suppkey), Supplier(suppkey, suppname) over TPC-H data.
While in general $\qfivecycle$ is $\NPC$, a careful reader may notice that all joins have a primary-foreign key dependencies. 
We do not inform our algorithms about these dependencies nor make any changes to accommodate them.
Yet the solver is able to leverage the dependencies \emph{from the data} and $\resilp$ scales in $\PTIME$.
We see that the ILP is faster than the both dedicated flow algorithm and flow approximation.
In both cases, all algorithms (exact and approximate) return the correct responsibility.

\introparagraph{Setting 3: Queries with Self-Joins under Bag Semantics}
\cref{fig:Fig_Expt_sj_contrast} compares two queries with self-joins: 
$\textrm{SJ-conf}\datarule R(x,y),R(x,z),$ $A(x), C(z)$ 
is easy 
and 
$\textrm{SJ-chain}\datarule R(x,y), R(y,z)$ is hard. 
The stark difference in the solve-time growth clearly indicates their theoretical complexity.
While \texttt{LP-UB} increases as the $\textrm{SJ-chain}$ instance grows, it is still far from the theorized $4$-factor worst case bound.
We see that \texttt{ILP-10} is a good indicator for the objective value, even when the ILP takes far longer.

\cref{SEC:APPENDIX:EXPERIMENTS} provides more experimental settings, such as comparing set and bag 
semantics~\cite{makhija2023unifiedarxiv}.

\subsection{Key Takeaways from Experiments}
We summarize the key takeaways from our experiments:

\resultbox{\begin{resultW}(\textbf{Scalability of ILP for PTIME Cases}) 
For easy cases, solving our ILP encoding is in $\PTIME$ and at times even faster than a previously proposed dedicated flow algorithm.
\end{resultW}}

\noindent
We see the scalability of ILP for $\PTIME$ cases in 
\cref{fig:Fig_expt_tpch_5chain,fig:Fig_expt_sj_conf}. 
As expected, solving the ILP formulation takes similar time as LP.
We see that Gurobi can solve responsibility around 12 times faster for a $\PTIME$ query (\cref{fig:Fig_expt_tpch_5chain}) than the previously proposed flow encoding.

\resultbox{\begin{resultW}(\textbf{Correctness of LP for PTIME Cases}) 
Over all experiments, $\reslp= \resilp$ and $\rspmilp = \rspilp$.
\end{resultW}}

\noindent
\cref{fig:Fig_expt_tpch_5chain,fig:Fig_expt_sj_conf} 
corroborate the correctness of the LP relaxation for $\PTIME$ queries, as expected due to the theorems proved in \cref{sec:theoreticalresults}.

\resultbox{\begin{resultW}(\textbf{Scalability of ILP and its Relaxations for Hard Cases}) 
For hard queries, we observe that the time taken by the LP and MILP relaxations grows polynomially, while the time taken by the ILP solution grows exponentially.
However, in practice (and in the absence of ``hardness-creating interactions'' in data)
the ILP can often be solved efficiently.
\end{resultW}}

\noindent
\cref{fig:Fig_Expt_3Star,fig:Fig_expt_tpch_5cycle,fig:Fig_expt_sj_chain} show hard cases.
The difference in solve-time is best seen in \cref{fig:Fig_Expt_3Star,fig:Fig_Expt_sj_contrast}, where the ILP overtakes linear scalability.
However, interestingly some hard queries don't show exponential time complexity, 
and for more complicated queries it actually quite difficult to even synthetically create random data for which solving the ILP shows exponential growth.

\resultbox{\begin{resultW}(\textbf{Approximation quality}) 
\texttt{LP-UB} is better in practice than the worst-case $m$-factor bound.
The flow based approximations give better approximations, but are slower than the LP relaxation.
\end{resultW}}

\noindent
\cref{fig:Fig_Expt_3Star,fig:Fig_expt_sj_chain} 
show that the results from approximation algorithms
are well within theorized bounds and run in $\PTIME$.
All approximations are very close to the exact answer, and we need the $\Delta$ plots in \cref{fig:Fig_Expt_3Star} to see any difference between exact and approximate results.
We observe that in this case \texttt{Flow-CW} performs better than \texttt{Flow-CT} and is faster as well. 
\texttt{LP-UB} is faster than the flow-based approximations but can be worse.
We also see that the LP approximation is worst when the ILP takes much longer than the LP.

\section{Conclusion and Future Work}

This paper presented a novel way of determining the complexity of resilience.
We give a universal encoding as ILP and then investigate when an LP approximation is guaranteed to give an integral solution, thereby proving that modern solvers can return the answer in guaranteed $\PTIME$.
While this approach is known in the optimization literature \cite{schrijver2003combinatorial},
it has so far not been applied as \emph{proof method} to establish dichotomy results in reverse data management.
Since the resulting theory is somewhat simpler and naturally captures all prior known $\PTIME$ cases, 
we believe that this approach will also help in related open problems
for reverse data management, in particular a so far elusive complete dichotomy for resilience of queries with self-joins \cite{DBLP:conf/pods/FreireGIM20}.

\section*{Acknowledgements}

This work was supported in part by the National Science Foundation (NSF) under award numbers IIS-1762268 and IIS-1956096, 
and conducted in part while the authors were visiting the Simons Institute for the Theory of Computing.

\newpage
\bibliographystyle{ACM-Reference-Format}
\bibliography{BIB/resilience.bib}


\begin{thebibliography}{81}


\ifx \showCODEN    \undefined \def \showCODEN     #1{\unskip}     \fi
\ifx \showDOI      \undefined \def \showDOI       #1{#1}\fi
\ifx \showISBNx    \undefined \def \showISBNx     #1{\unskip}     \fi
\ifx \showISBNxiii \undefined \def \showISBNxiii  #1{\unskip}     \fi
\ifx \showISSN     \undefined \def \showISSN      #1{\unskip}     \fi
\ifx \showLCCN     \undefined \def \showLCCN      #1{\unskip}     \fi
\ifx \shownote     \undefined \def \shownote      #1{#1}          \fi
\ifx \showarticletitle \undefined \def \showarticletitle #1{#1}   \fi
\ifx \showURL      \undefined \def \showURL       {\relax}        \fi
\providecommand\bibfield[2]{#2}
\providecommand\bibinfo[2]{#2}
\providecommand\natexlab[1]{#1}
\providecommand\showeprint[2][]{arXiv:#2}

\bibitem[Aardal et~al\mbox{.}(2005)]%
        {aardal2005handbooks}
\bibfield{author}{\bibinfo{person}{Karen Aardal}, \bibinfo{person}{George~L Nemhauser}, {and} \bibinfo{person}{Robert Weismantel}.} \bibinfo{year}{2005}\natexlab{}.
\newblock \bibinfo{booktitle}{\emph{Handbooks in Operations Research and Management Science: Discrete Optimization}}.
\newblock \bibinfo{publisher}{Elsevier}.
\newblock
\urldef\tempurl%
\url{https://doi.org/10.1016/s0927-0507(05)x1200-2}
\showDOI{\tempurl}


\bibitem[Achterberg et~al\mbox{.}(2020)]%
        {achterberg2020presolve}
\bibfield{author}{\bibinfo{person}{Tobias Achterberg}, \bibinfo{person}{Robert~E Bixby}, \bibinfo{person}{Zonghao Gu}, \bibinfo{person}{Edward Rothberg}, {and} \bibinfo{person}{Dieter Weninger}.} \bibinfo{year}{2020}\natexlab{}.
\newblock \showarticletitle{Presolve reductions in mixed integer programming}.
\newblock \bibinfo{journal}{\emph{INFORMS Journal on Computing}} \bibinfo{volume}{32}, \bibinfo{number}{2} (\bibinfo{year}{2020}), \bibinfo{pages}{473--506}.
\newblock
\urldef\tempurl%
\url{https://doi.org/10.1287/ijoc.2018.0857}
\showDOI{\tempurl}


\bibitem[Atserias and Kolaitis(2022)]%
        {atserias2022structure}
\bibfield{author}{\bibinfo{person}{Albert Atserias} {and} \bibinfo{person}{Phokion~G Kolaitis}.} \bibinfo{year}{2022}\natexlab{}.
\newblock \showarticletitle{Structure and complexity of bag consistency}.
\newblock \bibinfo{journal}{\emph{ACM SIGMOD Record}} \bibinfo{volume}{51}, \bibinfo{number}{1} (\bibinfo{year}{2022}), \bibinfo{pages}{78--85}.
\newblock
\urldef\tempurl%
\url{https://doi.org/10.1145/3542700.3542719}
\showDOI{\tempurl}


\bibitem[Beeri et~al\mbox{.}(1983)]%
        {Beeri+83}
\bibfield{author}{\bibinfo{person}{Catriel Beeri}, \bibinfo{person}{Ronald Fagin}, \bibinfo{person}{David Maier}, {and} \bibinfo{person}{Mihalis Yannakakis}.} \bibinfo{year}{1983}\natexlab{}.
\newblock \showarticletitle{On the Desirability of Acyclic Database Schemes}.
\newblock \bibinfo{journal}{\emph{J. ACM}} \bibinfo{volume}{30}, \bibinfo{number}{3} (\bibinfo{date}{July} \bibinfo{year}{1983}), \bibinfo{pages}{479--513}.
\newblock
\showISSN{0004-5411}
\urldef\tempurl%
\url{https://doi.org/10.1145/2402.322389}
\showDOI{\tempurl}


\bibitem[Bertossi(2021)]%
        {bertossi2021specifying}
\bibfield{author}{\bibinfo{person}{Leopoldo Bertossi}.} \bibinfo{year}{2021}\natexlab{}.
\newblock \showarticletitle{Specifying and computing causes for query answers in databases via database repairs and repair-programs}.
\newblock \bibinfo{journal}{\emph{Knowledge and Information Systems}} \bibinfo{volume}{63}, \bibinfo{number}{1} (\bibinfo{year}{2021}), \bibinfo{pages}{199--231}.
\newblock
\urldef\tempurl%
\url{https://doi.org/10.1007/s10115-020-01516-6}
\showDOI{\tempurl}


\bibitem[Bodirsky et~al\mbox{.}(2023)]%
        {bodirsky2023complexity}
\bibfield{author}{\bibinfo{person}{Manuel Bodirsky}, \bibinfo{person}{Žaneta Semanišinová}, {and} \bibinfo{person}{Carsten Lutz}.} \bibinfo{year}{2023}\natexlab{}.
\newblock \showarticletitle{The Complexity of Resilience Problems via Valued Constraint Satisfaction Problems}.
\newblock  (\bibinfo{year}{2023}).
\newblock
\showeprint[arxiv]{2309.15654}~[math.LO]
\urldef\tempurl%
\url{https://arxiv.org/abs/2309.15654}
\showURL{%
\tempurl}


\bibitem[Bollob{\'a}s(1998)]%
        {bollobas1998modern}
\bibfield{author}{\bibinfo{person}{B{\'e}la Bollob{\'a}s}.} \bibinfo{year}{1998}\natexlab{}.
\newblock \bibinfo{booktitle}{\emph{Modern graph theory}}. Vol.~\bibinfo{volume}{184}.
\newblock \bibinfo{publisher}{Springer Science \& Business Media}.
\newblock
\urldef\tempurl%
\url{https://doi.org/10.1007/978-1-4612-0619-4}
\showDOI{\tempurl}


\bibitem[Brucato et~al\mbox{.}(2019)]%
        {brucato2019scalable}
\bibfield{author}{\bibinfo{person}{Matteo Brucato}, \bibinfo{person}{Azza Abouzied}, {and} \bibinfo{person}{Alexandra Meliou}.} \bibinfo{year}{2019}\natexlab{}.
\newblock \showarticletitle{Scalable computation of high-order optimization queries}.
\newblock \bibinfo{journal}{\emph{Commun. ACM}} \bibinfo{volume}{62}, \bibinfo{number}{2} (\bibinfo{year}{2019}), \bibinfo{pages}{108--116}.
\newblock
\urldef\tempurl%
\url{https://doi.org/10.1145/3299881}
\showDOI{\tempurl}


\bibitem[Buneman et~al\mbox{.}(2001)]%
        {DBLP:conf/icdt/BunemanKT01}
\bibfield{author}{\bibinfo{person}{Peter Buneman}, \bibinfo{person}{Sanjeev Khanna}, {and} \bibinfo{person}{Wang~Chiew Tan}.} \bibinfo{year}{2001}\natexlab{}.
\newblock \showarticletitle{Why and Where: A Characterization of Data Provenance}. In \bibinfo{booktitle}{\emph{{ICDT}}}. \bibinfo{pages}{316--330}.
\newblock
\urldef\tempurl%
\url{https://doi.org/10.1007/3-540-44503-x_20}
\showDOI{\tempurl}


\bibitem[Buneman et~al\mbox{.}(2002)]%
        {Buneman:2002}
\bibfield{author}{\bibinfo{person}{Peter Buneman}, \bibinfo{person}{Sanjeev Khanna}, {and} \bibinfo{person}{Wang-Chiew Tan}.} \bibinfo{year}{2002}\natexlab{}.
\newblock \showarticletitle{On Propagation of Deletions and Annotations Through Views}. In \bibinfo{booktitle}{\emph{{PODS}}}. \bibinfo{pages}{150--158}.
\newblock
\showISBNx{1-58113-507-6}
\urldef\tempurl%
\url{https://doi.org/10.1145/543613.543633}
\showDOI{\tempurl}


\bibitem[Buneman and Tan(2007)]%
        {Buneman07}
\bibfield{author}{\bibinfo{person}{Peter Buneman} {and} \bibinfo{person}{Wang-Chiew Tan}.} \bibinfo{year}{2007}\natexlab{}.
\newblock \showarticletitle{Provenance in Databases}. In \bibinfo{booktitle}{\emph{{SIGMOD}}}. \bibinfo{pages}{1171--1173}.
\newblock
\showISBNx{978-1-59593-686-8}
\urldef\tempurl%
\url{https://doi.org/10.1145/1247480.1247646}
\showDOI{\tempurl}


\bibitem[Capelli et~al\mbox{.}(2022)]%
        {capelli2022linear}
\bibfield{author}{\bibinfo{person}{Florent Capelli}, \bibinfo{person}{Nicolas Crosetti}, \bibinfo{person}{Joachim Niehren}, {and} \bibinfo{person}{Jan Ramon}.} \bibinfo{year}{2022}\natexlab{}.
\newblock \showarticletitle{Linear programs with conjunctive queries}.
\newblock  (\bibinfo{year}{2022}).
\newblock
\urldef\tempurl%
\url{https://doi.org/10.4230/LIPIcs.ICDT.2022.5}
\showDOI{\tempurl}


\bibitem[Chandra and Merlin(1977)]%
        {DBLP:conf/stoc/ChandraM77}
\bibfield{author}{\bibinfo{person}{Ashok~K. Chandra} {and} \bibinfo{person}{Philip~M. Merlin}.} \bibinfo{year}{1977}\natexlab{}.
\newblock \showarticletitle{Optimal Implementation of Conjunctive Queries in Relational Data Bases}. In \bibinfo{booktitle}{\emph{{STOC}}}. \bibinfo{pages}{77--90}.
\newblock
\urldef\tempurl%
\url{https://doi.org/10.1145/800105.803397}
\showDOI{\tempurl}


\bibitem[Chaudhuri and Vardi(1993)]%
        {chaudhuri1993optimization}
\bibfield{author}{\bibinfo{person}{Surajit Chaudhuri} {and} \bibinfo{person}{Moshe~Y Vardi}.} \bibinfo{year}{1993}\natexlab{}.
\newblock \showarticletitle{Optimization of real conjunctive queries}. In \bibinfo{booktitle}{\emph{{PODS}}}. \bibinfo{pages}{59--70}.
\newblock
\urldef\tempurl%
\url{https://doi.org/10.1145/153850.153856}
\showDOI{\tempurl}


\bibitem[Cheney et~al\mbox{.}(2009)]%
        {DBLP:journals/ftdb/CheneyCT09}
\bibfield{author}{\bibinfo{person}{James Cheney}, \bibinfo{person}{Laura Chiticariu}, {and} \bibinfo{person}{Wang~Chiew Tan}.} \bibinfo{year}{2009}\natexlab{}.
\newblock \showarticletitle{Provenance in Databases: Why, How, and Where}.
\newblock \bibinfo{journal}{\emph{Foundations and Trends in Databases}} \bibinfo{volume}{1}, \bibinfo{number}{4} (\bibinfo{year}{2009}), \bibinfo{pages}{379--474}.
\newblock
\urldef\tempurl%
\url{https://doi.org/10.1561/9781601982339}
\showDOI{\tempurl}


\bibitem[Chockler and Halpern(2004)]%
        {ChocklerH04}
\bibfield{author}{\bibinfo{person}{Hana Chockler} {and} \bibinfo{person}{Joseph~Y. Halpern}.} \bibinfo{year}{2004}\natexlab{}.
\newblock \showarticletitle{Responsibility and Blame: A Structural-Model Approach}.
\newblock \bibinfo{journal}{\emph{J. Artif. Intell. Res. (JAIR)}}  \bibinfo{volume}{22} (\bibinfo{year}{2004}), \bibinfo{pages}{93--115}.
\newblock
\urldef\tempurl%
\url{https://doi.org/10.1613/jair.1391}
\showDOI{\tempurl}


\bibitem[Cohen et~al\mbox{.}(2021)]%
        {cohen2021solving}
\bibfield{author}{\bibinfo{person}{Michael~B Cohen}, \bibinfo{person}{Yin~Tat Lee}, {and} \bibinfo{person}{Zhao Song}.} \bibinfo{year}{2021}\natexlab{}.
\newblock \showarticletitle{Solving linear programs in the current matrix multiplication time}.
\newblock \bibinfo{journal}{\emph{Journal of the ACM (JACM)}} \bibinfo{volume}{68}, \bibinfo{number}{1} (\bibinfo{year}{2021}), \bibinfo{pages}{1--39}.
\newblock
\urldef\tempurl%
\url{https://doi.org/10.1145/3424305}
\showDOI{\tempurl}


\bibitem[Conforti et~al\mbox{.}(2006)]%
        {conforti2006balanced}
\bibfield{author}{\bibinfo{person}{Michele Conforti}, \bibinfo{person}{G{\'e}rard Cornu{\'e}jols}, {and} \bibinfo{person}{Kristina Vu{\v{s}}kovi{\'c}}.} \bibinfo{year}{2006}\natexlab{}.
\newblock \showarticletitle{Balanced matrices}.
\newblock \bibinfo{journal}{\emph{Discrete Mathematics}} \bibinfo{volume}{306}, \bibinfo{number}{19-20} (\bibinfo{year}{2006}), \bibinfo{pages}{2411--2437}.
\newblock
\urldef\tempurl%
\url{https://doi.org/10.1016/j.disc.2005.12.033}
\showDOI{\tempurl}


\bibitem[Cornu{\'e}jols and Guenin(2002)]%
        {cornuejols2002ideal}
\bibfield{author}{\bibinfo{person}{G{\'e}rard Cornu{\'e}jols} {and} \bibinfo{person}{Bertrand Guenin}.} \bibinfo{year}{2002}\natexlab{}.
\newblock \showarticletitle{Ideal clutters}.
\newblock \bibinfo{journal}{\emph{Discrete Applied Mathematics}} \bibinfo{volume}{123}, \bibinfo{number}{1-3} (\bibinfo{year}{2002}), \bibinfo{pages}{303--338}.
\newblock
\urldef\tempurl%
\url{https://doi.org/10.1016/S0166-218X(01)00344-4}
\showDOI{\tempurl}


\bibitem[Crama and Hammer(2011)]%
        {CramaHammer2010:BooleanFunctions}
\bibfield{author}{\bibinfo{person}{Yves Crama} {and} \bibinfo{person}{Peter~L. Hammer}.} \bibinfo{year}{2011}\natexlab{}.
\newblock \bibinfo{booktitle}{\emph{Boolean Functions: Theory, Algorithms, and Applications}}.
\newblock \bibinfo{publisher}{Cambridge University Press}.
\newblock
\urldef\tempurl%
\url{https://doi.org/10.1017/cbo9780511852008.003}
\showDOI{\tempurl}


\bibitem[Dalvi and Suciu(2007)]%
        {DBLP:journals/vldb/DalviS07}
\bibfield{author}{\bibinfo{person}{Nilesh~N. Dalvi} {and} \bibinfo{person}{Dan Suciu}.} \bibinfo{year}{2007}\natexlab{}.
\newblock \showarticletitle{Efficient query evaluation on probabilistic databases}.
\newblock \bibinfo{journal}{\emph{{VLDB} J.}} \bibinfo{volume}{16}, \bibinfo{number}{4} (\bibinfo{year}{2007}), \bibinfo{pages}{523--544}.
\newblock
\urldef\tempurl%
\url{https://doi.org/10.1007/s00778-006-0004-3}
\showDOI{\tempurl}


\bibitem[Dalvi and Suciu(2012)]%
        {DBLP:journals/jacm/DalviS12}
\bibfield{author}{\bibinfo{person}{Nilesh~N. Dalvi} {and} \bibinfo{person}{Dan Suciu}.} \bibinfo{year}{2012}\natexlab{}.
\newblock \showarticletitle{The dichotomy of probabilistic inference for unions of conjunctive queries}.
\newblock \bibinfo{journal}{\emph{J. ACM}} \bibinfo{volume}{59}, \bibinfo{number}{6} (\bibinfo{year}{2012}), \bibinfo{pages}{30}.
\newblock
\urldef\tempurl%
\url{https://doi.org/10.1145/2395116.2395119}
\showDOI{\tempurl}


\bibitem[Dantsin et~al\mbox{.}(2001)]%
        {10.1145/502807.502810}
\bibfield{author}{\bibinfo{person}{Evgeny Dantsin}, \bibinfo{person}{Thomas Eiter}, \bibinfo{person}{Georg Gottlob}, {and} \bibinfo{person}{Andrei Voronkov}.} \bibinfo{year}{2001}\natexlab{}.
\newblock \showarticletitle{Complexity and Expressive Power of Logic Programming}.
\newblock \bibinfo{journal}{\emph{ACM Comput. Surv.}} \bibinfo{volume}{33}, \bibinfo{number}{3} (\bibinfo{year}{2001}), \bibinfo{pages}{374--425}.
\newblock
\showISSN{0360-0300}
\urldef\tempurl%
\url{https://doi.org/10.1145/502807.502810}
\showDOI{\tempurl}


\bibitem[Davis et~al\mbox{.}(1962)]%
        {DPLL}
\bibfield{author}{\bibinfo{person}{Martin Davis}, \bibinfo{person}{George Logemann}, {and} \bibinfo{person}{Donald Loveland}.} \bibinfo{year}{1962}\natexlab{}.
\newblock \showarticletitle{A Machine Program for Theorem-Proving}.
\newblock \bibinfo{journal}{\emph{Commun. ACM}} \bibinfo{volume}{5}, \bibinfo{number}{7} (\bibinfo{date}{jul} \bibinfo{year}{1962}), \bibinfo{pages}{394--397}.
\newblock
\showISSN{0001-0782}
\urldef\tempurl%
\url{https://doi.org/10.1145/368273.368557}
\showDOI{\tempurl}


\bibitem[Dayal and Bernstein(1982)]%
        {Dayal82}
\bibfield{author}{\bibinfo{person}{Umeshwar Dayal} {and} \bibinfo{person}{Philip~A. Bernstein}.} \bibinfo{year}{1982}\natexlab{}.
\newblock \showarticletitle{On the Correct Translation of Update Operations on Relational Views}.
\newblock \bibinfo{journal}{\emph{ACM TODS}} \bibinfo{volume}{7}, \bibinfo{number}{3} (\bibinfo{year}{1982}), \bibinfo{pages}{381--416}.
\newblock
\showISSN{0362-5915}
\urldef\tempurl%
\url{https://doi.org/10.1145/319732.319740}
\showDOI{\tempurl}


\bibitem[Eiter and Gottlob(1993)]%
        {EITER1993231}
\bibfield{author}{\bibinfo{person}{Thomas Eiter} {and} \bibinfo{person}{Georg Gottlob}.} \bibinfo{year}{1993}\natexlab{}.
\newblock \showarticletitle{Propositional circumscription and extended closed-world reasoning are $\Pi^P_2$-complete}.
\newblock \bibinfo{journal}{\emph{Theoretical Computer Science}} \bibinfo{volume}{114}, \bibinfo{number}{2} (\bibinfo{year}{1993}), \bibinfo{pages}{231--245}.
\newblock
\showISSN{0304-3975}
\urldef\tempurl%
\url{https://doi.org/10.1016/0304-3975(93)90073-3}
\showDOI{\tempurl}


\bibitem[Eiter and Gottlob(1995)]%
        {eiter1995computational}
\bibfield{author}{\bibinfo{person}{Thomas Eiter} {and} \bibinfo{person}{Georg Gottlob}.} \bibinfo{year}{1995}\natexlab{}.
\newblock \showarticletitle{On the computational cost of disjunctive logic programming: Propositional case}.
\newblock \bibinfo{journal}{\emph{Annals of Mathematics and Artificial Intelligence}}  \bibinfo{volume}{15} (\bibinfo{year}{1995}), \bibinfo{pages}{289--323}.
\newblock
\urldef\tempurl%
\url{https://doi.org/10.1007/bf01536399}
\showDOI{\tempurl}


\bibitem[Eiter et~al\mbox{.}(1997)]%
        {10.1145/261124.261126}
\bibfield{author}{\bibinfo{person}{Thomas Eiter}, \bibinfo{person}{Georg Gottlob}, {and} \bibinfo{person}{Heikki Mannila}.} \bibinfo{year}{1997}\natexlab{}.
\newblock \showarticletitle{Disjunctive Datalog}.
\newblock \bibinfo{journal}{\emph{ACM Trans. Database Syst.}} \bibinfo{volume}{22}, \bibinfo{number}{3} (\bibinfo{year}{1997}), \bibinfo{pages}{364--418}.
\newblock
\showISSN{0362-5915}
\urldef\tempurl%
\url{https://doi.org/10.1145/261124.261126}
\showDOI{\tempurl}


\bibitem[Eiter et~al\mbox{.}(2009)]%
        {eiter2009answer}
\bibfield{author}{\bibinfo{person}{Thomas Eiter}, \bibinfo{person}{Giovambattista Ianni}, {and} \bibinfo{person}{Thomas Krennwallner}.} \bibinfo{year}{2009}\natexlab{}.
\newblock \bibinfo{booktitle}{\emph{Answer set programming: A primer}}.
\newblock \bibinfo{publisher}{Springer}.
\newblock
\urldef\tempurl%
\url{https://doi.org/10.1007/978-3-642-03754-2_2}
\showDOI{\tempurl}


\bibitem[Eiter and Polleres(2006)]%
        {eiter2006towards}
\bibfield{author}{\bibinfo{person}{Thomas Eiter} {and} \bibinfo{person}{Axel Polleres}.} \bibinfo{year}{2006}\natexlab{}.
\newblock \showarticletitle{Towards automated integration of guess and check programs in answer set programming: a meta-interpreter and applications}.
\newblock \bibinfo{journal}{\emph{Theory and Practice of Logic Programming}} \bibinfo{volume}{6}, \bibinfo{number}{1-2} (\bibinfo{year}{2006}), \bibinfo{pages}{23--60}.
\newblock
\urldef\tempurl%
\url{https://doi.org/10.1017/s1471068405002577}
\showDOI{\tempurl}


\bibitem[Ford and Fulkerson(1956)]%
        {ford1956maximal}
\bibfield{author}{\bibinfo{person}{Lester~Randolph Ford} {and} \bibinfo{person}{Delbert~R Fulkerson}.} \bibinfo{year}{1956}\natexlab{}.
\newblock \showarticletitle{Maximal flow through a network}.
\newblock \bibinfo{journal}{\emph{Canadian journal of Mathematics}}  \bibinfo{volume}{8} (\bibinfo{year}{1956}), \bibinfo{pages}{399--404}.
\newblock
\urldef\tempurl%
\url{https://doi.org/10.4153/cjm-1956-045-5}
\showDOI{\tempurl}


\bibitem[Freire et~al\mbox{.}(2015)]%
        {DBLP:journals/pvldb/FreireGIM15}
\bibfield{author}{\bibinfo{person}{Cibele Freire}, \bibinfo{person}{Wolfgang Gatterbauer}, \bibinfo{person}{Neil Immerman}, {and} \bibinfo{person}{Alexandra Meliou}.} \bibinfo{year}{2015}\natexlab{}.
\newblock \showarticletitle{The Complexity of Resilience and Responsibility for Self-Join-Free Conjunctive Queries}.
\newblock \bibinfo{journal}{\emph{{PVLDB}}} \bibinfo{volume}{9}, \bibinfo{number}{3} (\bibinfo{year}{2015}), \bibinfo{pages}{180--191}.
\newblock
\urldef\tempurl%
\url{http://www.vldb.org/pvldb/vol9/p180-freire.pdf}
\showURL{%
\tempurl}


\bibitem[Freire et~al\mbox{.}(2020)]%
        {DBLP:conf/pods/FreireGIM20}
\bibfield{author}{\bibinfo{person}{Cibele Freire}, \bibinfo{person}{Wolfgang Gatterbauer}, \bibinfo{person}{Neil Immerman}, {and} \bibinfo{person}{Alexandra Meliou}.} \bibinfo{year}{2020}\natexlab{}.
\newblock \showarticletitle{New Results for the Complexity of Resilience for Binary Conjunctive Queries with Self-Joins}. In \bibinfo{booktitle}{\emph{{PODS}}}. \bibinfo{pages}{271--284}.
\newblock
\urldef\tempurl%
\url{https://doi.org/10.1145/3375395.3387647}
\showDOI{\tempurl}


\bibitem[Galhotra et~al\mbox{.}(2017)]%
        {galhotra2017fairness}
\bibfield{author}{\bibinfo{person}{Sainyam Galhotra}, \bibinfo{person}{Yuriy Brun}, {and} \bibinfo{person}{Alexandra Meliou}.} \bibinfo{year}{2017}\natexlab{}.
\newblock \showarticletitle{Fairness testing: testing software for discrimination}. In \bibinfo{booktitle}{\emph{Proceedings of the 2017 11th Joint meeting on foundations of software engineering}}. \bibinfo{pages}{498--510}.
\newblock
\urldef\tempurl%
\url{https://doi.org/10.1145/3106237.3106277}
\showDOI{\tempurl}


\bibitem[Gatterbauer and Suciu(2017)]%
        {DBLP:journals/vldb/GatterbauerS17}
\bibfield{author}{\bibinfo{person}{Wolfgang Gatterbauer} {and} \bibinfo{person}{Dan Suciu}.} \bibinfo{year}{2017}\natexlab{}.
\newblock \showarticletitle{Dissociation and propagation for approximate lifted inference with standard relational database management systems}.
\newblock \bibinfo{journal}{\emph{{VLDB} J.}} \bibinfo{volume}{26}, \bibinfo{number}{1} (\bibinfo{year}{2017}), \bibinfo{pages}{5--30}.
\newblock
\urldef\tempurl%
\url{https://doi.org/10.1007/s00778-016-0434-5}
\showDOI{\tempurl}


\bibitem[Gebser et~al\mbox{.}(2011)]%
        {gebser2011potassco}
\bibfield{author}{\bibinfo{person}{Martin Gebser}, \bibinfo{person}{Benjamin Kaufmann}, \bibinfo{person}{Roland Kaminski}, \bibinfo{person}{Max Ostrowski}, \bibinfo{person}{Torsten Schaub}, {and} \bibinfo{person}{Marius Schneider}.} \bibinfo{year}{2011}\natexlab{}.
\newblock \showarticletitle{Potassco: The Potsdam answer set solving collection}.
\newblock \bibinfo{journal}{\emph{Ai Communications}} \bibinfo{volume}{24}, \bibinfo{number}{2} (\bibinfo{year}{2011}), \bibinfo{pages}{107--124}.
\newblock
\urldef\tempurl%
\url{https://doi.org/10.3233/aic-2011-0491}
\showDOI{\tempurl}


\bibitem[Gelfond and Kahl(2014)]%
        {gelfond2014knowledge}
\bibfield{author}{\bibinfo{person}{Michael Gelfond} {and} \bibinfo{person}{Yulia Kahl}.} \bibinfo{year}{2014}\natexlab{}.
\newblock \bibinfo{booktitle}{\emph{Knowledge representation, reasoning, and the design of intelligent agents: The answer-set programming approach}}.
\newblock \bibinfo{publisher}{Cambridge University Press}.
\newblock
\urldef\tempurl%
\url{https://doi.org/10.1017/cbo9781139342124}
\showDOI{\tempurl}


\bibitem[Glavic et~al\mbox{.}(2021)]%
        {glavic2021trends}
\bibfield{author}{\bibinfo{person}{Boris Glavic}, \bibinfo{person}{Alexandra Meliou}, {and} \bibinfo{person}{Sudeepa Roy}.} \bibinfo{year}{2021}\natexlab{}.
\newblock \showarticletitle{Trends in explanations: Understanding and debugging data-driven systems}.
\newblock \bibinfo{journal}{\emph{Foundations and Trends in Databases}} \bibinfo{volume}{11}, \bibinfo{number}{3} (\bibinfo{year}{2021}).
\newblock
\urldef\tempurl%
\url{https://doi.org/10.1561/9781680838817}
\showDOI{\tempurl}


\bibitem[Golumbic and Gurvich(2011)]%
        {GolumbicGurvich2010:ReadOnceFunctions}
\bibfield{author}{\bibinfo{person}{Martin~Charles Golumbic} {and} \bibinfo{person}{Vladimir Gurvich}.} \bibinfo{year}{2011}\natexlab{}.
\newblock \bibinfo{booktitle}{\emph{Read-once functions}}.
\newblock \bibinfo{publisher}{Cambridge University Press}, Chapter~10.
\newblock
\urldef\tempurl%
\url{https://doi.org/10.1017/cbo9780511852008.011}
\showDOI{\tempurl}


\bibitem[Golumbic et~al\mbox{.}(2006)]%
        {DBLP:journals/dam/GolumbicMR06}
\bibfield{author}{\bibinfo{person}{Martin~Charles Golumbic}, \bibinfo{person}{Aviad Mintz}, {and} \bibinfo{person}{Udi Rotics}.} \bibinfo{year}{2006}\natexlab{}.
\newblock \showarticletitle{Factoring and recognition of read-once functions using cographs and normality and the readability of functions associated with partial k-trees}.
\newblock \bibinfo{journal}{\emph{Discrete Applied Mathematics}} \bibinfo{volume}{154}, \bibinfo{number}{10} (\bibinfo{year}{2006}), \bibinfo{pages}{1465--1477}.
\newblock


\bibitem[Gr{\"o}tschel et~al\mbox{.}(1993)]%
        {grotschel1993ellipsoid}
\bibfield{author}{\bibinfo{person}{Martin Gr{\"o}tschel}, \bibinfo{person}{L{\'a}szl{\'o} Lov{\'a}sz}, \bibinfo{person}{Alexander Schrijver}, \bibinfo{person}{Martin Gr{\"o}tschel}, \bibinfo{person}{L{\'a}szl{\'o} Lov{\'a}sz}, {and} \bibinfo{person}{Alexander Schrijver}.} \bibinfo{year}{1993}\natexlab{}.
\newblock \showarticletitle{The ellipsoid method}.
\newblock \bibinfo{journal}{\emph{Geometric Algorithms and Combinatorial Optimization}} (\bibinfo{year}{1993}), \bibinfo{pages}{64--101}.
\newblock
\urldef\tempurl%
\url{https://doi.org/10.1007/978-3-642-78240-4_4}
\showDOI{\tempurl}


\bibitem[Gurobi~Optimization(2021)]%
        {gurobi_working}
\bibfield{author}{\bibinfo{person}{LLC Gurobi~Optimization}.} \bibinfo{year}{2021}\natexlab{}.
\newblock \bibinfo{title}{Mixed-Integer Programming (MIP) -- A Primer on the Basics}.
\newblock
\newblock
\urldef\tempurl%
\url{https://www.gurobi.com/resource/mip-basics/}
\showURL{%
\tempurl}


\bibitem[Gurobi~Optimization(2022a)]%
        {gurobinum}
\bibfield{author}{\bibinfo{person}{LLC Gurobi~Optimization}.} \bibinfo{year}{2022}\natexlab{a}.
\newblock \bibinfo{title}{Gurobi Guidelines For Numerical Issues}.
\newblock
\newblock
\urldef\tempurl%
\url{https://www.gurobi.com/documentation/10.0/refman/guidelines_for_numerical_i.html}
\showURL{%
\tempurl}


\bibitem[Gurobi~Optimization(2022b)]%
        {gurobi}
\bibfield{author}{\bibinfo{person}{LLC Gurobi~Optimization}.} \bibinfo{year}{2022}\natexlab{b}.
\newblock \bibinfo{title}{Gurobi Optimizer Reference Manual}.
\newblock
\newblock
\urldef\tempurl%
\url{http://www.gurobi.com}
\showURL{%
\tempurl}


\bibitem[Halpern and Pearl(2005a)]%
        {HalpernPearl:Cause2005}
\bibfield{author}{\bibinfo{person}{Joseph~Y. Halpern} {and} \bibinfo{person}{Judea Pearl}.} \bibinfo{year}{2005}\natexlab{a}.
\newblock \showarticletitle{Causes and Explanations: A structural-model Approach. {P}art {I}: Causes}.
\newblock \bibinfo{journal}{\emph{Brit.\ J.\ Phil.\ Sci.}}  \bibinfo{volume}{56} (\bibinfo{year}{2005}), \bibinfo{pages}{843--887}.
\newblock
\urldef\tempurl%
\url{https://doi.org/10.1093/bjps/axi147}
\showDOI{\tempurl}


\bibitem[Halpern and Pearl(2005b)]%
        {HalpernPearl:Explanations2005}
\bibfield{author}{\bibinfo{person}{Joseph~Y. Halpern} {and} \bibinfo{person}{Judea Pearl}.} \bibinfo{year}{2005}\natexlab{b}.
\newblock \showarticletitle{Causes and Explanations: A structural-model Approach. {P}art {II}: Explanations}.
\newblock \bibinfo{journal}{\emph{Brit.\ J.\ Phil.\ Sci.}}  \bibinfo{volume}{56} (\bibinfo{year}{2005}), \bibinfo{pages}{889--911}.
\newblock


\bibitem[Herschel et~al\mbox{.}(2009)]%
        {DBLP:journals/pvldb/HerschelHT09}
\bibfield{author}{\bibinfo{person}{Melanie Herschel}, \bibinfo{person}{Mauricio~A. Hern{\'a}ndez}, {and} \bibinfo{person}{Wang~Chiew Tan}.} \bibinfo{year}{2009}\natexlab{}.
\newblock \showarticletitle{Artemis: A System for Analyzing Missing Answers}.
\newblock \bibinfo{journal}{\emph{PVLDB}} \bibinfo{volume}{2}, \bibinfo{number}{2} (\bibinfo{year}{2009}), \bibinfo{pages}{1550--1553}.
\newblock
\urldef\tempurl%
\url{https://doi.org/10.14778/1687553.1687588}
\showDOI{\tempurl}


\bibitem[Hu et~al\mbox{.}(2020)]%
        {ADP}
\bibfield{author}{\bibinfo{person}{Xiao Hu}, \bibinfo{person}{Shouzhuo Sun}, \bibinfo{person}{Shweta Patwa}, \bibinfo{person}{Debmalya Panigrahi}, {and} \bibinfo{person}{Sudeepa Roy}.} \bibinfo{year}{2020}\natexlab{}.
\newblock \showarticletitle{Aggregated Deletion Propagation for Counting Conjunctive Query Answers}.
\newblock \bibinfo{journal}{\emph{{PVLDB}}} \bibinfo{volume}{14}, \bibinfo{number}{2} (\bibinfo{year}{2020}), \bibinfo{pages}{228--240}.
\newblock
\showISSN{2150-8097}
\urldef\tempurl%
\url{https://doi.org/10.14778/3425879.3425892}
\showDOI{\tempurl}


\bibitem[Huang et~al\mbox{.}(2008)]%
        {DBLP:journals/pvldb/HuangCDN08}
\bibfield{author}{\bibinfo{person}{Jiansheng Huang}, \bibinfo{person}{Ting Chen}, \bibinfo{person}{AnHai Doan}, {and} \bibinfo{person}{Jeffrey~F. Naughton}.} \bibinfo{year}{2008}\natexlab{}.
\newblock \showarticletitle{On the provenance of non-answers to queries over extracted data}.
\newblock \bibinfo{journal}{\emph{PVLDB}} \bibinfo{volume}{1}, \bibinfo{number}{1} (\bibinfo{year}{2008}), \bibinfo{pages}{736--747}.
\newblock
\urldef\tempurl%
\url{https://doi.org/10.14778/1453856.1453936}
\showDOI{\tempurl}


\bibitem[Karp(1972)]%
        {karp1972reducibility}
\bibfield{author}{\bibinfo{person}{Richard~M Karp}.} \bibinfo{year}{1972}\natexlab{}.
\newblock \showarticletitle{Reducibility among combinatorial problems}.
\newblock In \bibinfo{booktitle}{\emph{Complexity of computer computations}}. \bibinfo{publisher}{Springer}, \bibinfo{pages}{85--103}.
\newblock
\urldef\tempurl%
\url{https://doi.org/10.1007/978-1-4684-2001-2_9}
\showDOI{\tempurl}


\bibitem[Khamis et~al\mbox{.}(2021)]%
        {10.1145/3472391}
\bibfield{author}{\bibinfo{person}{Mahmoud~Abo Khamis}, \bibinfo{person}{Phokion~G. Kolaitis}, \bibinfo{person}{Hung~Q. Ngo}, {and} \bibinfo{person}{Dan Suciu}.} \bibinfo{year}{2021}\natexlab{}.
\newblock \showarticletitle{Bag Query Containment and Information Theory}.
\newblock \bibinfo{journal}{\emph{{ACM TODS}}} \bibinfo{volume}{46}, \bibinfo{number}{3} (\bibinfo{year}{2021}).
\newblock
\showISSN{0362-5915}
\urldef\tempurl%
\url{https://doi.org/10.1145/3472391}
\showDOI{\tempurl}


\bibitem[Kolahi(2009)]%
        {Kolahi2009}
\bibfield{author}{\bibinfo{person}{Solmaz Kolahi}.} \bibinfo{year}{2009}\natexlab{}.
\newblock \bibinfo{booktitle}{\emph{Functional Dependency (Encyclopedia of Database Systems)}}.
\newblock \bibinfo{publisher}{Springer}, \bibinfo{pages}{1200--1201}.
\newblock
\showISBNx{978-0-387-39940-9}
\urldef\tempurl%
\url{https://doi.org/10.1007/978-0-387-39940-9_1247}
\showDOI{\tempurl}


\bibitem[Kolmogorov et~al\mbox{.}(2017)]%
        {kolmogorov2017complexity}
\bibfield{author}{\bibinfo{person}{Vladimir Kolmogorov}, \bibinfo{person}{Andrei Krokhin}, {and} \bibinfo{person}{Michal Rol{\'\i}nek}.} \bibinfo{year}{2017}\natexlab{}.
\newblock \showarticletitle{The complexity of general-valued CSPs}.
\newblock \bibinfo{journal}{\emph{SIAM J. Comput.}} \bibinfo{volume}{46}, \bibinfo{number}{3} (\bibinfo{year}{2017}), \bibinfo{pages}{1087--1110}.
\newblock
\urldef\tempurl%
\url{https://doi.org/10.1109/focs.2015.80}
\showDOI{\tempurl}


\bibitem[Konstantinidis and Mogavero(2019)]%
        {DBLP:conf/pods/KonstantinidisM19}
\bibfield{author}{\bibinfo{person}{George Konstantinidis} {and} \bibinfo{person}{Fabio Mogavero}.} \bibinfo{year}{2019}\natexlab{}.
\newblock \showarticletitle{Attacking Diophantus: Solving a Special Case of Bag Containment}. In \bibinfo{booktitle}{\emph{{PODS}}}. \bibinfo{pages}{399--413}.
\newblock
\urldef\tempurl%
\url{https://doi.org/10.1145/3294052.3319689}
\showDOI{\tempurl}


\bibitem[Lau et~al\mbox{.}(2011)]%
        {lau2011iterative}
\bibfield{author}{\bibinfo{person}{Lap~Chi Lau}, \bibinfo{person}{Ramamoorthi Ravi}, {and} \bibinfo{person}{Mohit Singh}.} \bibinfo{year}{2011}\natexlab{}.
\newblock \bibinfo{booktitle}{\emph{Iterative methods in combinatorial optimization}}. Vol.~\bibinfo{volume}{46}.
\newblock \bibinfo{publisher}{Cambridge University Press}.
\newblock
\urldef\tempurl%
\url{https://doi.org/10.1017/cbo9780511977152}
\showDOI{\tempurl}


\bibitem[Lim et~al\mbox{.}(2009)]%
        {DBLP:conf/chi/LimDA09}
\bibfield{author}{\bibinfo{person}{Brian~Y. Lim}, \bibinfo{person}{Anind~K. Dey}, {and} \bibinfo{person}{Daniel Avrahami}.} \bibinfo{year}{2009}\natexlab{}.
\newblock \showarticletitle{\emph{Why} and \emph{why not} explanations improve the intelligibility of context-aware intelligent systems}. In \bibinfo{booktitle}{\emph{CHI}}. \bibinfo{pages}{2119--2128}.
\newblock
\urldef\tempurl%
\url{http://doi.acm.org/10.1145/1518701.1519023}
\showURL{%
\tempurl}


\bibitem[Makhija and Gatterbauer(2023a)]%
        {resiliencecode}
\bibfield{author}{\bibinfo{person}{Neha Makhija} {and} \bibinfo{person}{Wolfgang Gatterbauer}.} \bibinfo{year}{2023}\natexlab{a}.
\newblock \bibinfo{title}{A Unified Approach for Resilience and Causal Responsibility: Code and Experiments}.
\newblock
\newblock
\urldef\tempurl%
\url{https://github.com/northeastern-datalab/resilience-responsibility-ilp/}
\showURL{%
\tempurl}


\bibitem[Makhija and Gatterbauer(2023b)]%
        {makhija2023unifiedarxiv}
\bibfield{author}{\bibinfo{person}{Neha Makhija} {and} \bibinfo{person}{Wolfgang Gatterbauer}.} \bibinfo{year}{2023}\natexlab{b}.
\newblock \showarticletitle{A Unified Approach for Resilience and Causal Responsibility with Integer Linear Programming (ILP) and LP Relaxations}.
\newblock  (\bibinfo{year}{2023}).
\newblock
\showeprint[arxiv]{2212.08898}~[cs.DB]
\urldef\tempurl%
\url{https://arxiv.org/abs/2212.08898}
\showURL{%
\tempurl}


\bibitem[Meliou et~al\mbox{.}(2010a)]%
        {MeliouGHKMS10}
\bibfield{author}{\bibinfo{person}{Alexandra Meliou}, \bibinfo{person}{Wolfgang Gatterbauer}, \bibinfo{person}{Joseph~Y. Halpern}, \bibinfo{person}{Christoph Koch}, \bibinfo{person}{Katherine~F. Moore}, {and} \bibinfo{person}{Dan Suciu}.} \bibinfo{year}{2010}\natexlab{a}.
\newblock \showarticletitle{Causality in Databases}.
\newblock \bibinfo{journal}{\emph{IEEE Data Eng. Bull.}} \bibinfo{volume}{33}, \bibinfo{number}{3} (\bibinfo{year}{2010}), \bibinfo{pages}{59--67}.
\newblock
\urldef\tempurl%
\url{http://sites.computer.org/debull/A10sept/suciu.pdf}
\showURL{%
\tempurl}


\bibitem[Meliou et~al\mbox{.}(2009)]%
        {DBLP:journals/corr/abs-0912-5340}
\bibfield{author}{\bibinfo{person}{Alexandra Meliou}, \bibinfo{person}{Wolfgang Gatterbauer}, \bibinfo{person}{Katherine~F. Moore}, {and} \bibinfo{person}{Dan Suciu}.} \bibinfo{year}{2009}\natexlab{}.
\newblock \showarticletitle{{W}hy so? or {W}hy no? {F}unctional Causality for Explaining Query Answers}, In \bibinfo{booktitle}{4th International Workshop on Management of Uncertain Data (MUD)}.
\newblock \bibinfo{journal}{\emph{CoRR}}, \bibinfo{pages}{3--17}.
\newblock
\urldef\tempurl%
\url{http://arxiv.org/abs/0912.5340}
\showURL{%
\tempurl}


\bibitem[Meliou et~al\mbox{.}(2010b)]%
        {DBLP:journals/pvldb/MeliouGMS11}
\bibfield{author}{\bibinfo{person}{Alexandra Meliou}, \bibinfo{person}{Wolfgang Gatterbauer}, \bibinfo{person}{Katherine~F. Moore}, {and} \bibinfo{person}{Dan Suciu}.} \bibinfo{year}{2010}\natexlab{b}.
\newblock \showarticletitle{The Complexity of Causality and Responsibility for Query Answers and non-Answers}.
\newblock \bibinfo{journal}{\emph{{PVLDB}}} \bibinfo{volume}{4}, \bibinfo{number}{1} (\bibinfo{year}{2010}), \bibinfo{pages}{34--45}.
\newblock
\urldef\tempurl%
\url{http://www.vldb.org/pvldb/vol4/p34-meliou.pdf}
\showURL{%
\tempurl}


\bibitem[Meliou et~al\mbox{.}(2011)]%
        {DBLP:journals/pvldb/MeliouGS11}
\bibfield{author}{\bibinfo{person}{Alexandra Meliou}, \bibinfo{person}{Wolfgang Gatterbauer}, {and} \bibinfo{person}{Dan Suciu}.} \bibinfo{year}{2011}\natexlab{}.
\newblock \showarticletitle{Reverse Data Management}.
\newblock \bibinfo{journal}{\emph{{PVLDB}}} \bibinfo{volume}{4}, \bibinfo{number}{12} (\bibinfo{year}{2011}), \bibinfo{pages}{1490--1493}.
\newblock
\urldef\tempurl%
\url{http://www.vldb.org/pvldb/vol4/p1490-meliou.pdf}
\showURL{%
\tempurl}


\bibitem[Meliou and Suciu(2012)]%
        {meliou2012tiresias}
\bibfield{author}{\bibinfo{person}{Alexandra Meliou} {and} \bibinfo{person}{Dan Suciu}.} \bibinfo{year}{2012}\natexlab{}.
\newblock \showarticletitle{Tiresias: the database oracle for how-to queries}. In \bibinfo{booktitle}{\emph{Proceedings of the 2012 ACM SIGMOD International Conference on Management of Data}}. \bibinfo{pages}{337--348}.
\newblock
\urldef\tempurl%
\url{https://doi.org/doi/10.1145/2213836.2213875}
\showDOI{\tempurl}


\bibitem[Mitchell et~al\mbox{.}(2011)]%
        {mitchell2011pulp}
\bibfield{author}{\bibinfo{person}{Stuart Mitchell}, \bibinfo{person}{Michael OSullivan}, {and} \bibinfo{person}{Iain Dunning}.} \bibinfo{year}{2011}\natexlab{}.
\newblock \showarticletitle{PuLP: a linear programming toolkit for python}.
\newblock \bibinfo{journal}{\emph{The University of Auckland, Auckland, New Zealand}}  \bibinfo{volume}{65} (\bibinfo{year}{2011}).
\newblock
\urldef\tempurl%
\url{https://optimization-online.org/?p=11731}
\showURL{%
\tempurl}


\bibitem[Potassco(2022)]%
        {clingo}
\bibfield{author}{\bibinfo{person}{the Potsdam Answer Set Solving~Collection Potassco}.} \bibinfo{year}{2022}\natexlab{}.
\newblock \bibinfo{title}{clingo}.
\newblock
\newblock
\urldef\tempurl%
\url{https://potassco.org/clingo/}
\showURL{%
\tempurl}


\bibitem[Pradhan et~al\mbox{.}(2022)]%
        {pradhan2021interpretable}
\bibfield{author}{\bibinfo{person}{Romila Pradhan}, \bibinfo{person}{Jiongli Zhu}, \bibinfo{person}{Boris Glavic}, {and} \bibinfo{person}{Babak Salimi}.} \bibinfo{year}{2022}\natexlab{}.
\newblock \showarticletitle{Interpretable data-based explanations for fairness debugging}. In \bibinfo{booktitle}{\emph{{SIGMOD}}}. \bibinfo{pages}{247--261}.
\newblock
\urldef\tempurl%
\url{https://doi.org/10.1145/3514221.3517886}
\showDOI{\tempurl}


\bibitem[Przymusinski(1991)]%
        {10.1007/BF03037171}
\bibfield{author}{\bibinfo{person}{Teodor~C. Przymusinski}.} \bibinfo{year}{1991}\natexlab{}.
\newblock \showarticletitle{Stable Semantics for Disjunctive Programs}.
\newblock \bibinfo{journal}{\emph{New Generation Computing}} \bibinfo{volume}{9}, \bibinfo{number}{3--4} (\bibinfo{year}{1991}), \bibinfo{pages}{401--424}.
\newblock
\showISSN{0288-3635}
\urldef\tempurl%
\url{https://doi.org/10.1007/BF03037171}
\showDOI{\tempurl}


\bibitem[Roy and Suciu(2014)]%
        {SudeepaSuciu14}
\bibfield{author}{\bibinfo{person}{Sudeepa Roy} {and} \bibinfo{person}{Dan Suciu}.} \bibinfo{year}{2014}\natexlab{}.
\newblock \showarticletitle{A Formal Approach to Finding Explanations for Database Queries}. In \bibinfo{booktitle}{\emph{SIGMOD}}. \bibinfo{pages}{1579--1590}.
\newblock
\urldef\tempurl%
\url{https://doi.org/10.1145/2588555.2588578}
\showDOI{\tempurl}


\bibitem[Salimi et~al\mbox{.}(2019)]%
        {salimi2019interventional}
\bibfield{author}{\bibinfo{person}{Babak Salimi}, \bibinfo{person}{Luke Rodriguez}, \bibinfo{person}{Bill Howe}, {and} \bibinfo{person}{Dan Suciu}.} \bibinfo{year}{2019}\natexlab{}.
\newblock \showarticletitle{Interventional fairness: Causal database repair for algorithmic fairness}. In \bibinfo{booktitle}{\emph{{SIGMOD}}}. \bibinfo{pages}{793--810}.
\newblock
\urldef\tempurl%
\url{https://doi.org/10.1145/3299869.3319901}
\showDOI{\tempurl}


\bibitem[Schrijver(1998)]%
        {schrijver1998theory}
\bibfield{author}{\bibinfo{person}{Alexander Schrijver}.} \bibinfo{year}{1998}\natexlab{}.
\newblock \bibinfo{booktitle}{\emph{Theory of linear and integer programming}}.
\newblock \bibinfo{publisher}{John Wiley \& Sons}.
\newblock
\urldef\tempurl%
\url{https://doi.org/10.1137/1030065}
\showDOI{\tempurl}


\bibitem[Schrijver(2003)]%
        {schrijver2003combinatorial}
\bibfield{author}{\bibinfo{person}{Alexander Schrijver}.} \bibinfo{year}{2003}\natexlab{}.
\newblock \bibinfo{booktitle}{\emph{Combinatorial optimization: polyhedra and efficiency}}. \bibinfo{series}{Algorithms and Combinatorics}, Vol.~\bibinfo{volume}{24}.
\newblock \bibinfo{publisher}{Springer}.
\newblock
\showISBNx{978-3-540-44389-6}
\urldef\tempurl%
\url{https://doi.org/book/9783540443896}
\showDOI{\tempurl}


\bibitem[Schult and Swart(2008)]%
        {networkx}
\bibfield{author}{\bibinfo{person}{Daniel~A Schult} {and} \bibinfo{person}{P Swart}.} \bibinfo{year}{2008}\natexlab{}.
\newblock \showarticletitle{Exploring network structure, dynamics, and function using NetworkX}. In \bibinfo{booktitle}{\emph{Proceedings of the 7th Python in science conferences (SciPy 2008)}}, Vol.~\bibinfo{volume}{2008}. Pasadena, CA, \bibinfo{pages}{11--16}.
\newblock
\urldef\tempurl%
\url{https://permalink.lanl.gov/object/tr?what=info:lanl-repo/lareport/LA-UR-08-05495}
\showURL{%
\tempurl}


\bibitem[TPC-H(2022)]%
        {tpch}
\bibfield{author}{\bibinfo{person}{TPC-H}.} \bibinfo{year}{2022}\natexlab{}.
\newblock \bibinfo{title}{TPC-H Homepage}.
\newblock
\newblock
\urldef\tempurl%
\url{https://www.tpc.org/tpch/}
\showURL{%
\tempurl}


\bibitem[Ullman(1990)]%
        {UllmanBookPDK}
\bibfield{author}{\bibinfo{person}{Jeffrey~D. Ullman}.} \bibinfo{year}{1990}\natexlab{}.
\newblock \bibinfo{booktitle}{\emph{Principles of Database and Knowledge-Base Systems: Volume II: The New Technologies}}.
\newblock \bibinfo{publisher}{W. H. Freeman \& Co.}, \bibinfo{address}{New York, NY, USA}.
\newblock
\showISBNx{071678162X}


\bibitem[Vardi(1982)]%
        {DBLP:conf/stoc/Vardi82}
\bibfield{author}{\bibinfo{person}{Moshe~Y. Vardi}.} \bibinfo{year}{1982}\natexlab{}.
\newblock \showarticletitle{The Complexity of Relational Query Languages (Extended Abstract)}. In \bibinfo{booktitle}{\emph{STOC}}. \bibinfo{pages}{137--146}.
\newblock
\urldef\tempurl%
\url{https://doi.org/10.1145/800070.802186}
\showDOI{\tempurl}


\bibitem[Vazirani(2001)]%
        {vazirani2001approximation}
\bibfield{author}{\bibinfo{person}{Vijay~V Vazirani}.} \bibinfo{year}{2001}\natexlab{}.
\newblock \bibinfo{booktitle}{\emph{Approximation algorithms}}. Vol.~\bibinfo{volume}{1}.
\newblock \bibinfo{publisher}{Springer}.
\newblock
\urldef\tempurl%
\url{https://doi.org/10.1007/978-3-662-04565-7}
\showDOI{\tempurl}


\bibitem[Wang et~al\mbox{.}(2015)]%
        {wang2015error}
\bibfield{author}{\bibinfo{person}{Xiaolan Wang}, \bibinfo{person}{Mary Feng}, \bibinfo{person}{Yue Wang}, \bibinfo{person}{Xin~Luna Dong}, {and} \bibinfo{person}{Alexandra Meliou}.} \bibinfo{year}{2015}\natexlab{}.
\newblock \showarticletitle{Error diagnosis and data profiling with data x-ray}.
\newblock \bibinfo{journal}{\emph{Proceedings of the VLDB Endowment}} \bibinfo{volume}{8}, \bibinfo{number}{12} (\bibinfo{year}{2015}), \bibinfo{pages}{1984--1987}.
\newblock
\urldef\tempurl%
\url{https://doi.org/10.14778/2824032.2824117}
\showDOI{\tempurl}


\bibitem[Wang et~al\mbox{.}(2017)]%
        {wang2017qfix}
\bibfield{author}{\bibinfo{person}{Xiaolan Wang}, \bibinfo{person}{Alexandra Meliou}, {and} \bibinfo{person}{Eugene Wu}.} \bibinfo{year}{2017}\natexlab{}.
\newblock \showarticletitle{QFix: Diagnosing errors through query histories}. In \bibinfo{booktitle}{\emph{Proceedings of the 2017 ACM International Conference on Management of Data}}. \bibinfo{pages}{1369--1384}.
\newblock
\urldef\tempurl%
\url{https://doi.org/10.1145/3035918.3035925}
\showDOI{\tempurl}


\bibitem[Wu and Madden(2013)]%
        {DBLP:journals/pvldb/0002M13}
\bibfield{author}{\bibinfo{person}{Eugene Wu} {and} \bibinfo{person}{Samuel Madden}.} \bibinfo{year}{2013}\natexlab{}.
\newblock \showarticletitle{Scorpion: Explaining Away Outliers in Aggregate Queries}.
\newblock \bibinfo{journal}{\emph{PVLDB}} \bibinfo{volume}{6}, \bibinfo{number}{8} (\bibinfo{year}{2013}), \bibinfo{pages}{553--564}.
\newblock
\urldef\tempurl%
\url{https://doi.org/10.14778/2536354.2536356}
\showDOI{\tempurl}


\bibitem[Yannakakis(2022)]%
        {yannakakis2022technical}
\bibfield{author}{\bibinfo{person}{Mihalis Yannakakis}.} \bibinfo{year}{2022}\natexlab{}.
\newblock \showarticletitle{Technical Perspective: Structure and Complexity of Bag Consistency}.
\newblock \bibinfo{journal}{\emph{ACM SIGMOD Record}} \bibinfo{volume}{51}, \bibinfo{number}{1} (\bibinfo{year}{2022}), \bibinfo{pages}{77--77}.
\newblock
\urldef\tempurl%
\url{https://doi.org/10.1145/3542700.3542718}
\showDOI{\tempurl}


\bibitem[Youngmann et~al\mbox{.}(2022)]%
        {youngmann2022explaining}
\bibfield{author}{\bibinfo{person}{Brit Youngmann}, \bibinfo{person}{Michael Cafarella}, \bibinfo{person}{Yuval Moskovitch}, {and} \bibinfo{person}{Babak Salimi}.} \bibinfo{year}{2022}\natexlab{}.
\newblock \showarticletitle{On Explaining Confounding Bias}.
\newblock \bibinfo{journal}{\emph{arXiv preprint arXiv:2210.02943}} (\bibinfo{year}{2022}).
\newblock
\urldef\tempurl%
\url{https://arxiv.org/abs/2210.02943}
\showURL{%
\tempurl}


\end{thebibliography}
\received{April 2023}
\received[revised]{July 2023}
\received[accepted]{August 2023}

\appendix
\clearpage
\appendix

\section{Nomenclature}
\label{sec:nomenclature}

The Notation Table (\cref{tbl:nomenclature}) contains common nomenclature, and Query Table (\cref{tbl:example-queries}) lists example queries used through the paper. 

\begin{table}[h]
\centering
\small
\begin{tabularx}{\linewidth}{@{\hspace{0pt}} >{$}l<{$} @{\hspace{2mm}}X@{}} %
\hline
\textrm{Symbol}		& Definition 	\\
\hline
    \hline
	Q			& Conjunctive query	\\
	D			& Database Instance, i.e.\ a set of tables	\\
    W			& Set of witnesses $W = \witnesses(Q,D)$\\
	\vec w  	& Witness \\
	x, y, z		& Query variables \\
    m           & Number of atoms in a CQ \\	
    \var(R)     & Variables in relation $R$ \\
    \res	    & The decision problem of resilience \\
	\rsp	    & The decision problem of responsibility\\
    \res^*      & The optimization problem of resilience \\
	\rsp^*      & The optimization problem of responsibility\\
    \Gamma      & A contingency set for \\
	k		    & Number of tables in Q \\	
    \exoset     & A set of exogenous tuples \\
    \resilp, \reslp & ILP and LP for $\res$\\
    \rspilp, \rsplp & ILP and LP for $\rsp$\\
	\rspmilp	& MILP for $\rsp$			\\
    \textrm{JP} 	& A Join Path \\
    \textrm{IJP}    & Independent Join Path \\
    \mathcal{S}, \mathcal{T} 	& Set of start and terminal endpoints of a JP \\
    \ijpdlp     & A DLP to find IJPs for queries \\
    X[v]        & A variable in an (Integer) Linear Program \\
\hline
\end{tabularx}
\caption{Nomenclature table}
\label{tbl:nomenclature}
\end{table}
\begin{table}[h]
    \centering
    \small
    \begin{tabularx}{\linewidth}{@{\hspace{0pt}} >{$}l<{$} @{\hspace{2mm}}X@{}} %
    \hline
    \textrm{Query}		& Definition 	\\
    \hline
        \hline
        \qtwochain 		& 2-chain query $R(x,y), S(y,z)$ \\
        \qthreechain 	& 3-chain query $R(x,y), S(y,z), T(z,u)$ \\
        \qfourchain 	& 4-chain query $P(u,x),R(x,y), S(y,z), T(z,v)$ \\
        \qfivechain 	& 5-chain query $L(a,u),P(u,x),R(x,y), S(y,z), T(z,v)$ \\
        \qtwostar 		& 2-star query $R(x)S(y),W(x,y)$ \\
        \qthreestar 	& 3-star query $R(x)S(y),T(z)W(x,y,z)$ \\
        \qtriangle 		& Triangle query $R(x,y)S(y,z),T(z,x)$ \\
        \qtriangleunary & Triangle Unary query $A(x)R(x,y)S(y,z),T(z,x)$ \\
        \qtrianglebinary & Triangle Binary query $A(x)R(x,y)S(y,z),T(z,x),B(z)$ \\
        \qsjtwochain & Self-Join 2-chain query $R(x,y)R(y,z)$ \\
        \qsjtwoconf & Self-Join 2-confluence query $A(x)R(x,y)S(z,y),B(z)$ \\
        \qsjzsix & Self-Join z6 query $A(x)R(x,y)R(y,y),R(y,z),C(z)$ \\

        \hline
    \end{tabularx}
    
\caption{Example Queries}
\label{tbl:example-queries}
\end{table}

\section{Real-World examples for Resilience and Causal Responsibility}

In this section, we give example of real world-applications of resilience and responsibility.
\cref{ex:resilience-oscar,ex:causal-responsibility-oscar}
are new while
\cref{ex:resilience-system-migration,ex:causal-responsibility-system-migration} are slightly adapted from
work by Freire et al.~\cite{DBLP:journals/pvldb/FreireGIM15}.

\begin{figure}[]
    \begin{minipage}[t]{\columnwidth}
    \small
    \centering			
    \mbox{
        \begin{tabular}[t]{ >{}c<{} | >{}c<{} >{}c<{} }
        \multicolumn{2}{c}{Oscar} \\
        \hline
        &  person	\\
        \hline
        $o_1$	& Frances McDormand \\	
        \end{tabular}			
    }
    \hspace{5mm}
    \mbox{
        \begin{tabular}[t]{ >{}c<{} | >{}c<{} >{}c<{} >{}c<{} >{}c<{} >{}c<{}}
        \multicolumn{3}{c}{ActsIn} \\
        \hline
        & actor	&  movie \\
        \hline
        $a_1$	& Frances McDormand	& Blood Simple	\\
        $a_2$	& Frances McDormand	& Fargo	\\
        $a_3$	& Frances McDormand	& Raising Arizona	\\
        $a_4$	& Frances McDormand	& Nomadland \\	
        $a_5$	& Helena Bonham Carter & Alice in Wonderland  \\
        $a_6$	& Helena Bonham Carter & The King's Speech	\\ 
        \end{tabular}
    }
    \vspace{5mm}\\
    \mbox{
        \begin{tabular}[t]{ >{}c<{} | >{}c<{} >{}c<{} >{}c<{} >{}c<{} >{}c<{}}
        \multicolumn{3}{c}{DirectedBy} \\
        \hline
        & director	&  movie \\
        \hline
        $d_1$	& Joel Coen & Blood Simple	\\
        $d_2$	& Joel Coen	& Fargo	\\
        $d_3$	& Joel Coen	& Raising Arizona	\\
        $d_4$	& Tim Burton & Alice in Wonderland \\	
        \end{tabular}
    }
    \hspace{5mm}
    \mbox{
            \begin{tabular}[t]{ >{}c<{} | >{}c<{} >{}c<{}}
            \multicolumn{3}{c}{Spouse} \\
            \hline 
            actor & director		\\
            \hline
            $s_1$	& Frances McDormand	& Joel Coen\\	
            $s_2$	& Helena Bonham Carter & Tim Burton \\	
            \end{tabular}			
    }
    \end{minipage}
    \caption{
	\Cref{ex:resilience-oscar,ex:causal-responsibility-oscar}:
	Data for Exploratory Data Analysis
	}
	\label{ex:oscar-data}	
\end{figure}    

\begin{example}[Resilience: Exploratory Data Analysis Example]\label{ex:resilience-oscar}     
    How surprising is it if an Oscar winning actor has acted in a movie directed by their spouse? 
    We can quantify this by calculating the resilience of the query $\qtriangleunary \datarule$ Oscar(actor), ActsIn(actor, movie), DirectedBy(movie, dir), Spouse(actor, dir).
    Finding the resilience does not equate to simply the number of satisfying output rows that must be deleted
    but rather asks for the minimum number of changes in the world needed to have no satisfying output.
    For example, if we do not include the spouse pair $s_1$ of Frances McDormand and Joel Coen (\cref{ex:oscar-data}), 
	the single deletion would take away 3 rows from the output.
    Intuitively, if the resilience is small, there have been a very small number of events that have led to an Oscar winning actor being in a movie directed by their spouse.

    Interestingly, the resilience for this query can be calculated in $\PTIME$ under set semantics, but not bag semantics (such as when accounting for multiple Oscar wins).
    If we now change the query to remove the constraint of the actor having won an Oscar, then finding the resilience of the resulting query  $\qtriangle \datarule$ ActsIn(actor, movie), DirectedBy(movie, dir), Spouse(actor, dir) is $\NPC$!
\end{example}

\begin{example}[Causal Responsibility: Exploratory Data Analysis Example]
    \label{ex:causal-responsibility-oscar}
    Assume we wished to ask: ``What is the responsibility of Frances McDormand's Oscar win towards the output of our query?''
    If this Oscar was solely responsible for the output, it would be a counterfactual cause -- 
    i.e.\ if she had not won, 
    there would be no satisfying output.
    However, this tuple still has ``partial'' responsibility.
    By measuring how far we are from a world where the tuple is counterfactual, we can get a notion of its responsibility to the output.(The responsibility is inversely proportional to the minimum number of tuples to be deleted $|\tau|$ and is given by $1 / (1+|\tau|) $.)
    Interestingly, due to our new fine-grained complexity results, 
    we can find the responsibility of a particular Oscar win in $\PTIME$, but finding the responsibility of a tuple from 
    the ActsIn, DirectedBy or Spouse Table is $\NPC$.
\end{example}

\begin{figure}[]
    \begin{minipage}[t]{\columnwidth}			
    \centering
    \small
    \setlength{\tabcolsep}{0.9mm}
    \mbox{
        \begin{tabular}[t]{ >{}c<{} | >{}c<{} >{}c<{} }
        \multicolumn{3}{c}{Users} \\
        \hline  
        uid & name	\\
        \hline
        $u_1$	& 1	& Alice	\\
        $u_2$	& 2	& Bob	\\
        $u_3$	& 3	& Charlie			
        \end{tabular}			
    }
    \hspace{5mm}
    \mbox{
        \begin{tabular}[t]{ >{}c<{} | >{}c<{} >{}c<{} >{}c<{} >{}c<{} >{}c<{}}
        \multicolumn{4}{c}{AccessLog} \\
        \hline  
        uid	&  type &server\\
        \hline
        $a_1$	& 1	& IMAP	& S	\\
        $a_2$	& 2	& DB	& S	\\
        $a_3$	& 1	& SMTP	& S	\\
        $a_4$	& 1	& DB	& S	\\	
        $a_5$	& 3 & IMAP  & X  \\
        $a_6$	& 3	& DB	& S	\\
        $a_7$	& 2	& SMTP	& X	\\	
        $a_8$	& 1 & DB & T     
        \end{tabular}
    }
    \hspace{5mm}
    \mbox{
            \begin{tabular}[t]{ >{}c<{} | >{}c<{} >{}c<{}}
            \multicolumn{3}{c}{Requests} \\
            \hline	
            & type & details		\\
            \hline
            $r_1$	& IMAP	& email (in)\\	
            $r_2$	& SMTP & email (out)	\\	
            $r_3$	& DB &data access							
            \end{tabular}			
    }
    \end{minipage}
    \caption{
	\Cref{ex:resilience-system-migration,ex:causal-responsibility-system-migration}:
	Data for System Migration Example}
	\label{ex:systemmigration}	
\end{figure}

\begin{example}[Resilience: System Migration Example]
	\label{ex:resilience-system-migration}
    A department would like to retire an old server $S$. 
    The IT department needs to understand if and how the server is currently used, to perform the migration to other servers more efficiently. 
    More formally, the administrator wants to understand why the following query $Q$ evaluates to true:

    $Q_s \datarule$ Users(x, n), AccessLog(x, y, ``S''), Requests(y, d) 

    Detailed analysis of the data (\cref{ex:systemmigration}) reveals that $Q_s$ is true due to
    (a) email-related requests by Alice, and (b) data access requests by several users. Thus, to perform the migration,
    the IT department should transfer user Alice to a different
    email server, and migrate the databases residing on $S$ to
    a different server.

    We can see that since this is a \emph{linear} query, we can find this minimal explanation in $\PTIME$.
\end{example}

\begin{example}[Responsibility: System Migration Example]\label{ex:causal-responsibility-system-migration}
    Consider the same scenario as in \cref{ex:resilience-system-migration}, but we would just like to reduce the load on server $S$ instead of retiring it. 

    We would like to understand the casual responsibility of each input tuple considered towards the output of $Q_s$.

    We see that both $u_1$ and $r_3$ tuples have a counterfactual contingency set of $1$, giving them the highest responsibilities.

\end{example}   

\section{An interesting connection to Valued CSPs}
\label{SEC:VCSPRESILIENCECONNECTION}

After our paper was accepted, a very related and interesting preprint 
by Bodirsky et al.\ 
appeared on arXiv~\cite{bodirsky2023complexity} that focuses on the resilience dichotomy conjecture, 
yet in the context of a more general problem of valued constraint satisfaction problems (VCSPs) of valued structures with an oligomorphic automorphism group. 
The paper uses universal algebra and prior results on VCSPs ~\cite{kolmogorov2017complexity} to give one formalism (Theorem 7.17) that if fulfilled makes a query easy, 
and another formalism (Corollary 5.13) that allows checking if a query is hard. 
The paper's conjecture (Conjecture 8.18) is that those two cases are tight (i.e. every query fulfills either one or the other case). 
Our paper and theirs~\cite{bodirsky2023complexity} are similar in that:
\begin{enumerate}
    \item They both present a unified framework to solve resilience problems for conjunctive queries including those with self-joins.
    \item They both conjecture that the complexity of resilience of any query can be completely decided by the (seemingly different, but likely related) hardness criteria proposed in the papers: 
	we conjecture in~\cref{sec:IJP} that IJPs are a universal hardness criterion (a query is hard if and only there is a database that forms an IJP for that query), 
	while they conjecture in Conjecture~8.18 that pp-reductions from a particular valued structure in Corollary 5.13 is a universal hardness criterion.
\end{enumerate}
Besides the methods, other conceptual differences are as follows:

\begin{enumerate}
    \item Interestingly, the theoretical results in their paper appears are only applicable to bag semantics, making the bag case seemingly easier to analyze than set semantics. However, our approach can be applied for both set and bag semantics.
    \item Our approach comes with an explicit construction of a disjunctive logic program that takes a query as input and constructs an easy-to-verify hardness certificate if the query is hard.
    \item Our work comes with code implementations for actually solving resilience computationally, 
	both with exact and approximate algorithms. 
\end{enumerate}
It will be interesting to see how the methods and the tractability criterion in the two papers relate to each other and whether they are possibly complementary.

\section{Proofs for \cref{sec:res-ilp}: ILP For Resilience}

\thmresilpbag*
\begin{proof}[Proof \cref{thm:res-bag}]
    Assume there exists an optimally minimal resilience set $R$ such that it contains a tuple $t$, but it does not contain an identical tuple $t'$.
    Since $t$ and $t'$ are identical, they join with the same tuples and must participate in same number of witnesses.
    Since $t'$ is not in the resilience set, for every witness $\vec w_i$ that contains $t'$, there must be at least one tuple $x_i$ that is in the resilience set.
    All the witnesses that $t$ participates in, must also contain a tuple from the set of $x_i$.
    If none of the $x_i$ tuples is $t$ itself, then we can safely remove $t$ from $R$.
    Thus, R is not minimal, and we have a contradiction.
    
    However, in the case that there exists an $x_i = t$, this implies that $\vec w_i$ contains $t$ and $t'$ (along with $0$ or more other tuples $T[\vec w_i]$). 
    Since $t$ and $t'$ are identical, it follows that there is an identical witness created due to joining $t'$ with itself.
    This witness too must be destroyed - hence one of $T[\vec w_i]$ is in the resilience set, and we safely remove $t$, leading to a contradiction.
\qedhere
\end{proof}

\thmresilpcorrect*

\begin{proof}[Proof \cref{thm:res-correctness}]
    The proof is divided into parts to separately show the validity and optimality of $\resilpparam{Q, D}$. 
    An invalid solution would not destroy all the witnesses in the output, while a suboptimal solution would have size bigger than the minimum resilience set.
    \begin{itemize}

    \item Proof of Validity: Assume a solution is invalid i.e.\ after deleting the tuples in the resilience set, the number of output witnesses is not $0$.
    Since $Q$ is monotone, this witness existed in the original database as well. 
    A witness can only survive if all the tuples in the witness are not a part of the resilience set.
    Such a solution would hence violate the constraint for the surviving witness and hence would not be generated by the ILP.

    \item Proof of Optimality: Assume a solution is not optimal i.e.\ there exists a strictly smaller, valid resilience set $R'$.
    We could translate this set into a variable assignment $\bar{X}$ to $X[t]$ where $\bar{X}[t]=1$ if $t \in R'$.
    Since $R'$ is a valid resilience set, it would satisfy all the constraints to destroy all witnesses in $D$
    and also be a valid solution for $\resilpparam{Q, D}$. 
    Thus, it cannot be smaller than the optimal solution for $\resilpparam{Q, D}$.
	\qedhere
    \end{itemize} 
\end{proof}

\section{Proofs for \cref{sec:resp-ilp}: ILP For Responsibility}

\thmrespilpcorrect*

\begin{proof}[Proof \cref{thm:resp-correctness}]
    Similar to \cref{thm:res-correctness}, we show validity and optimality.
    \begin{itemize}
        \item Proof of Validity: An invalid solution is not counterfactual, i.e., either it does not destroy all witnesses without $\resptuple$ or it destroys all witnesses. 
        The former violates the resilience constraints, while the latter violates the Counterfactual constraint. 
        \item Proof of Optimality: Any strictly smaller, valid responsibility set $R'$ can also be translated into variable assignment $\bar{X}$ to $X[t]$ such that it satisfies all constraints (where $\bar{X}[t] = 1$ if $t \in R'$).
        Since $R'$ is valid, at least one tuple for each witness that does contain $\resptuple$ is destroyed - thus resilience constraints are fulfilled.
        There must be at least one witness $\w_p$ containing $\resptuple$ that is preserved. 
        For this witness, we know that $X[w_p]=0$ is valid (since there is no $t' \in w_p s.t. X[t']=1$). 
        Thus, the counterfactual constraint is also fulfilled since $\sum X[w] < |W_p|$.
        Thus, $\rspilp{Q,D}$ calculates the optimal responsibility.
    \qedhere
    \end{itemize} 
\end{proof}

\section{Proofs for \cref{SEC:RELAXATIONS}: ILP Relaxations} 

\thmrspmilpptime*

\begin{proof}[Proof \cref{thm:rspmilpptime}]
    Assume that there are $c_t$ witnesses that contain $\resptuple$. 
    The counterfactual constraint enforces that at least one of these witnesses is preserved.
    Notice that the witness indicator variables have no effect on the objective, and can be set to any value so long as all constraints are fulfilled.
    Any assignment where $1$ witness is preserved, and the rest are destroyed fulfills all constraint (even witness tracking constraints, which only enforce that a witness is destroyed if one of its tuples is destroyed, but does not enforce that the witness cannot be destroyed otherwise).
    Thus, we can restrict ourselves to $c_t$ potential assignments of witness indicator variables instead of $2^{c_t}$.
    Trivially, we can now solve the problem by running $c_t$ Linear Programs (where the only variables are tuple indicator variables and the witness indicator variables are fixed to one out of $c_t$ assignments). 
    Since $c_t$ is polynomial in the database size, we see that $\rspmilp$ can be solved in $\PTIME$.
    In practice, ILP solvers solve the problem faster than the algorithm in the proof, since the leverage common insights across the $c_t$ Linear Programs.
\end{proof}

\section{Additional Details for \cref{SEC:IJP}: Finding Hardness Certificates} 

We show a full end to end $\ijpdlp$ in \cref{sec:appendix:dlp}.

\subsection{More Example IJPs}
We also give simpler automatically derived IJPs for $k=3$ for the following 3 previously known hard queries.
The original hardness proofs for those queries \cite{DBLP:journals/pvldb/FreireGIM15} are pretty involved and cover several pages.
Our new hardness proofs are just \cref{Fig_IJPs_for_chains} given \cref{th:ICPs}.
\begin{align*}
	q_\textup{chain} 		&\datarule R(x,y),R(y,z)  					\\
	q_\textup{chain}^b 		&\datarule R(x,y), B(y), R(y,z)  			\\	
	q_\textup{chain}^{abc}	&\datarule A(x), R(x,y), B(y), R(y,z), C(z)  	
\end{align*}

\begin{figure}[t]
\centering
\includegraphics[scale=0.3]{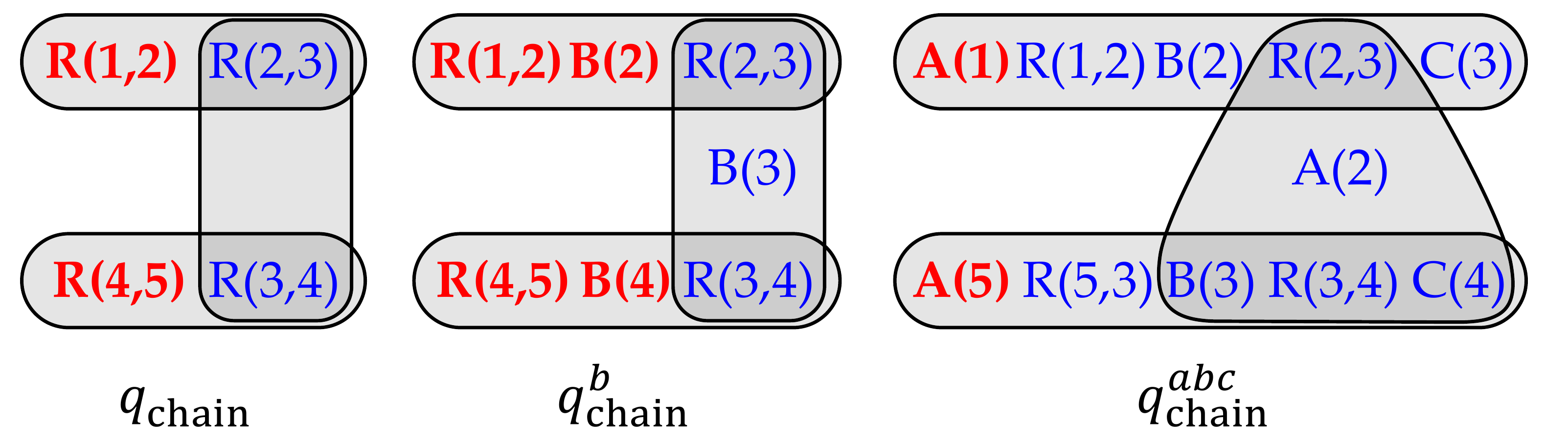}
\caption{Simple IJPs for prior known hard queries.
}
\label{Fig_IJPs_for_chains}
\end{figure}

\subsection{Explaining the domain bound of \cref{conj:hardness}}

We explain here further the intuition for bounding the size of an IJP to domain $d=7 \cdot |\var(Q)|$, and show examples of an IJP broken down into its components of ``core'', ``dominated'' or ``legs''.
\cref{Fig_IJP_sj_new_legs} shows the $5$ automatically generated IJPs for queries with self-joins, whose complexities where previously unknown, broken down into these components.

We see that \cref{Fig_example_IJP_3CC} and \cref{Fig_example_IJP_3perm_SC} consist only of $3$ core witnesses. This is the simplest possible IJP and is like the self-join-free case, where the $3$ witnesses correspond to the $3$ atoms of the triad.
However, notice that the $3$ core witnesses of \cref{Fig_example_IJP_3perm_AS} necessitate the presence of a ``dominated'' gray witness. 
This witness does not use any tuple that is not already part of the core, and does not increase the domain size of the IJP.
Additionally, this witness does not increase the \emph{transversal number} between the endpoint tuples i.e. the number of witnesses in the path from one endpoint tuple set to another. 
However, due to this witness, the endpoint tuple $A(5)$ is no longer independent. 
Thus, we require the introduction of a ``leg'' to obtain an independent endpoint tuple $A(2)$.

In \cref{Fig_example_IJP_3perm_SB}, we see a similar classification where the core creates a dominated witness, and $1$ leg of $2$ witnesses is needed to make the endpoint tuples independent.
\cref{Fig_example_IJP_z6} shows a slightly more complicated IJP, where the core results in $2$ dominated witnesses, but notice that neither added witness adds to the domain value or affects the transversal number.

\begin{figure}[t]
\centering
\begin{subfigure}[b]{.48\linewidth}
	\centering
	\includegraphics[scale=0.3]{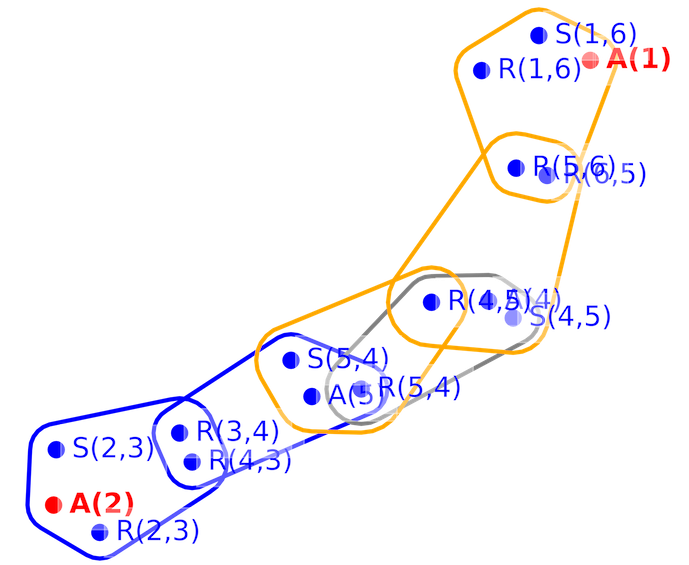}
	\caption{$q^{AS_{xy}}_{3perm-R} \datarule A(x), S(x,y), R(x,y), R(y,z), R(z,y)$}\label{Fig_example_IJP_3perm_AS}
\end{subfigure}
\begin{subfigure}[b]{.48\linewidth}
	\centering
	\includegraphics[scale=0.3]{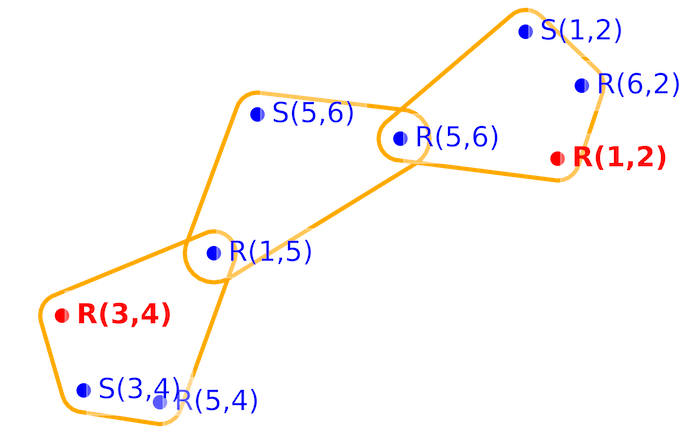}
	\caption{$q^{S}_{3CC} \datarule R(x,y), R(y,z), R(w,z), S(w,z)$}\label{Fig_example_IJP_3CC}
\end{subfigure}
\begin{subfigure}[b]{.48\linewidth}
	\centering
	\includegraphics[scale=0.3]{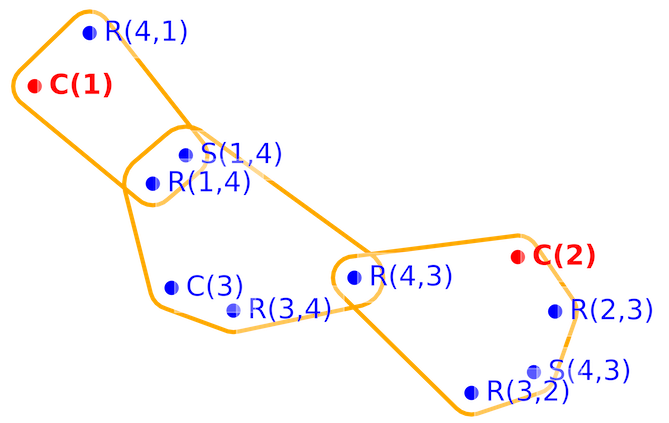}
	\caption{$q^{S_{xy}C}_{3perm-R} \datarule S(x,y), R(x,y), R(y,z), R(z,y), C(z)$}\label{Fig_example_IJP_3perm_SC}
\end{subfigure}
\begin{subfigure}[b]{.48\linewidth}
	\centering
	\includegraphics[scale=0.3]{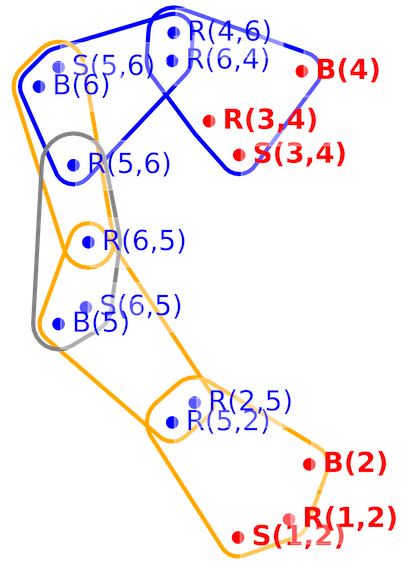}
	\caption{$q^{S_{xy}B}_{3perm-R} \datarule S(x,y), R(x,y), B(y), R(y,z), R(z,y)$}\label{Fig_example_IJP_3perm_SB}
\end{subfigure}
\begin{subfigure}[b]{.48\linewidth}
	\centering
	\includegraphics[scale=0.3]{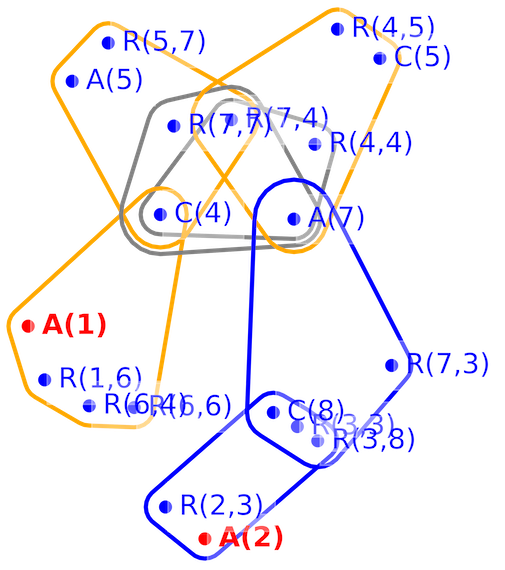}
	\caption{$z_6 \datarule A(x), R(x,y), R(y,y), R(y,z), C(z)$}\label{Fig_example_IJP_z6}
\end{subfigure}
\caption{
Automatically generated and visualized IJPs for $5$ previously open queries. The nodes corresponding to tuples in $\mathcal{S} \cup \mathcal{T}$ are in red.
In contrast to \cref{Fig_IJP_sj_new},
we color each witness (hyperedge) with a different color, depending on if it is part of the ``core'' IJP (in orange), or if it is a ``leg'' (in blue), or it is ``dominated'' witness
(in gray), automatically generated by the tuples present in the core witnesses (whose presence does not affect the transversal number between the endpoint tuples).
The assignment of an IJP into these components is not unique, i.e. it is possible, for example in \cref{Fig_example_IJP_3perm_AS} to treat the other endpoint tuples as the ``leg''.
}

\label{Fig_IJP_sj_new_legs}
\end{figure}

\section{Proofs for \cref{SEC:IJP}: Finding Hardness Certificates}

\proptrianglecomposition*
\begin{proof}[Proof \cref{prop:triangle:composition}]

Condition ($3ii$) of \cref{def:JP} implies that given two canonical join paths 
(i.e.\ they are isomorphic, and all constants are distinct), sharing constants in one end point of each join path, guarantees that the only endogenous tuples that the join paths share are the endpoint tuples.
What can happen is that this sharing of endpoints creates \emph{additional} witnesses which will affect the resilience of the resulting database instance.
What we like to prove is that if the composition from \cref{Fig_minimal_JP_composition} is not leaking, then \emph{any} composition is non-leaking (and thus creates no new witnesses).

From the condition that the endpoints of a join path have disjoint constants, 
and the fact that any two join paths can share maximally one endpoint, 
it follows that the tuples from two join paths that are not sharing any endpoint cannot create additional witnesses;
additional witnesses can only be created by two join paths sharing an endpoint. 

Since join paths can be asymmetric, there are three ways that two join paths can create additional witnesses: 
they are either sharing the start tuples, or the terminal tuples, or one start tuple is identical to the other end tuple.
All three cases are covered by \cref{Fig_minimal_JP_composition}.

It remains to be shown that sharing the same end tuples across multiple join paths can't add additional witnesses.
This follows now from induction with the three base cases covered above.
To illustrate, assume adding a third join path by their terminal to a start tuples shared by two join paths leads to additional witnesses. 
Then from the isomorphism between the two prior join paths it follows that a new witness would have to be created from having only one join path with the end tuples as start tuples. This is a contradiction. The same argument can be used for adding join paths to the other three bases cases.
\qedhere
\end{proof}	

\thmijpsnpc*

\begin{proof}[Proof \cref{th:ICPs}]
    The proof follows from a simple reduction from vertex cover.
    Assume $Q$ can form IJPs of resilience $c$.
    Take any directed simple graph $G(V, E)$ with $n$ nodes and $m$ edges.
    Encode each node $v \in V$ with a unique tuple 
    $\bm v = (\langle v_1 \rangle, \langle i_v \rangle, \ldots, \langle v_d \rangle) \in R$ 
    where $d$ is the arity of $R$.
    Encode each edge $(v,u) \in E$ as separate IJP from $R(\bm v)$ to $R(\bm u)$ 
    with fresh constants except their endpoints.
    Then $G$ has a Vertex Cover of size $k$ iff resilience $\res^*(Q,D)$ is $k+m (c-1)$.
    
    Notice that the semantic condition 5 is needed. 
    It guarantees that there is no tuple (other than the endpoints) that are shared between two different join paths (corresponding to the edges),
    and no additional joins are created.
    Without that condition, the join paths are not \emph{independent} and \emph{leakage} across join paths could otherwise change the resilience of the composition.
    \qedhere
    \end{proof}
    
\thmijpsjfcqs*

\begin{proof}[Proof \cref{th:IJPs_for_SJFCQs}]
First let us consider this theorem under set semantics.
We already know from \cref{th:ICPs} that IJPs $\Rightarrow$ NPC.
We also know from \cite{DBLP:journals/pvldb/FreireGIM15} that all hard queries under set semantics must have an active triad.
Recall that an \emph{active triad} is a set of three endogenous (and therefore, non-dominated) 
atoms, $\mathcal{T} = \set{R_1,R_2,R_3}$ 
such that for every pair $i \neq j$, 
there is a path from $R_i$ to $R_j$ that uses no variable occurring in the other 
atom of $\mathcal{T}$. 
It remains to be shown that queries with triads also have an IJP.

Let $q$ be a query with triad $\mathcal{T}=\set{R_1, R_2, R_3}$.
We will choose appropriate constants to build an IJP from $\{R_1(\vec a)\}$ to $\{R_1(\vec b)\}$  
consisting of three witnesses
$\{\vec w_1, \vec w_2, \vec w_3 \}$
s.t.\ $\vec w_1$ and $\vec w_2$ share $R_2$
and $\vec w_2$ and $\vec w_3$ share $R_3$.
In other words, we compose the 3 paths from the triad as 
$R_1(\vec a) \rightarrow R_2(\vec c) \rightarrow R_3(\vec d)  \rightarrow  R_1(\vec b)$.

We will assume that no variable is shared by all three elements of $\mathcal{T}$  
(we can ignore any
such variable by setting it to a constant). 
Our proof splits into two cases:

\emph{Case 1}:  $\var(R_1), \var(R_2), \var(R_2)$ 
are pairwise disjoint: 
We use unique constants $a, b, c, d, w_1, w_2, \\w_3$
and add tuples 
$R_1(a, a, \ldots, a)$,
$R_2(c, c, \ldots, c)$,
$R_3(d, d, \ldots, d)$, and
$R_1(b, b, \ldots, b)$ to $D$.

To define the relations corresponding to the other atoms in $\vec w_1$, we first 
partition the variables of $q$ into 3 disjoint sets: $\var(q)= \var(R_1) \cup \var(R_2) \cup W_1$.
Now for each atom $A_i \in q \setminus \{R_1, R_2\}$, arrange its variables in these three groups. 
Then define a tuple $(a; d; w_1)$ to relation $R_i$ of $D$ corresponding to atom $A_i$.
For example, all the variables $v \in \var(R_1)$ are assigned the value $a$
and all the variables $v\in W_1$ are assigned $w_1$.
Repeat the same process analogously for witnesses $\vec w_2$ and $\vec w_3$.

From our construction $\vec w_1$ and $\vec w_2$ share only one single endogenous tuple: $R_2(\vec c)$.
This follows from the fact that there is no other tuple that dominates (has a subset of variables) of endogenous tuples.
It follows that every endogenous tuple in $\vec w_1$ needs to contain at either at least constant $a$ or $w_1$ (and optionally $c$)
Similarly every endogenous tuple in $\vec w_2$ needs to contain at either at least constant $d$ or $w_2$ (and optionally $c$).

It follows that the resilience of the resulting database is identical to vertex cover of the graph
$R_1(\vec a) \rightarrow R_2(\vec c) \rightarrow R_3(\vec d)  \rightarrow  R_1(\vec b)$,
which fulfills condition (4) of \cref{def:IJP}.
Condition (5) follows from the same fact that every tuple in one join path needs to contain at least one constant not contained a tuple from another join path, other than the maximally one shared endpoint.

\emph{Case 2}:  $\var(R_i) \cap \var(R_j) \ne \emptyset $  for some $i\ne j$:
The previous construction can now be generalized from Case 1 by partitioning $\var(R_i)$ into those 
unshared, those shared with $R_{i-1}$, and those shared with $R_{i+1}$ (addition here is mod 3)
and verifying that the resulting database still fulfills the same conditions (4) and (5).

    The construction of IJPs for bag semantics follows the same argument, with the addition that some tuples have a fixed large number of copies such that the tuple would never be picked for the minimum resilience set i.e.\ it is made exogenous.
    We can use this to prove that that all triads (whether active or deactivated) are hard by making all tuples of dominating tables exogenous in this manner and following the rest of the proof of the set semantics case.
    We show in \cref{SEC:THEORECTICALRESULTS} by \cref{thm:reseasy} that resilience for all queries without triads (linear queries) is $\PTIME$. 
    Hence, for all SJ-free CQs under bag semantics, resilience is NPC iff it has an IJP.
\qedhere
\end{proof}

\section{Proofs for \cref{SEC:THEORECTICALRESULTS}}

\subsection{Proofs for \cref{sec:theory:resilience}: Theoretical Results for Resilience}

\thmreseasydomination*

\begin{proof}[Proof \cref{thm:reseasydomination}]
	Assume $Q$ contains a triad with tables $R$, $S$, $T$. 
	However, since the triad is deactivated, at least one of these tables must be dominated by another table $A$. 
	WLOG, assume $R$ is dominated by $A$.
	We show that $R$ can be made exogenous because there exists an optimal resilience set that does not contain any tuple from $R$.
	If $r_i$ is part of the resilience set, then it can be replaced with $a_i$ where $a_i \subset r_i$ while still destroying the same or more witnesses.	
	Since no tuple from $R$ is actually used in the resilience, the size of the resilience set will not change if we make $R$ exogenous i.e.\ add all the variables of the query to $R$. 
	Let $Q'$ be the query where for each deactivated triad, all dominated tables have been made exogenous.
	Then $\res(Q',D)= \res(Q,D)$.
	$Q'$ is linear, and we can then use \cref{thm:reseasy} to show that $\reslpparam{Q,D}$ is optimal.
\qedhere
\end{proof}

\thmreshardbag*

\begin{proof}[Proof \cref{thm:reshardbag}]
	A non-linear query by definition must contain triads.
	If the query contain active triads, then \cref{th:IJPs_for_SJFCQs} can be applied in the bag semantics setting as well to show that $\res$ is $\NPC$.
	However, the same IJP does not directly work for (fully) deactivated triads - since the endpoints are part of the triad tables, they can be dominated by another tuple in the IJP.
	Then the optimal resilience would be to choose the dominating tuple, thus no longer fulfilling the first criteria of independence. 
	Hence, we must have a slightly different IJP with the property that the dominating table is exogenous.
	To make the dominating table exogenous, it suffices that we have $c_w$ copies of each tuple from the table in the IJP (where $c_w$ is the number of witnesses in the IJP under set semantics), and $1$ copy of all other tuples.
	Using \cref{thm:res-bag} where we showed that it is never beneficial to remove some copies of a tuples, and the fact that the resilience of the IJP is at most $k$, we can see that it is never necessary to remove tuples from the dominating table.
\end{proof}

  \subsection{Proofs for \cref{sec:theory:responsibility}: Theoretical Results for Responsibility}

  \thmrspeasyfullydominated*

  \begin{proof}[Proof \cref{thm:rspeasyfullydominated}]
      Assume $Q$ contains a triad with tables $R$, $S$, $T$. 
      However, since the triad is fully deactivated, at least one of these tables must be dominated by set of tables $A_1, A_2 \hdots$. 
      WLOG, assume $R$ is fully dominated.
      We show that $R$ can be made exogenous because there exists an optimal responsibility set that does not contain any tuple from $R$.
      If $r_i$ is part of the responsibility set, then it can be replaced with $a_i$ where $a_i \subset r_i$ while still destroying the same or more witnesses.
      However, it is still possible that including $a_i$ in the responsibility set may destroy all witnesses. 
      This is possible only if $a_i$ dominates $\resptuple$ as well.
      If all $a_i$ such that $a_i \subset r_i$ dominate $\resptuple$, then it must be that $r_i$ dominates $\resptuple$ (since $r_i$ is fully dominated and uniquely determined by the tuples that dominate it).
      It is not possible for such an $r_i$ to be in the responsibility set as it would destroy all witnesses containing $\resptuple$.
      Thus, no tuple from $R$ can be used in the responsibility set, the size of the responsibility set will not change if we make $R$ exogenous i.e.\ add all the variables of the query to $R$. 
      Let $Q'$ be the query where for each deactivated triad, all fully dominated tables have been made exogenous.
      Then $\rsp(Q', D, \resptuple)= \rsp(Q,D, \resptuple)$.
      $Q'$ is linear, and we can then use \cref{thm:rspeasy} to show that $\rspmilpparam{Q,D}$ is optimal.
      \qedhere
  
  \end{proof}

  \thmrspeasydominatingtable*
  
  \begin{proof}[Proof \cref{thm:rspeasydominatingtable}]
      Let $R$ be the table in a deactivated triad that $A$ dominates. 
      We show that no tuple of $R$ is required in the responsibility set, and we can make it exogenous.
      If some $r_i$ is in the responsibility database, it can be replaced with some $a_i$ if the variables and valuation of $a_i$ are a strict subset of $r_i$ and then $a_i$ deletes all the witnesses as before, and potentially some more.
      This is permitted unless removal of $a_i$ deletes all witnesses containing $\resptuple$ as well.
      However, since $a_i$ and $\resptuple$ belong to the same table, this is not possible.
      Thus, at least one table from each triad can be made exogenous, and the query can be replaced with a linear query.
      \qedhere
  
  \end{proof}
  
 \thmrspharddominatedtable*

  \begin{proof}[Proof \cref{thm:rspharddominatedtable}]
      If T is part of an active triad, the same IJP as $\res$ is proof for this theorem.
      However, if $T$ is part of a deactivated triad, then we need to slightly modify the hardness proof.
      Let $A$ be the table that dominates one of the tables in the deactivated triad. 
      In our IJP we ensure that an atom from $A$ is an \emph{exogenous tuple}- one that cannot be deleted.
      This is possible by constructing an $a_i$ that dominates $\resptuple$.
      Since this is always possible, we can now construct the rest of the IJP.
      We connect a witness containing $a_0$ to two others by using two tables of the deactivated triad. 
      Then we finally add two more witnesses to the triad with the common tuple being the third table of the deactivated triad.
      We treat the $A$ table as the endpoints of the IJP.
      Since $a_i$ is exogenous, the gadget must choose between the first or the second table to destroy all witnesses in the IJP.
      Such a gadget does not form new witnesses when composed as well as any two isomorphs share only tuples from $A$.
      \qedhere
  
  \end{proof}

\begin{figure}[t]
	\centering
		\includegraphics[scale=0.3]{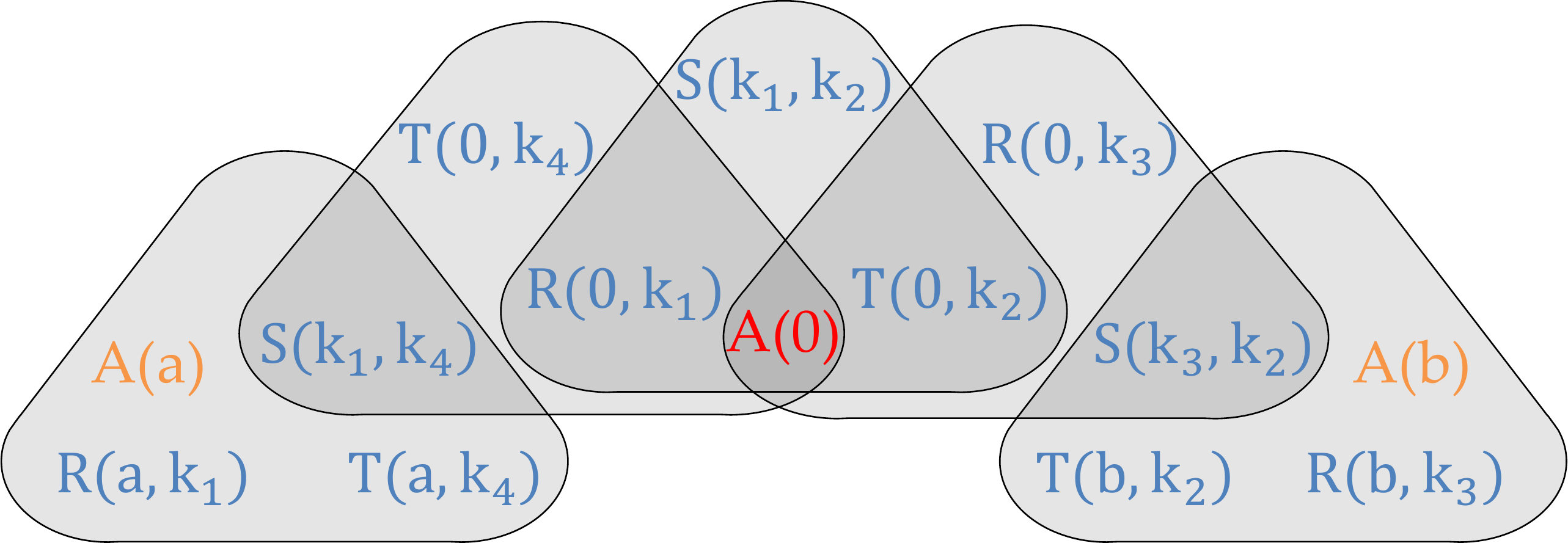}
	\caption{IJP for $\rsp(\qtriangleunary)$ for tables $R$, $S$ and $T$}
	\label{Fig_TU_RESP_IJP}
\end{figure}

  In figure \cref{Fig_TU_RESP_IJP}, we show an example for the IJP that greatly simplifies the previous hardness gadget (the earlier gadget was a reduction from $3SAT$ whose variable gadget had $80$ witnesses)

 \thmrspharderthanres*

  \begin{proof}[Proof \cref{thm:rspharderthanres} ]
      Consider an arbitrary database instance $D$ and add all tuples from a witness $w_r$ that is disjoint from all tuples in $D$.
      The responsibility of the resulting database instance is simply the resilience of $D$ (since all witnesses in $D$ must be destroyed, and the other singleton witness must be preserved).
      Thus, we can reduce $\rsp(Q,D,\resptuple)$ to $\res(Q,D)$ and $\rsp(Q)$ must be hard whenever $\res(Q)$ is.
      \qedhere
  
  \end{proof}

\section{Additional Instance-Based Results}
\label{SEC:THEORY:READONCEANDFDS}
\label{sec:theory:readonceandfds}

We give here two cases for when our unified algorithm is \emph{guaranteed to terminate in $\PTIME$} for generally hard queries.
The interesting aspect is
that our unified algorithm terminates in $\PTIME$ if the database instance fulfills those conditions, but the algorithm \emph{does not need to know about these conditions as input},
it just automatically leverages those during query time.
We believe that this really shows the power of our unconventional approach of proposing one unified approach for all problems and then proving termination in $\PTIME$ for increasing number of cases (instead of starting from a dedicated $\PTIME$ solution for special cases).

\introparagraph{Read-Once Instances}
We show that database instances which allow a read-once factorization of the provenance 
for a given query 
are always tractable. 
A Boolean function is called \emph{read-once} if it can be expressed as a Boolean expression 
in which every variable appears exactly once~\cite{CramaHammer2010:BooleanFunctions,
DBLP:journals/dam/GolumbicMR06,
GolumbicGurvich2010:ReadOnceFunctions}.
We call a database $D$ \emph{read-once instance} for query $Q$ if the provenance of the query over $D$ can be represented by a read-once expression.

\begin{restatable}{theorem}{thmresreadonce}
  \label{thm:resreadonce}
  $\reslp$ and $\rspmilp$ always have optimal, integral solutions under set or bag semantics for all database instances $D$ that are read-once for query $Q$.
\end{restatable}

\begin{proof}[Proof \cref{thm:resreadonce}]
    We use a structural property of the constraint matrix of the LP to show that $\reslpparam{Q,D} = \res(Q,D)$.
    A $\{0, 1\}$-matrix $M$ is balanced iff $M$ has no square submatrix of odd order, such that each row and each column of the submatrix has exactly two 1s.
    If a matrix $M$ is balanced, then the polytope $Mx \geq 1$ is Total Dual Integral (TDI), which means all vertices of the polytope are integral \cite{schrijver1998theory}.
    For such a system, the optimal Linear Program solution will always have an Integral solution.
    We first show that the constraint matrix of $\reslpparam{Q,D}$ is $0,1$-balanced when $D$ is read-once.
    A $0,1$ balanced matrix is one that does not contain any odd square submatrix having all row sums and all column sums equal to 2.
    
    Assume the constraint matrix is unbalanced. 
    Then there must be a set of witnesses ($w_1$, $w_2$, $w_3 \hdots$) such that $w_1$ and $w_2$ share tuple $t_1$ but not $t_2$, and $w_2$ and $w_3$ share $t_2$.
    This defines a $P4$, which is not permitted in a read-once instance. 
    Thus, the constraint matrix is balanced and $\reslpparam{Q,D} = \res(Q,D)$.
    \qedhere

    Now for $\rspilpparam{Q,D,\resptuple}$, if there is a tuple $x$ that exists in a witness with $\resptuple$ ($w_i$) as well as in a witness without $\resptuple$ ($w_j$), then $x$ must exist in all witnesses ($w_k \hdots$) containing $\resptuple$ to prevent the formation of a $P4$. (There would be a $P4$ as $w_k$ and $w_i$ share $\resptuple$, $w_i$ and $w_j$ share $x$ but $w_k$ and $w_j$ do not share $\resptuple$ or $x$.)
    If $x$ participates in all witnesses containing $\resptuple$ it cannot be part of the responsibility set as it would violate the counterfactual constraint by preserving no witnesses.
    Hence, the responsibility set consists wholly of tuples that do not interact with $\resptuple$ and the problem reduces to resilience, which we know is $\PTIME$ for read-once instances.
\end{proof}

\introparagraph{Functional Dependencies (FDs)}
A Functional Dependency (FD) is a constraint between two sets of attributes $X$ and $Y$ in a relation of a database instance $D$.
We say that $X$ functionally determines $Y$ ($X \rightarrow Y$) if whenever two tuples $r_1, r_2 \in R$ contain the same values for attributes in $X$, they also have the same values for attributes in $Y$~\cite{Kolahi2009}. 
Prior work introduced an \emph{induced rewrites} 
procedure~\cite{DBLP:journals/pvldb/FreireGIM15} 
which, given a set of FDs, rewrites a query to a simpler query without changing the resilience or responsibility. 
If the query after an induced rewrite is in $\PTIME$, then the original could be solved after performing a transformation.
We prove that any instance that is $\PTIME$ after an induced rewrite is \emph{automatically} easy for our ILPs.
Thus, if there are \emph{undetected FDs} in the data that would allow a $\PTIME$ rewrite, our framework guarantees $\PTIME$ performance, while prior approaches would classify it as hard.

\begin{restatable}{theorem}{thmresfd}
\label{thm:fd}
Let $Q'$ be the induced rewrite of $Q$ under a set of FDs.
If $\res(Q')$ or $\rsp(Q')$ are in $\PTIME$ under set or bag semantics
then $\reslp$ and $\rspmilp$ always have optimal integral solutions under the same semantics.
\end{restatable}

\begin{proof}[Proof \cref{thm:fd}]
    Prior work \cite{DBLP:journals/pvldb/FreireGIM15} showed that FDs can make things easy and be used to transform non-linear queries to linear queries.
    We can make the same argument as \cref{thm:rspeasy} to show that $\reslpparam{Q}$ or $\rspmilpparam{Q}$ cannot be smaller than the resilience or responsibility respectively found by the min-cut algorithm of the flow graph produced by the query after linearization.
    \qedhere
  
  \end{proof}

\section{Proofs for \cref{SEC:APPROXIMATIONS}: 
Approximation Algorithms}

\thmkfactorapproxres*

\begin{proof}[Proof \cref{thm:k-factor-approx}]
  The LP-Rounding algorithm is $\PTIME$ since it requires the solution of a linear program, which can be found in $\PTIME$, and a single iteration over the tuple variables.
  We also see that it is bounded by $m* \reslpparam{Q,D}$ since each variable is multiplied by at most $m$, and since $\reslpparam{Q,D} \leq \resilpparam{Q, D}$, the algorithm is at most $m$-factor the optimal value.
  Thus, it remains to prove that $X_I$ returned by the rounding, satisfies all constraints of $\resilpparam{Q, D}$.
  We know that for every constraint, we involve at most $m$ tuple variables\footnote{For SJ-free cases, exactly $m$ tuples are involved, but for queries with self-join a witness can have less than $m$ tuples}.
  Since the sum of these variables in $X_f$ must be at least $1$ (due to the constraints of $\reslp$), there must exist at least one tuple variable in each constraint with value $\geq 1/m$.
  Thus, in $X_I$, for each constraint, there is a tuple variable $t$ such that $X_I[t] =1$ and all constraints are satisfied.

  Now to prove the correctness of the approximation for $\rspmilp$ as well, we need to verify the extra constraints.
  We must ensure that the resultant variable assignment fulfills the Counterfactual Constraints to ensure that not all witnesses are deleted.
  However, since in the Mixed ILP, the witness variables already took on integral values, there was at least one witness $w_p$ containing $\resptuple$ such that $X[w] =0$.
  This implies that in the MILP, all tuples $t$ in $w_p$ have $X_f[t]=0$. 
  They will stay $0$ after rounding as well, and thus the Witness Tracking Constraints and Counterfactual Constraint are still satisfied.
\end{proof}  

\section{Two More Experimental Scenarios 
(\cref{SEC:EXPERIMENTS} Extended)}
\label{SEC:APPENDIX:EXPERIMENTS}

\begin{figure}
	\begin{subfigure}{\columnwidth}
        \centering
		\includegraphics[scale=0.1]{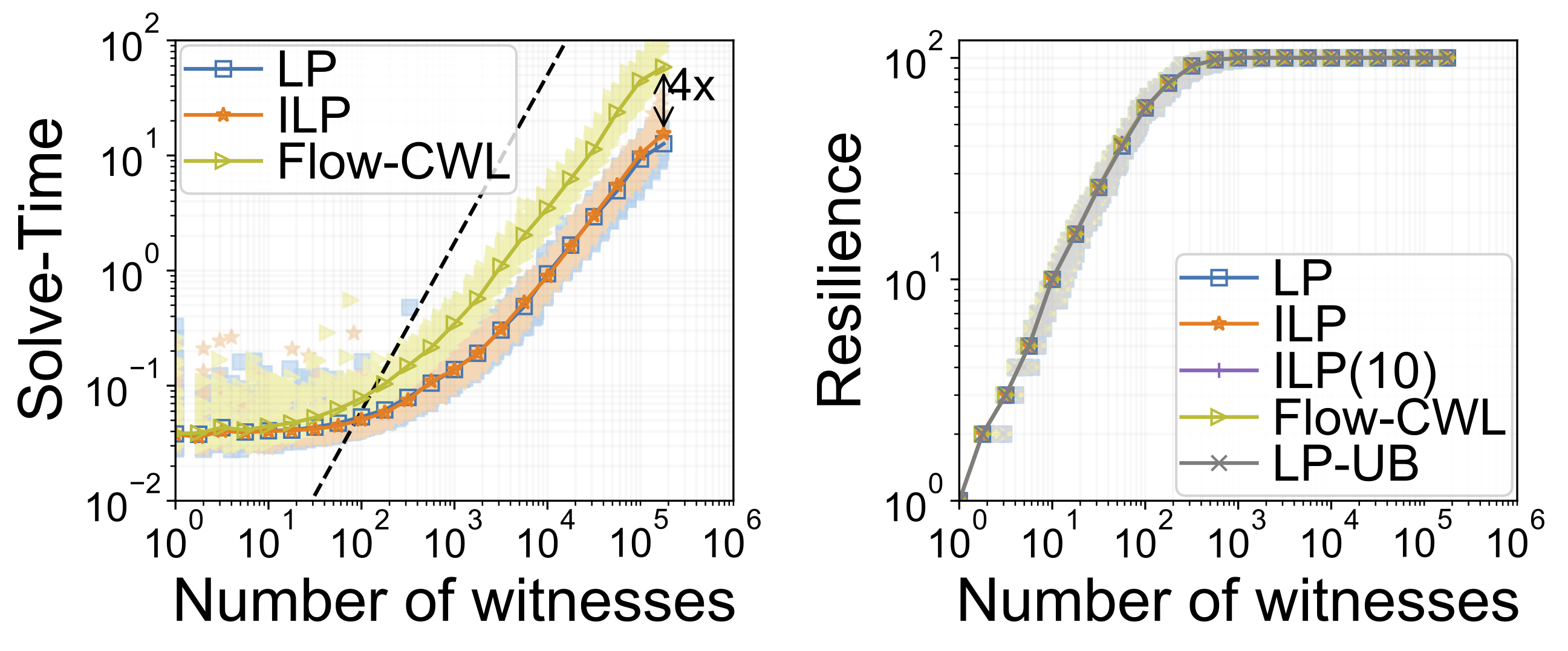}
		\caption{$\res(\qtriangleunary)$ under Set Semantics (an easy scenario)}
		\label{fig:Fig_expt_tu_set}
    \end{subfigure}
	\begin{subfigure}{\columnwidth}
		\centering
			\includegraphics[scale=0.1]{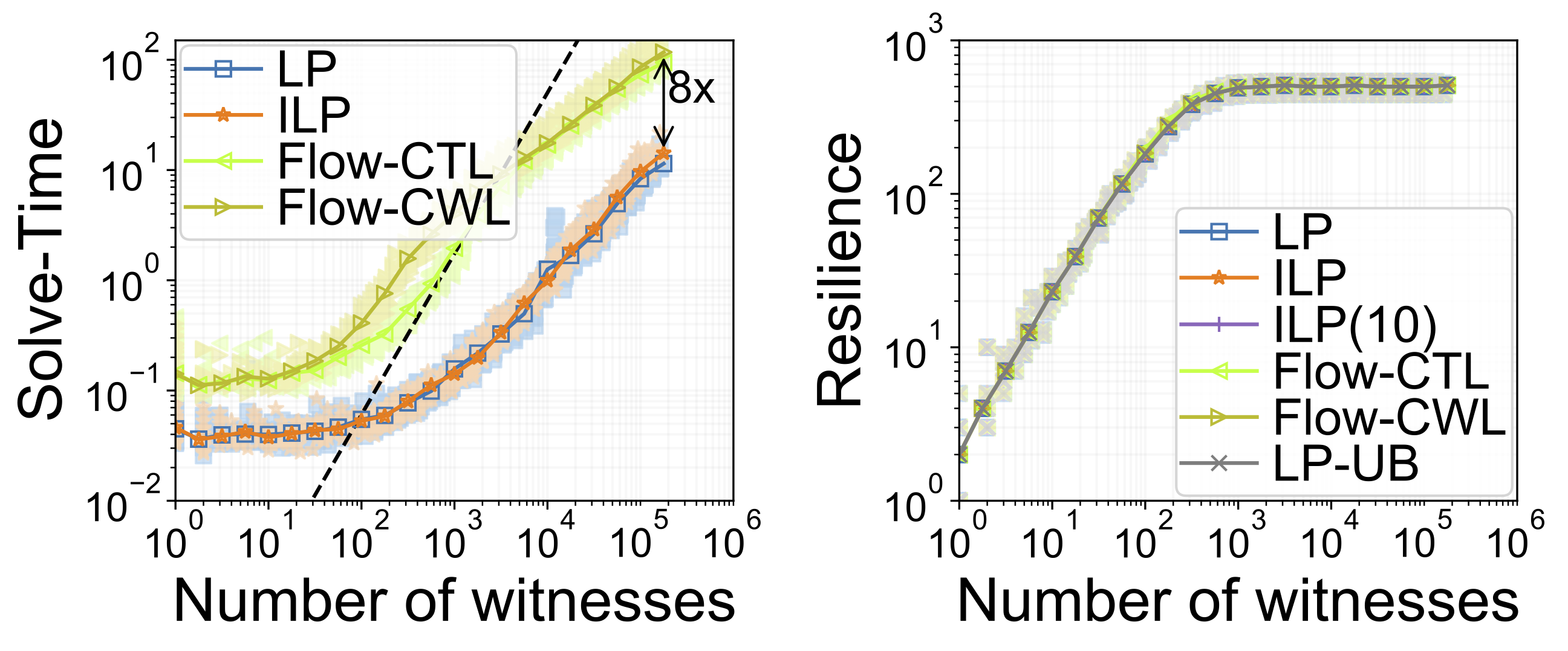}
		\caption{$\res(\qtriangleunary)$ under Bag Semantics (a hard scenario)}
		\label{fig:Fig_expt_tu_bag}
	\end{subfigure}
	\caption{Setting 4: $\res(\qtriangleunary)$ is easy for sets and hard for bags.}
    \label{fig:Fig_Expt_tu}
\end{figure}

\introparagraph{Setting 4: Resilience Under Set vs.\ Bag Semantics}
\Cref{fig:Fig_Expt_tu} shows $\qtriangleunary$, 
a query that contains a deactivated triad.
It is easy under set semantics and hard for bag semantics. 
However, surprisingly, even with a high max bag size of $1e4$, we always observed $\reslp = \resilp$, and the growth of ILP solve-time remained polynomial.
The approximation algorithms are slower, and almost always optimal, differing by less that $1.1 \times$ to the optimal in the worst case.

\begin{figure}[t]
	\centering
		\includegraphics[scale=0.1]{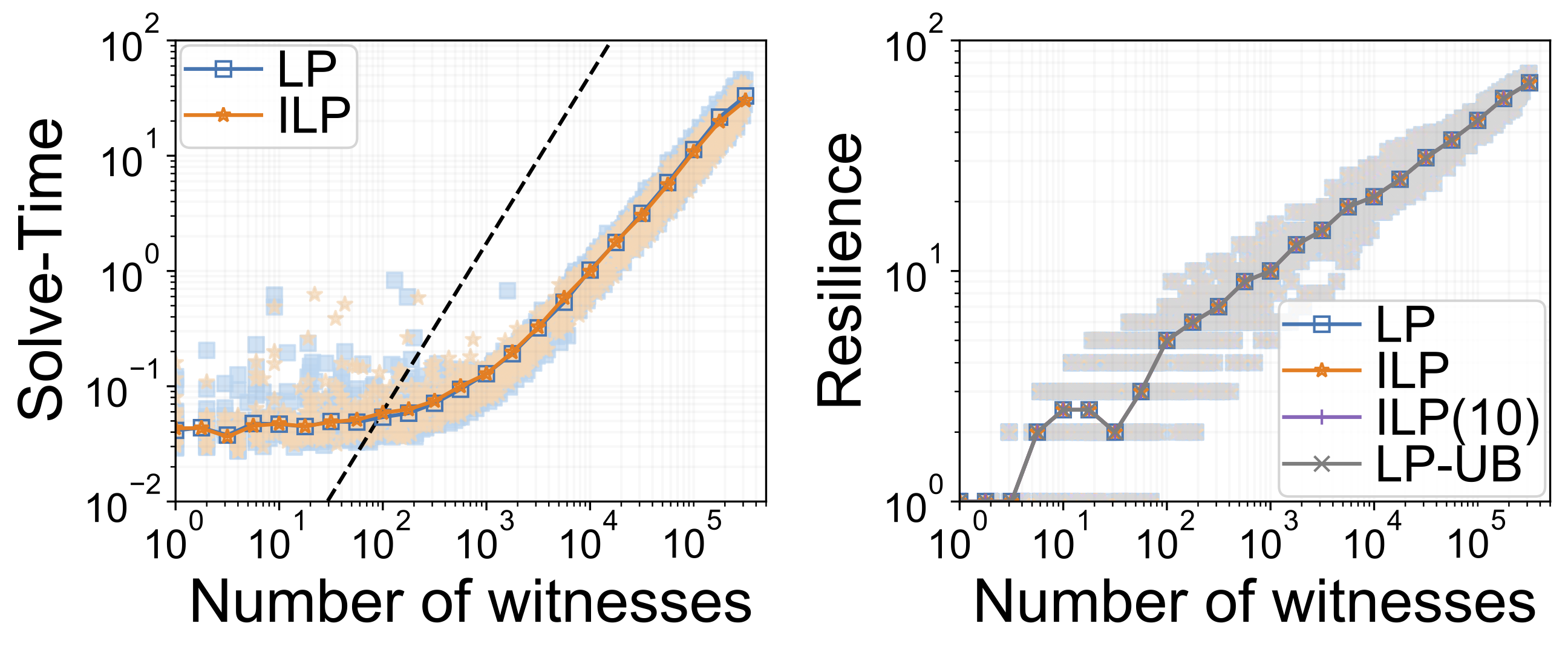}
	\caption{Setting 5: $\res$ for a newly proven hard SJ query.}
	\label{Fig_expt_sj_z6}
\end{figure}

\introparagraph{Setting 5: Self-Join Queries with newly founded hardness}
\cref{Fig_expt_sj_z6} investigates $z_6$ whose complexity we proved in \cref{sec:IJP} to be hard.
Although resilience for this query is hard, it is unlikely to create a random database instance where solving resilience is actually difficult.
Although the domain is pretty dense and the database instance large, for all experiments we run, the LP solution is integral and identical to the ILP solution.
However, by using our IJP, we could create an artificial synthetic database with 21 witnesses for which the LP solution is fractional.

These settings help us answer another interesting question:
(5) Do experimental scalabilities give hints about the hardness of queries? We see a rather surprising result.

\resultbox{\begin{resultW}(\textbf{Practical ILP scalability}) 
Hard queries may or may not show exponential time requirement in practice.
\end{resultW}}

\noindent
\cref{fig:Fig_expt_sj_chain} is a hard query that shows exponential growth. 
However, while exponential growth of solve-time is a hint for the hardness of a query, the converse is not necessarily true (\cref{Fig_expt_sj_z6,fig:Fig_expt_tu_bag}).
This (together with \cref{fig:Fig_expt_tpch_5cycle} over TPC-H) explains why our approach of using ILP to solve the problem is \emph{practically motivated}:
For realistic instances, or even dense instance but more complicated queries, scenarios where the hardness of the problem actually renders the problem infeasible may be rare.

\introparagraph{Additional Notes on Implementation}
We observed some surprising cases where the ILP was consistently faster than the LP. We learned from Gurobi Support
that this may be due to optimizations applied to the ILP that are not applied to the LP \cite{achterberg2020presolve}, and if such optimizations eliminate numerical issues in the LP \cite{gurobinum} such as issues due to floating-point arithmetic. 

\onecolumn

\section{IJP Disjunctive Logic Program}
\label{sec:appendix:dlp}

We show an example $\ijpdlp$ for the $2$-chain with self-join query $\qsjtwochain \datarule R(x,y),R(y, z)$.
Here we are able to show the code in its entirety for $d=5$\footnote{\cref{conj:hardness} implies that this hard query with $3$ variables has an IJP of domain size $\leq$ $3*7 = 21$. 
The conjecture is indeed true for $\qsjtwochain$, and in fact we have a far smaller IJP with domain size $5$.} and endpoints $\{R(1,2)\}$ and $\{R(3,4)\}$.

We solve this formulation with clingo ~\cite{clingo}, and find a hardness certificate in just $0.3$ seconds, running on a local Intel(R) Core(TM) i7-1065G7 CPU @ 1.30GHz with $8$ cores.

This example, along with many others, is available with our code online ~\cite{resiliencecode}.

\begin{scriptsize}
\begin{lstlisting}
r(1,1,1).
r(2,1,2).
r(3,1,3).
r(4,1,4).
r(5,1,5).
r(6,2,1).
r(7,2,2).
r(8,2,3).
r(9,2,4).
r(10,2,5).
r(11,3,1).
r(12,3,2).
r(13,3,3).
r(14,3,4).
r(15,3,5).
r(16,4,1).
r(17,4,2).
r(18,4,3).
r(19,4,4).
r(20,4,5).
r(21,5,1).
r(22,5,2).
r(23,5,3).
r(24,5,4).
r(25,5,5).

% 2. "Guess" an IJP
% For every tuple we define if it is in the IJP or not.
% We also calculate the witnesses and number of witnesses in the IJP
indb(r,Tid,1) | indb(r,Tid,0) :- r(Tid,_,_).
witness(X,Z,Y,T1,T2) :- r(T1,X,Y),r(T2,Y,Z),indb(r,T1,1),indb(r,T2,1).
number_of_witnesses(K) :- #count{X,Z,Y,T1,T2 : witness(X,Z,Y,T1,T2) } = K.

% 3. JP End Condition
range_triangle(1..3).ijp_domain(1..5).

% Define endpoint constants i.e. domain values that belong to endpoint 1 (Source) and 2 (Target)
end1const(1).
end1const(2).
end2const(3).
end2const(4).

end1witness(T1,T2):-witness(X,Z,Y,T1,T2),indb(r,T1,1),r(T1,X,Y),end1const(X),end1const(Y).
end1witness(T1,T2):-witness(X,Z,Y,T1,T2),indb(r,T2,1),r(T2,Y,Z),end1const(Y),end1const(Z).
:-not#count{T1,T2:end1witness(T1,T2)}=1.
end2witness(T1,T2):-witness(X,Z,Y,T1,T2),indb(r,T1,1),r(T1,X,Y),end2const(X),end2const(Y).
end2witness(T1,T2):-witness(X,Z,Y,T1,T2),indb(r,T2,1),r(T2,Y,Z),end2const(Y),end2const(Z).
:-not#count{T1,T2:end2witness(T1,T2)}=1.


:- witness(X,Z,Y,T1,T2),end1const(X),end1const(Z),end1const(Y).
:- witness(X,Z,Y,T1,T2),end2const(X),end2const(Z),end2const(Y).

% 4. Calculate Resilience

valid_res2(r,2,1).
invalid_res2(r,2,1).
valid_res3(r,14,1).
invalid_res3(r,14,1).
valid_res4(r,2,1).
invalid_res4(r,2,1).
valid_res4(r,14,1).
invalid_res4(r,14,1).

invalid_res1(r,Tid,1) | invalid_res1(r,Tid,0) :- r(Tid,_,_).
invalid_res2(r,Tid,1) | invalid_res2(r,Tid,0) :- r(Tid,_,_).
invalid_res3(r,Tid,1) | invalid_res3(r,Tid,0) :- r(Tid,_,_).
invalid_res4(r,Tid,1) | invalid_res4(r,Tid,0) :- r(Tid,_,_).
valid_res1(r,Tid,1) | valid_res1(r,Tid,0) :- r(Tid,_,_).
valid_res2(r,Tid,1) | valid_res2(r,Tid,0) :- r(Tid,_,_).
valid_res3(r,Tid,1) | valid_res3(r,Tid,0) :- r(Tid,_,_).
valid_res4(r,Tid,1) | valid_res4(r,Tid,0) :- r(Tid,_,_).


invalid_resilience1 :- witness(X,Z,Y,T1,T2),invalid_res1(r,T1,0),invalid_res1(r,T2,0).
invalid_resilience1 :- #count{Table,Tid: invalid_res1(Table,Tid,1)} >= K,res(K).
invalid_resilience2 :- witness(X,Z,Y,T1,T2),invalid_res2(r,T1,0),invalid_res2(r,T2,0).
invalid_resilience2 :- #count{Table,Tid: invalid_res2(Table,Tid,1)} >= K,res(K).
invalid_resilience3 :- witness(X,Z,Y,T1,T2),invalid_res3(r,T1,0),invalid_res3(r,T2,0).
invalid_resilience3 :- #count{Table,Tid: invalid_res3(Table,Tid,1)} >= K,res(K).
invalid_resilience4 :- witness(X,Z,Y,T1,T2),invalid_res4(r,T1,0),invalid_res4(r,T2,0).
invalid_resilience4 :- #count{Table,Tid: invalid_res4(Table,Tid,1)} >= K+1,res(K).

% Here we are ``saturating'' the solution
invalid_res1(r,Tid,0) :- invalid_resilience1,r(Tid,_,_).
invalid_res1(r,Tid,1) :- invalid_resilience1,r(Tid,_,_).

invalid_res2(r,Tid,0) :- invalid_resilience2,r(Tid,_,_).
invalid_res2(r,Tid,1) :- invalid_resilience2,r(Tid,_,_).

invalid_res3(r,Tid,0) :- invalid_resilience3,r(Tid,_,_).
invalid_res3(r,Tid,1) :- invalid_resilience3,r(Tid,_,_).

invalid_res4(r,Tid,0) :- invalid_resilience4,r(Tid,_,_).
invalid_res4(r,Tid,1) :- invalid_resilience4,r(Tid,_,_).

:- not invalid_resilience1.
:- not invalid_resilience2.
:- not invalid_resilience3.
:- not invalid_resilience4.

% 5. Check for the OR Property

:- witness(X,Z,Y,T1,T2),valid_res1(r,T1,0),valid_res1(r,T2,0).
res(K) :- #count{Table,Tid: valid_res1(Table,Tid,1)} = K.
:- witness(X,Z,Y,T1,T2),valid_res2(r,T1,0),valid_res2(r,T2,0).
:- not #count{Table,Tid: valid_res2(Table,Tid,1)} = K,res(K).
:- witness(X,Z,Y,T1,T2),valid_res3(r,T1,0),valid_res3(r,T2,0).
:- not #count{Table,Tid: valid_res3(Table,Tid,1)} = K,res(K).
:- witness(X,Z,Y,T1,T2),valid_res4(r,T1,0),valid_res4(r,T2,0).
:- not #count{Table,Tid: valid_res4(Table,Tid,1)} = K+1,res(K).

% 6. Check for non-leaking composition
% Build an isomorph map mapping IJP to different isomorphs

% end1 gets mapped to itself for edge 1
iso_map(C,1,C) :-  end1const(C),range_triangle(I). 

%end1 gets mapped to 2 for edge 2 - add end arity
iso_map(C,2,X) :-  end1const(C),range_triangle(I),X = C + 2. 

%end1 gets mapped to itself for edge 3
iso_map(C,3,C) :-  end1const(C),range_triangle(I). 

%end2 gets mapped to itself for edge 1
iso_map(C,1,C) :-  end2const(C),range_triangle(I). 

%end2 gets mapped to 3 for edge 2
iso_map(C,2,X) :-  end2const(C),range_triangle(I),X = C + 2. 

%end2 gets mapped to 3 for edge 3
iso_map(C,3,X) :-  end2const(C),range_triangle(I),X = C + 2. 
iso_map(C,I,X) :- range_triangle(I),ijp_domain(C),X = C+(5+1)*I,not end1const(C),not end2const(C).

ijp_iso_1_r(TID,VI0,VI1):-indb(r,TID,1),r(TID,V0,V1),iso_map(V0,1,VI0),iso_map(V1,1,VI1).
ijp_iso_2_r(TID,VI0,VI1):-indb(r,TID,1),r(TID,V0,V1),iso_map(V0,2,VI0),iso_map(V1,2,VI1).
ijp_iso_3_r(TID,VI0,VI1):-indb(r,TID,1),r(TID,V0,V1),iso_map(V0,3,VI0),iso_map(V1,3,VI1).

ijp_iso_triangle_r(TID,V0,V1) :- ijp_iso_1_r(TID,V0,V1).
ijp_iso_triangle_r(TID,V0,V1) :- ijp_iso_2_r(TID,V0,V1).
ijp_iso_triangle_r(TID,V0,V1) :- ijp_iso_3_r(TID,V0,V1).
ijp_triangle_witness(X,Z,Y) :- ijp_iso_triangle_r(T1,X,Y),ijp_iso_triangle_r(T2,Y,Z).
:- number_of_witnesses(K),not  #count{X,Z,Y : ijp_triangle_witness(X,Z,Y) }= 3*K.

% 7. (Optional) Minimize the size of the IJP
:~ witness(Z,Y,X,T1,T2). [1@1,Z,Y,X]


#show.
#show number_of_witnesses(K) : number_of_witnesses(K).
#show witness(X,Z,Y) : witness(X,Z,Y,T1,T2).
#show res(K) : res(K).

\end{lstlisting}
\end{scriptsize}

The code gives the following output, finding an IJP with 3 witnesses:

\begin{footnotesize}
\begin{lstlisting}
clingo version 5.6.2
Reading from ...gen_asp_scripts\ijp_expt_cases-1003.dl
Solving...
Progression : [1;inf]
Progression : [2;inf]
Answer: 1
res(3) witness(5,2,1) witness(4,3,5) witness(3,5,2) witness(3,5,5) witness(5,5,2) witness(5,5,5) number_of_witnesses(6)
Optimization: 6
Answer: 2
res(2) witness(5,2,1) witness(4,3,5) witness(3,5,2) number_of_witnesses(3)
Optimization: 3
OPTIMUM FOUND

Models       : 2
  Optimum    : yes
Optimization : 3
Calls        : 1
Time         : 0.392s (Solving: 0.18s 1st Model: 0.11s Unsat: 0.05s)
CPU Time     : 1.578s
Threads      : 8        (Winner: 4)

\end{lstlisting}
\end{footnotesize}

The IJP can then be automatically visualized as in \cref{fig:sjchainautoijp}.

\begin{figure}[h]
	\begin{subfigure}{\columnwidth}
	\centering
	\includegraphics[scale=0.4]{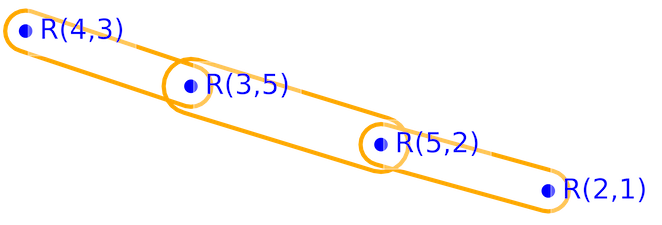}	%
	\end{subfigure}
    \caption{Automatically generated IJP for $\qsjtwochain$}
    \label{fig:sjchainautoijp}
\end{figure}

\end{document}